\documentclass[reqno]{amsart}


\usepackage{xspace}
\usepackage{subcaption}
\usepackage{todonotes}
\usepackage{mathtools}
\usepackage{eso-pic}

\usepackage{amsmath}
\usepackage{enumerate}
\usepackage{subcaption}
\usepackage{xcolor}
\usepackage{amsfonts}
\usepackage{hyperref}
\usepackage{cite}
\usepackage{booktabs}
\usepackage{multirow}
\usepackage{amssymb}
\usepackage{amsthm}
\usepackage{enumerate}
\usepackage{mathtools}
\usepackage{bm}
\usepackage{ulem}
\usepackage{fancyhdr}
\usepackage{mdframed}
\usepackage{tcolorbox}
\usepackage{graphicx,lipsum,wrapfig}
\usepackage[font={tiny}]{caption}
\usepackage[font={scriptsize}]{caption}

\renewcommand{\emph}[1]{\textit{#1}}

\theoremstyle{plain}
\newtheorem{thm}{Theorem}[section]

\newtheorem{lemma}[thm]{Lemma}
\newtheorem{prop}[thm]{Proposition}
\newtheorem*{theorem*}{Theorem}
\newtheorem{rem}[thm]{Remark}
\newtheorem{dfn}[thm]{Definition}
\newtheorem{problem}[thm]{Problem}
\newtheorem{assumption}[thm]{Assumption}

\numberwithin{equation}{section}

\newcommand{\IG}[1]{\textcolor{blue}{[IG: #1]}}

\newcommand{\reals}{\mathbb{R}}
\newcommand{\naturals}{\mathbb{N}}

\newcommand\indicator{\textnormal{\textbf{1}}}

\newcommand{\x}{\textbf{{x}}}   
\newcommand{\pu}{\textbf{u}}

\newcommand{\policy}{\sigma}
\newcommand{\optpolicy}{\sigma^*}

\newcommand{\Setofpolicies}{\Sigma}
\newcommand{\FullSet}{X}

\newcommand{\noise}{\textbf{v}}

\newcommand{\distribution}{p}

\newcommand{\mypath}{\omega}
\newcommand{\mypathX}{\mypath_{x}}

\newcommand{\PathFin}{\Omega_{x}^{\rm{fin}}}

\newcommand{\M}{\mathcal{M}}
\newcommand{\I}{\mathcal{I}}

\newcommand{\probDist}{\mathcal{D}}

\newcommand{\Amdp}{A}
\newcommand{\Qmdp}{Q}

\newcommand{\amdp}{a}
\newcommand{\qmdp}{q}

\newcommand{\Qimdp}{Q}
\newcommand{\Aimdp}{A}
\newcommand{\qimdp}{q}

\newcommand{\aimdp}{a}

\newcommand{\pathmdp}{\omega}

\newcommand{\pathmdpfin}{\pathmdp^{\mathrm{fin}}}

\newcommand{\Pathmdp}{\Omega}
\newcommand{\Pathmdpfin}{\Omega^{\mathrm{fin}}}

\newcommand{\PathInf}{\Omega_x^\infty}
\newcommand{\last}{\mathrm{last}}

\newcommand{\adversary}{\xi}
\newcommand{\Adversary}{\Xi}

\begin{document}

\AddToShipoutPictureBG*{%
  \AtPageUpperLeft{%
    \hspace{16.5cm}%
    \raisebox{-1.5cm}{%
      \makebox[0pt][r]{Submitted to \textit{Nonlinear Analysis: Hybrid Systems} (\textit{NAHS}), an IFAC journal, 2024.}}}}

\title[Efficient Strategy Synthesis for Switched Stochastic Systems]{Efficient Strategy Synthesis for Switched Stochastic Systems with Distributional Uncertainty}

\author[I.Gracia]{Ibon Gracia}

\author[D. Boskos]{Dimitris Boskos}

\author[M. Lahijanian]{Morteza Lahijanian}

\author[L. Laurenti]{Luca Laurenti}

\author[M. Mazo Jr.]{Manuel Mazo Jr.}

\thanks{The authors are  with the Smead
Department of Aerospace Engineering Sciences, CU Boulder, Boulder, CO, \tt{\{ibon.gracia, morteza.lahijanian\}@colorado.edu}, and 
the Delft Center for Systems and Control, 
Faculty of Mechanical, Maritime and Materials Engineering, Delft University of Technology \tt{\{d.boskos,l.laurenti,m.mazo\}@tudelft.nl}.}

\maketitle


\begin{abstract}
We introduce a framework for the control of discrete-time switched stochastic systems with uncertain distributions.
In particular, we consider stochastic dynamics with
additive noise whose 
distribution lies in an ambiguity set of distributions that are $\varepsilon-$close, in the Wasserstein distance sense, to a nominal one.
We propose algorithms for the efficient synthesis of distributionally robust control strategies that maximize the satisfaction probability of reach-avoid specifications with either a given or an arbitrary (not specified) time horizon, i.e., unbounded-time reachability.
The framework consists of two main steps: finite abstraction and control synthesis. First, we construct a finite abstraction of the switched stochastic system as a \emph{robust Markov decision process} (robust MDP) that encompasses both the stochasticity of the system and the uncertainty in the noise distribution. Then, we synthesize a strategy that is robust to the distributional uncertainty on the resulting robust MDP. We employ techniques from optimal transport and stochastic programming to 
reduce the strategy synthesis problem to a set of linear programs, and propose a tailored and efficient algorithm to solve them. The resulting strategies are correctly refined into switching strategies for the original stochastic system. We illustrate the efficacy of our framework on various case studies comprising both linear and non-linear switched stochastic systems.

\end{abstract}


\section{Introduction}
\label{sec:intro}


Switched stochastic systems are a ubiquitous class of control systems due to their capability to capture digitally controlled physical systems affected by noise. In this class,
a controller can switch among a finite set of modes to achieve the desired behavior \cite{yin2010hybrid}.  Systems that 
that be represented with
such models are found in various real-world applications, including robotics \cite{luna2015asymptotically} and cyber-physical systems \cite{cauchi2019efficiency}. These systems are often \emph{safety-critical}, requiring formal guarantees of correctness, and their noise characteristics are \emph{uncertain}, with statistical properties only available through the use of statistical estimation techniques and possibly subject to \emph{distributional shifts} \cite{rahimian2019distributionally}. 
However,  formal control synthesis and verification methods for switched stochastic systems typically assume exact and known noise characteristics. 
We aim to address this fundamental gap in the literature and focus on the following question: \emph{how to derive formal guarantees for stochastic systems with uncertain noise characteristics?}


In this work, we introduce a formal framework for synthesis of distributionally robust control strategies for discrete-time switched stochastic systems with uncertain noise distribution under reach-avoid specifications.
More precisely, we consider switched stochastic systems with 
general dynamics (possibly non-linear) with additive noise in each mode, and while the exact probability distribution of noise is unknown, we assume it belongs to a prescribed Wasserstein ambiguity set, i.e., the set of possible
distributions defined as a ball of radius $\varepsilon>0$ around a nominal distribution under the  Wasserstein distance~\cite{mohajerin2018data,gao2022distributionally}. 
For instance,
such a set can be estimated using data-driven techniques, e.g., see \cite{mohajerin2018data}. 
The probabilistic reach-avoid specification 
can include an arbitrary or predefined time horizon to reach a target set.  
Such specifications aim to establish a lower bound on the probability of reaching a goal region while avoiding undesired configurations, and
are the basis of logical properties such as syntactically-cosafe Linear Temporal Logic (sc-LTL)~\cite{kupferman:FMSD:2001}, LTL over finite traces (LTLf)~\cite{de2013linear},  and Probabilistic Computation Tree Logic (PCTL)~\cite{baier2003model}. 


The framework consists of two main steps: (i) construction of a finite abstraction of the stochastic system in the form of an uncertain Markov decision process (MDP) \cite{nilim2005robust,iyengar2005robust}, also known as \emph{robust MDP}, and (ii) optimal robust strategy synthesis on the robust MDP.
For the abstract construction, we first discretize the state space, and then use the nominal noise distribution (center of the ambiguity set) to compute bounds on the transition probabilities between the discrete regions.  This results in a nominal Interval MDP (IMDP)~\cite{givan2000bounded,lahijanian2015formal} that contains the discretization error but not the distributional uncertainty.  We then lift the ambiguity set from the noise space to the space of transition probabilities of the IMDP, and embed it as uncertainty, resulting in a robust MDP.  Hence, the obtained robust MDP accounts for the distributional ambiguity in the original system as well as the discretization error.  

Next, we synthesize an optimal strategy on the robust MDP that is robust with respect to the uncertain transition probabilities of this model. To that end, we first
introduce a value iteration method to maximize the robust (worst-case) probability of satisfying the reach-avoid specification on the robust MDP.  We prove that, for the unbounded-time specifications, the value iteration converges to the maximum worst-case probability of satisfaction.  
We then leverage recent results from distributionally robust optimization~\cite{rahimian2019distributionally} to reduce value iteration to the solution of a set of linear programs, improving efficiency. 
Further, 
for unbounded-time properties, we show that stationary strategies are sufficient to obtain the optimal worst-case probability, based on which we introduce an algorithm for the extraction of such an optimal and robust strategy.  
We also formally prove the correctness of our framework for the original switched stochastic system.  Finally, we illustrate the framework on four case studies including both linear and non-linear systems and for various ambiguity sets. 

In short, the contributions of this work is fivefold:
\begin{itemize}
    \item a framework for formal reasoning for distributionally-uncertain switched stochastic systems,
    \item an abstraction procedure to construct a finite robust MDP for nonlinear switched stochastic systems with additive uncertain noise,
    \item a robust strategy synthesis algorithm for robust MDPs with unbounded-time reachability specifications,
    \item an efficient robust optimization method via a dual formulation of linear programs, and
    \item a set of illustrative case studies and benchmarks for empirical evaluation of the framework. 
\end{itemize}

\subsection{Previous Work}

A preliminary version of this work appeared in \cite{gracia2023distributionally}, which the current manuscript extends in several aspects. First, in this paper, we present a new algorithmic approach to perform value iteration, which substantially improves the one in \cite{gracia2023distributionally} in terms of computational efficiency: an improvement of two orders of magnitude is observed empirically. Furthermore, work \cite{gracia2023distributionally} only considers finite-time reachability properties, whereas in this paper, we extend our results to unbounded-time reachability properties and prove convergence of value iteration and optimality of stationary strategies in this setting. We also introduce a method of extracting such a strategy, which is often ignored in the literature.
Finally, we extend the experimental evaluations with a system subject to distributional shifts and an example of a system under a complex specification given as a temporal logic formula. 
With the exception of Sections~\ref{sec:unbounded_horizon},\ref{sec:DedicatedAlgorithmValueITeration} and the corresponding Appendices (\ref{app:proofs_section_unbounded_horizon} and \ref{app:proof_dual}), all the theorems and definitions remain the same as in \cite{gracia2023distributionally}. Section \ref{sec:LinearProgrammingValueITeration} has been revised to account for the case of disturbances with a bounded support. The new case studies are presented in Sections~\ref{sec:switched} and \ref{sec:LTLf}. We also benchmark the new algorithm 
in Section~\ref{sec:DedicatedAlgorithmValueITeration} 
against the one in~\cite{gracia2023distributionally} to illustrate its computational superiority.

\subsection{Related Work}
Formal verification and synthesis algorithms have been studied for swit-\newline ched stochastic systems either employing stochastic barrier functions \cite{santoyo2021barrier} or abstractions in the form of \emph{finite} Markov models \cite{Lahijanian:CDC:2012,lavaei2022dissipativity,dutreix2022abstraction,wolff2012robust,skovbekk2023formal}. Within the latter, many works 
(e.g., \cite{lahijanian2015formal,cauchi2019efficiency, laurenti2020formal,jackson2020safety,jackson2021formal})
have employed \emph{interval Markov decision processes} (IMDPs) as abstract model of choice, a class of Markov decision processes in which the transition probabilities belong to intervals \cite{givan2000bounded,koutsoukos2006computational} and that admits efficient control synthesis algorithms \cite{lahijanian2015formal,cauchi2019efficiency}. 
In the literature it is a common assumption that both the dynamics and noise distribution of the system are assumed to be exactly known. This in practice is an unrealistic assumption and often violated in the presence of, e.g., unmodelled dynamics, distributional shifts, or data-driven components. 

More recently, researchers have proposed the use of machine learning algorithms, including neural networks and Gaussian processes, to synthesise formal control strategies when the dynamics are (partially) unknown or too complex to be modelled \cite{adams2022formal,jackson2021strategy, reed2023promises}. Yet, these works focus on the unknown dynamics and do not address the case of inexact knowledge of the probability distribution of stochastic variables in the dynamics. 
Another line of research tackles the opposite setting, in which the dynamics are known exactly but the probabilistic behavior of the system must be inferred from samples of the random terms \cite{badings2022sampling, mathiesen2023inner}. Our setting is related to the latter in the sense that we consider ambiguity regarding the distribution of the random terms. However, our approach is general with respect to the source of ambiguity, which encompasses the data-driven setting and those of distributionally robust control and robustness against distributional shifts.



Ambiguity sets are typically employed in distributionally robust optimization (DRO) to describe sets of probability distributions with respect to which one aims to make robust (optimal) decisions \cite{AS:17}. A typical way of defining an ambiguity set is as a set of probability distributions that are close to some nominal one. Depending on the employed measure of closeness, ambiguity sets can be constructed based on: moment constraints \cite{ED-YY:10,IP:07}, statistical divergences \cite{GCC-LEG:06}, or optimal transport discrepancies \cite{gao2022distributionally,JB-KM:19,JB-KM-FZ:22} like the Wasserstein distance. Wasserstein ambiguity sets, as considered in this paper, constitute a particularly convenient description for ambiguous distributions in data-driven problems. 
The Wasserstein metric penalizes horizontal dislocations between distributions \cite{FS:15}, providing a way to describe ambiguity sets that have finite-sample guarantees of containing the true distribution~\cite{fournier2015rate}. Furthermore, Wasserstein ambiguity sets enable the formulation of tractable DRO problems~\cite{mohajerin2018data}. In \cite{boskos2023high} dynamic aspects of distributional uncertainty under optimal transport ambiguity are studied. In particular, \cite{boskos2023high} studies the evolution of Wasserstein ambiguity sets for systems with unknown state disturbance distribution. Also, \cite{AH-IY:21} developed a risk-aware robot control scheme avoiding dynamic obstacles whose location is dictated by an ambiguous probability distribution.

A class of Markov processes that is closely related to robust MDPs are distributionally robust Markov decision processes (DR-MDPs) \cite{xu2010distributionally,yang2017convex,JGC-CK:21}. These are MDPs with transition probabilities dependent on uncertain parameters that lie in some ambiguity set. 
%
DR-MDPs and robust MDPs are substantially different as the latter do not consider any additional probabilistic structure over the ambiguous distributions to signify which uncertainty model is more likely to occur. 
Planning algorithms for complex specifications and diverse classes of robust Markov models have been already considered in the literature \cite{lahijanian2015formal,wolff2012robust,nilim2005robust,puggelli2013polynomial}.
However, none of them addresses the main contribution of this work, which is how to combine these algorithms with optimal transport techniques to formally abstract and synthesize strategies for continuous-space dynamical systems with uncertain distribution of the noise. 
Moreover, the robust MDPs over which we synthesize strategies in this paper generalize the ones considered in \cite{wolff2012robust, puggelli2013polynomial}. These works assume that the transition probabilities of their robust MDPs vanish at the exact same states for all transition matrices. Here, we alleviate this restriction, which provides us substantial modeling flexibility as we can handle for instance data-driven ambiguity sets whose reference distribution, obtained from samples, is bounded even when the data-generating distribution may not be. This leads to the final key contribution of our paper, which is the rigorous synthesis of optimal strategies for infinite horizon reachability specifications for a broader class of robust MDPs. In addition, unlike \cite{wolff2012robust, puggelli2013polynomial}, our approach does not require identifying and eliminating end components as a pre-processing step, which suffers from exponential complexity for general robust MDPs \cite{weininger2019satisfiability}.

\section{Basic Notation}
\label{sec:basic_notation}
The set of non-negative integers is denoted by $\mathbb{N}_0 := \mathbb{N}\cup\{0\}$. Given a set $A$,  $|A|$ denotes its cardinality. Given $\ell,m\in\mathbb{N}_0$ with $\ell\leq m$, we employ $[\ell:m]$ for the set $\{\ell,\ell+1,\ldots,m\}$. We also use this notation when $m = \infty$ to denote the set $\{\ell,\ell+1,\ldots\}$. For a separable metric space $X$, $\mathcal{B}(X)$ denotes its Borel $\sigma$-algebra and $\mathcal{D}(X)$ the set of probability distributions on $(X,\mathcal{B}(X))$. When $X$ is discrete and $\gamma\in\mathcal D(X)$, $\gamma(x):=\gamma(\{x\})$ is the probability of the event described by the singleton $\{x\}$. Let $c:X\times X\rightarrow \reals_{\ge 0}$ be a continuous cost function defined over the product space $X\times X$. The optimal transport discrepancy between two probability distributions $p,p'\in\mathcal{D}(X)$ is defined as 
\begin{align}
\label{eq:optimal_transport_discrepancy}
\mathcal{T}_c(p,p') := \inf_{\pi\in\Pi(p,p')} \int_{X\times X}c(x,y)d\pi(x,y),
\end{align}
where $\Pi(p,p')$ is the set of all transport plans between $p$ and $p'$, a.k.a. couplings,  i.e., probability distributions $\pi\in\mathcal{D}(X\times X)$ with marginals $p$ and $p'$, respectively. Since cost $c$ is non-negative, $\mathcal{T}_c$ provides a discrepancy measure between distributions in $\mathcal{D}(X)$. By continuity of $c$, there always exists a transport plan $\pi$ for which the infimum in \eqref{eq:optimal_transport_discrepancy} is attained~\cite[Theorem 1.3]{CV:03}. 

Assume that $X$ is equipped with a metric $d$.  Given $s\ge 1$, we denote by $\mathcal{D}_s(X)$ the set of probability distributions on $X$ with finite $s$-th moment, i.e., $$\mathcal{D}_s(X) = \left\{p\in\mathcal{D}(X) : \int_X d(x,y)^s dp(x) < \infty\:\text{for some}\:y\in X \right \}.$$ The discrepancy $\mathcal W_s:=(\mathcal T_{d^s})^{\frac{1}{s}}$ is then also a metric in the space $\mathcal{D}_s(X)$ coined as the $s$-Wasserstein distance~\cite{CV:03}. Additionally, we use bold symbols, e.g., $\noise$, to denote random variables and non-bold symbols to denote the respective realizations, i.e., $\text{v}$.

\section{Problem Formulation}
\label{sec:problem}

Consider a discrete-time switched stochastic system given by:
\begin{equation}
\label{eq:system}
    \x_{k+1} = f_{u_k}(\x_k) + \noise_k, 
\end{equation}
where $k\in \naturals,$ $\x_k$ takes values in $\reals^n$,  $u_k \in U$, and $U=\{1,...,m \}$ is a finite set of \textit{modes} or \textit{actions}. For every $u\in U,$ $f_u:\mathbb{R}^n \to \mathbb{R}^n$ is a possibly non-linear continuous function.
The noise term $\noise_k$ is an independent random variable takes values in $\reals^n$ with a distribution $\distribution_v^{\text{true}}$ that is identically distributed at each time step. 
We assume the exact noise distribution is unknown, but it is $\varepsilon$-close to a given (nominal) distribution as detailed below.
\begin{assumption}
   \label{ass:1}
The distribution $\distribution_v^{\text{true}}$ is $\varepsilon$-close (in the $s$-Wasserstein sense) to a known distribution $\widehat{\distribution}_{v}\in\mathcal{D}_s(\reals^n)$, which we call \emph{nominal}, i.e., $\distribution_v^{\text{true}}\in\mathcal{P}_v := \{p\in \mathcal{D}_s(\reals^n) : \mathcal{W}_s(p,\widehat p_v)\le \varepsilon\}$, where $\mathcal{W}_s$ is determined by the metric $d(x,y) = \|x-y\|$, where $\|\cdot \|$ is the Euclidean norm. 
\end{assumption}


Intuitively, $\x_k$ is a stochastic process driven by the additive noise
$\noise_k$, whose distribution is uncertain but close to a nominal known distribution in the $s$-Wasserstein metric. Consequently, System~\eqref{eq:system} represents a large class of controlled stochastic systems with additive and uncertain noise. 
Such systems arise, for instance, in data-driven settings, where measure concentration results \cite{fournier2015rate} can be employed to  build a Wasserstein ambiguity set from data of $\noise_k$ with high confidence 
\cite{mohajerin2018data}, or in settings in which one wants to synthesize strategies able to cope with distributional 
shifts of the noise.


Let $\mypathX =x_0 \xrightarrow{u_0} x_1 \xrightarrow{u_1}  \ldots $ be a  
path (trajectory realization) of System \eqref{eq:system}  and denote by $\mypathX(k)=x_{k}$ 
the visited state at time $k$ in the sequence $\mypathX$. 
Further, we denote by 
$\mypathX^k$ the prefix of finite length $k+1$ of $\mypathX$, and by $\last(\mypathX^k)$ the last state of the finite path $\mypathX^k$, i.e., $x_k$.  
We also denote by $\PathInf$ the set of paths of infinite length and by $\PathFin$ the set of all paths with finite length, 
i.e, the set of prefixes  $\mypathX^{k}= x_0 \xrightarrow{u_0} x_1 
\xrightarrow{u_1}  \ldots \xrightarrow{u_{k-1}} x_k$ for all $k\in 
\naturals$.
Given a finite path, a switching strategy selects an action, corresponding to a mode of 
System~\eqref{eq:system}.

\begin{dfn}[Switching Strategy]
    \label{def:switching-strategy}
    A \emph{switching strategy} $\policy_x: \PathFin \to U$ is a function that maps 
    each finite path $\mypathX^{k}\in \PathFin$ to an action $u \in U$. 
    We say that a strategy is \emph{memoryless (Markovian)} if it only depends on $\last(\mypathX^k)$ and $k$. Furthermore, if a memoryless strategy does not depend on $k$, we call it \emph{stationary}.
\end{dfn}
\noindent


For any $\distribution_{v}\in \mathcal{P}_{v},$  $u\in U$, $\FullSet\in \mathcal B(\reals^n)$, and $x\in \reals^n$, let 
\begin{align}\label{transition:kernel}
T^{u}_{\distribution_{v}}(\FullSet \mid x)=\int \indicator_\FullSet(f_u(x)+\bar v)\distribution_{v}(\bar v)d \bar v
\end{align}
be the stochastic transition kernel induced by system \eqref{eq:system} with 
noise fixed to $\distribution_{v}$ in mode $u\in U$, where 
$\indicator_\FullSet$ is the indicator function with 
$\indicator_\FullSet(x)=1$ if $x\in \FullSet$, and $\indicator_\FullSet(x)=0$ 
otherwise.
From the definition of $T^{u}_{\distribution_{v}}(\FullSet\mid x)$  it follows 
that, given a strategy $\policy_x$, a noise distribution $\distribution_{v}$, and
an initial condition $x_0$, System 
\eqref{eq:system} defines a stochastic process on the canonical space 
$(\mathbb{R}^n)^{\mathbb N}$ with the Borel sigma-algebra 
$\mathcal{B}((\mathbb{R}^n)^{\mathbb N})$ 
\cite{bertsekas2004stochastic}.
%
%
In particular, by the Ionescu-Tulcea theorem, the kernel $T^{u}_{\distribution_{v}}$ generates a unique probability distribution $P^{x_0,\policy_x}_{\distribution_{v}}$ over the paths of system \eqref{eq:system} originating from $x_0$,
such that for each $k\in \mathbb N$ 
\begin{align*}
    &P^{x_0,\policy_x}_{\distribution_{v}}[\mypathX(0)\in \FullSet] = 
    \indicator_{\FullSet}(x_0),\\
     &P^{x_0,\policy_x}_{\distribution_{v}}[\mypathX(k) \in \FullSet \mid 
    \mypathX^{k-1}] = 
    T^{\policy_x(\mypathX^{k-1})}_{\distribution_{v}}(\FullSet \mid \mypathX(k-1)).     
\end{align*}
%


In this paper we consider both finite-time and infinite-time
probabilistic reach-avoid specifications for System \eqref{eq:system}. 
\begin{dfn}[Reach-avoid probability]
For a time horizon $K\in\mathbb{N}_0\cup\{\infty\}$, a bounded safe set $X$, a target region 
$X_{\text{tgt}}\subset X$ and an initial state $x_0\in X$, the reach-avoid probability $P_{\rm{reach}}( X, X_{\text{tgt}},K \mid x_0, \policy_x,\distribution_v)$ is defined as:
\begin{multline}
    P_{\rm{reach}}( X, X_{\text{tgt}},K \mid x_0, \policy_x,\distribution_v) := \\ 
     P_{\distribution_{v}}^{x_0,\policy_x} \left[\exists k \in [0:K] 
     \mid 
     \mypath_\x(k)\in X_{\text{tgt}} \; \land \; \forall \,k'<k, \; \mypath_\x(k')\in X \right]. 
    \label{eq:reachability_original_system}
\end{multline}
\end{dfn}

%
We can now precisely formulate our problem, which is to synthesize control strategies that are robust to all distributions 
in the set $\mathcal{P}_{v}$.

\begin{problem}[Switching Strategy Synthesis]
 \label{prob:Syntesis}
    Consider the switched stochastic system \eqref{eq:system}, its 
    corresponding ambiguity set $\mathcal{P}_{v}$, a bounded safe set $X$, and 
    a target region $X_{\text{tgt}}\subset X$. Given an initial state $x_0\in 
    X$, a probability threshold $p_{\text{th}}\in[0,1]$, and a horizon 
    $K\in\mathbb{N}_0\cup\{\infty\}$, synthesize a switching strategy $\policy_x$ 
    such that, for all $\distribution_v\in\mathcal{P}_v$,
\begin{align}
\label{eq:problem}
 & P_{\rm{reach}}( X, X_{\rm{tgt}},K \mid x_0, \policy_x,\distribution_v)  \geq 
 p_{\rm{th}}.
\end{align}
\end{problem}

\begin{rem}
Note that, while our focus on reach-avoid specifications in Problem \ref{prob:Syntesis}, the proposed framework can be easily used for more complex logical specifications, such as scLTL, LTLf, and PCTL specification, which reduce to reachability problems.
We illustrate this in an example in Section~\ref{sec:LTLf}, where we synthesize a strategy for an LTLf specification.
\end{rem}

\subsection{Overview of the Approach}
To approach Problem \ref{prob:Syntesis}, we construct a finite-state abstraction of System \eqref{eq:system} in terms of a robust MDP as detailed in Section \ref{sec:robust_MDP_abstraction}. In Section \ref{sec:synthesis}, we show how to synthesize an optimal strategy on the resulting abstraction. In particular, we prove that, while in the case of finite-time reachability specifications the optimal strategy is Markovian, the optimal strategy for an unbounded-time (infinite horizon)  reachability property is stationary. We prove that in both cases an optimal strategy can be synthesized by solving a finite set of linear programs, thus guaranteeing efficiency. We then refine the resulting strategy into a switching strategy for System \eqref{eq:system} and show how the robust MDP abstraction provides upper and lower bounds on the probability that  System \eqref{eq:system} satisfies the specification under the refined strategy. 





\section{Preliminaries}


\subsection{Robust Markov Decision Processes}
Robust MDPs are a generalization of Markov decision 
processes in which the transition probability distributions between states are 
constrained to belong to an ambiguity set \cite{nilim2005robust}, 
\cite{wiesemann2013robust}. 



\begin{dfn}[Robust MDP]
\label{def:robust_mdp}
    A \emph{robust Markov decision process} (robust MDP) is a tuple $\M = (\Qimdp,\Aimdp,\Gamma)$, where
    \begin{itemize}
    	\setlength\itemsep{1mm}
    	\item $\Qimdp$ is a finite set of states,
    	\item $\Aimdp$ is a finite set of actions, and $\Aimdp(\qimdp)$ denotes 
    	the set of available actions at state $\qimdp \in \Qimdp$,
        \item $\Gamma = \{\Gamma_{\qimdp,a}\}_{\qimdp\in\Qimdp, a\in\Aimdp}$ 
        is the set of possible  transition probability distributions of $\M$, where 
        $\Gamma_{\qimdp,a}\subseteq\mathcal{D}(\Qimdp)$ is the set of transition probability distributions for state-action pair $(q,a) \in Q\times A(q)$.\footnote{Note that the sets of transition probability distributions of the robust MDP 
        are independent for each state and action. This is known as 
        \textit{rectangular} property of the set of transition probability 
        distributions \cite{nilim2005robust, wiesemann2013robust}. Furthermore, unlike \cite{wolff2012robust, puggelli2013polynomial}, this is the only assumption we impose regarding the structure of the robust MDP.}
    \end{itemize}
\end{dfn}

A path of a robust MDP is a sequence of states $\pathmdp = \qmdp_0 
\xrightarrow{\amdp_0} \qmdp_1 \xrightarrow{\amdp_1} \qmdp_2 
\xrightarrow{\amdp_2}  \ldots$ such that $\amdp_k \in \Amdp(\qmdp_k)$ and there 
exists $\gamma\in\Gamma_{\qmdp_k,a_k}$ with $ \gamma(\qmdp_{k+1})> 0$ for all 
$k \in \naturals$. We denote 
the $i$-th state of a path $\pathmdp$ by $\pathmdp(i)$, a finite path of length $k+1$ by $\pathmdp^k$ and the last state of a finite path $\pathmdpfin$ by $\last(\pathmdpfin)$. The set of all paths of finite and infinite length are denoted by $\Pathmdpfin$ and $\Pathmdp$, respectively. 

\begin{dfn}[IMDP]
\label{def:imdp}
    An \emph{interval MDP} (IMDP) \cite{cauchi2019efficiency}, \cite{laurenti2020formal}, also known as bounded-parameter MDP (BMDP) \cite{givan2000bounded}, \cite{koutsoukos2006computational}, is a class of robust MDP $\I = (\Qimdp, \Aimdp,\Gamma)$ where $\Gamma$ has the form
\begin{align}
\label{eq:set_transition_probabilities_IMDP}
    \Gamma_{q,a} = \{ \gamma\in\mathcal{D(\Qimdp)}: \underline P(\qimdp,a,\qimdp')\leq \gamma(\qimdp')\leq \overline P(\qimdp,a,\qimdp')\:\text{for all}\:\qimdp'\in\Qimdp\},
\end{align}
for every $\qmdp\in\Qmdp$, $a\in\Amdp(q)$. The bounds $\underline P, \overline P$ are called transition probability bounds, for which it must hold that, for every state $\qimdp\in\Qimdp$ and action $a\in \Amdp(q)$, $0\leq\underline{P}(\qmdp,a,\qmdp')\leq\overline{P}(\qmdp,a,\qmdp')\leq 1$ and
$
    \sum_{q'\in Q}\underline{P}(\qmdp,a,\qmdp') \leq 1 \leq \sum_{\qmdp'\in Q}\overline{P}(\qmdp,a,\qmdp').
$  
\end{dfn}
The actions of robust MDPs and IMDPs are chosen according to a strategy $\policy$ which is defined below. 
\begin{dfn}[Strategy]
\label{def:strategy}
    A strategy $\policy$ of a robust MDP model $\M$ is
    a function $\policy: \Pathmdpfin \rightarrow \Aimdp$ that maps a finite path $\pathmdp^k$, with $k\in\mathbb{N}$, of $\M$ onto an action in $\Aimdp(\last(\pathmdpfin))$. If a strategy depends only on $\last(\pathmdpfin)$ and $k$, it is called a memoryless (Markovian) strategy, and if it only depends on $\last(\pathmdpfin)$, it is called stationary. The set of all strategies is denoted by $\Setofpolicies$, and the set of all stationary strategies by $\Setofpolicies_s$.
\end{dfn}
Given an arbitrary strategy $\policy$, 
we are restricted to the set of robust Markov chains defined by the set of 
transition probability distributions induced by $\policy$. To reduce 
this to a Markov chain, we define the adversary~\cite{givan2000bounded} 
(called ``nature" in \cite{nilim2005robust}), to be the function that assigns a transition 
probability distribution to each state-action pair.
\begin{dfn}[Adversary]
\label{def:adversary}
    For a robust MDP $\M$, an adversary is a function $\adversary: \Pathmdpfin \times \Aimdp \rightarrow \probDist(\Qimdp)$ that, for each finite path $\pathmdpfin \in \Pathmdpfin$, state $\qimdp=\last(\pathmdpfin)$, and action $a \in \Aimdp(\qimdp)$, assigns an admissible distribution $\gamma_{\qmdp,a}\in\Gamma_{\qmdp,a}$. The set of all adversaries is denoted by $\Adversary$.
\end{dfn}

For an initial condition $q_0\in Q$, under a strategy and a valid adversary $\xi\in\Xi$, the robust MDP collapses to a Markov chain with a unique probability measure defined over its paths. With a small abuse of notation, we denote this measure by $P_\xi^{q_0,\policy}$.


\section{Robust MDP Abstraction}
\label{sec:robust_MDP_abstraction}
To approach Problem \ref{prob:Syntesis}, we start by abstracting System 
\eqref{eq:system} into the IMDP $\widehat\I = (\Qimdp,\Aimdp,\widehat\Gamma)$ 
with the noise distribution fixed to the nominal one, 
$\widehat\distribution_v$. In this way, we embed the error caused by the state 
discretization into $\widehat\I$. After that, we expand the set of transition 
probabilities $\widehat\Gamma$ of $\widehat\I$ to also capture the 
distributional ambiguity into the abstraction, obtaining the robust MDP $\M = 
(\Qimdp,\Aimdp,\Gamma)$. Note that the sets of states $\Qimdp$ and actions 
$\Aimdp$ are the same in $\widehat\I$ and $\M.$
 Below, we describe how we obtain $\Qimdp$ and $\Aimdp$, and in Section 
 \ref{Sec:set_transition_probabilities} we consider the set of transition 
 probability distributions $\Gamma$.

\subsection{States and Actions}
\label{sec:states_actions}
The state space $\Qimdp$ of $\M$ is constructed as follows: consider a set of 
non-overlapping regions $Q_{\rm{safe}} = 
\{\qimdp_1,\qimdp_2,\dots,q_{|Q_{\rm{safe}}|}\}$ partitioning the set $X$ so 
that either $q\cap X_{\rm tgt}=\emptyset$ or $q\cap(X\setminus X_{\rm 
tgt})=\emptyset$ for all $q\in Q_{\rm{safe}}$. We denote by $Q_{\rm tgt}$ the 
subset of $Q_{\rm{safe}}$ for which $q\cap X_{\rm{tgt}}= q $ and assume 
that it is a partition of $Q_{\rm tgt}$. The states of the abstraction comprise 
of  $Q_{\rm{safe}}$ and the unsafe region $\qmdp_u:=\reals^n\setminus X$, 
hence, $\Qimdp := Q_{\rm{safe}} \cup \{\qimdp_u\}$. We index $\Qimdp $ by 
$\mathcal N=\{1,\ldots,N\}$, where $N:=|Q|$ and denote the actions of the 
abstraction as $\Aimdp := U$.

\subsection{Transition Probability Distributions}
\label{Sec:set_transition_probabilities}
\subsubsection{Accounting for the Discretization Error}
\label{sec:nominal_IMDP}
To capture the state discretization error in the abstraction, we first 
consider an IMDP abstraction of System~\eqref{eq:system} for a fixed 
distribution, namely, the nominal probability distribution $\widehat\distribution_v$. 
Since this IMDP is constructed for the nominal distribution 
$\widehat\distribution_v$, we call it ``nominal" IMDP and use the notation 
$\widehat\I = (\Qmdp,\Amdp,\widehat\Gamma)$. We note that building an IMDP 
abstraction of a stochastic system with disturbances of a known distribution 
is widely studied in the literature~\cite{lahijanian2015formal,dutreix2020specification,cauchi2019efficiency,adams2022formal}, and we report the full procedure here below for completeness.  
Construction of the state $\Qmdp$ and action $\Amdp$ spaces of $\widehat\I$ are described in Section \ref{sec:states_actions}. 
We now describe the set $\widehat\Gamma$ of 
$\widehat\I$. 
According to Definition \ref{def:imdp}, the set $\widehat\Gamma$ 
is defined by the transition probability bounds $\underline P$ and $\overline 
P$ of $\widehat\I$. To formally account for the discretization error, the 
bounds must satisfy 
for all $\qimdp \in Q_{\rm{safe}}$, $\qimdp'\in Q$ and $\aimdp \in \Aimdp = U$:
\begin{equation}
	\label{eqn:transition_probability_bounds1}
        \underline P(\qimdp,\aimdp,\qimdp') \leq \min_{x\in \qimdp}  T_{\widehat\distribution_v}^{\aimdp}(\qimdp'\mid x)  
        \quad 
        \text{and}
        \quad
        \overline P(\qimdp,\aimdp,\qimdp')  \geq \max_{x\in \qimdp}  T_{\widehat\distribution_v}^{\aimdp}(\qimdp'\mid x).
\end{equation}

Since we are interested in the paths of system \eqref{eq:system} that do not 
exit set $X$, we make state $\qimdp_u$ absorbing, i.e.,
\begin{align}
\label{eq:transition_probability_bounds3}
    \underline P(\qimdp_u,a,\qimdp_u) = \overline P(\qimdp_u,a,\qimdp_u) = 1
\end{align}
for all $a\in\Aimdp$. In this way, we include the ``avoid" part of the specification into the definition of the abstraction: the paths that reach $\qmdp_u$, will remain there forever and, therefore, will never reach the target set $Q_{\rm{tgt}}$.
Consequently, for each $\qmdp\in\Qmdp$ and $a\in\Amdp$, we obtain 
\begin{align}
\label{eq:set_transition_probabilities_nominal_IMDP}
    \widehat\Gamma_{q,a} = \{ \gamma\in\mathcal{D}(\Qimdp): \underline P(\qimdp,a,\qimdp')\leq \gamma(\qimdp')\leq \overline P(\qimdp,a,\qimdp')\:\text{for all}\:\qimdp'\in\Qimdp\}.
\end{align}

\subsubsection{Accounting for the Distributional Uncertainty}
Now, we expand the sets $\{\widehat\Gamma_{\qimdp,a}\}_{\qimdp\in\Qimdp,a\in\Aimdp}$ of transition probabilities of $\widehat\I$ to also embed the distributional uncertainty into the abstraction. With this objective, we first define the cost 
\begin{align}
    \label{eq:cost}
    c(q,q') :=\inf\{\|x - y\|^s:x\in q, y\in q'\},
\end{align}
between any two states $\qmdp,\qmdp'\in Q$, where $\|\cdot\|$ and $s$ are the same as for $\mathcal{W}_s$ in Assumption~\ref{ass:1}.
The cost $c(q,q')^{\frac{1}{s}}$ is the minimum distance, in the sense of norm $\|\cdot\|$, between any pair of points in the regions $\qimdp$ and $\qimdp'$, respectively. Using this cost and the exponent  $s$ in $\mathcal{W}_s$, we define the optimal transport discrepancy $\mathcal{T}_c$ between distributions over $\Qimdp$ as in \eqref{eq:optimal_transport_discrepancy}.\footnote{Notice that, since $c$ is also not a metric, the resulting optimal transport discrepancy $\mathcal{T}_c$ is not a distance.} Given  a probability distribution $\gamma\in\mathcal{D}(Q)$ and $\epsilon\ge 0$, we denote by $\mathcal T_c^{\epsilon}(\gamma)$ 
the set of all distributions to which mass can be transported from $\gamma$ incurring a $c$-transport cost lower than $\epsilon$. 
Using the previous elements, we are finally able to define $\Gamma$.

\begin{dfn}
The discrete uncertainty set $\Gamma$ is defined for every state $\qimdp\in Q_{\rm{safe}}$ and action $a\in\Aimdp$ as
\begin{align}
    \label{eq:set_transition_probabilities_robust_MDP}
    \Gamma_{q,a}:= & \bigcup_{\gamma\in\widehat\Gamma_{q,a}}\mathcal T_c^{\epsilon}(\gamma),
\end{align}
with $\epsilon = \varepsilon^s$.
For state $\qmdp_u$ and action $a\in\Aimdp$, let $\Gamma_{q,a} := \widehat\Gamma_{q,a}$ (preserving the absorbing property of the unsafe state).
\end{dfn}

Each $\Gamma_{\qimdp,a}$ is the set of probability distributions over $\Qimdp$ that are $\epsilon$-close to $\widehat\Gamma_{\qimdp,a}$, in the sense of the optimal transport discrepancy $\mathcal{T}_c$. Note that the way in which we have defined $\Gamma$ makes it possible that there exist two probability distributions $\gamma,\gamma' \in \Gamma_{q,a}$, for some $q\in Q$, $a\in A$, which have different support. Therefore, the robust MDPs we consider belong to a more general class than the ones in \cite{wolff2012robust, puggelli2013polynomial}.
%

Once we have obtained the sets of transition probabilities $\Gamma$, along with 
the state $\Qimdp$ and action $\Aimdp$ spaces, our robust MDP abstraction $\M$ 
is fully defined. The following proposition ensures that the abstraction 
captures all possible transition probabilities of System~\eqref{eq:system} to 
regions in the partition.
\begin{prop}[Consistency of the Robust MDP Abstraction]
\label{prop:consistency}
Consider the robust MDP abstraction $\M = (\Qmdp,\Amdp,\Gamma)$ of System~\eqref{eq:system}. Let $\qmdp\in Q_{\rm{safe}}$, $a\in \Amdp$,  $x\in q$,  and
$\distribution_{v}\in\mathcal{P}_v$, and define $\gamma_{x,a}\in\mathcal{D}(Q)$ as
\begin{align*}
    \gamma_{x,a}(q') &:= T^{a}_{\distribution_{v}}(q'|x)
\end{align*}
for all $\qmdp'\in \Qmdp$. Then  $\gamma_{x,a}\in\Gamma_{\qmdp,a}$.
\end{prop}
%
\noindent The proof of Proposition \ref{prop:consistency} is given in \ref{app:proof_consistency}.
\noindent The intuition behind Proposition \ref{prop:consistency} is that set $\Gamma_{q,a}$ contains the transition probabilities  $\gamma_{x,a}$ obtained by starting from any $x\in \qmdp$, with $\qmdp\in Q_{\rm{safe}}$ under $a\in \Amdp$ and for every $\distribution_{v}\in\mathcal{P}_v$.
\begin{rem}[Model Choice for the Abstraction]\label{rem:model:choice}
An alternative way to include the distributional ambiguity into the abstraction is to use an IMDP abstraction $\I = (\Qimdp,\Aimdp,\Gamma^{\rm IMDP})$, which has the same state $\Qimdp$ and action $\Aimdp$ spaces as $\M$, and in which $\Gamma^{\rm IMDP}$ is defined by transition probability bounds that satisfy
\begin{equation}
	\label{eqn:transition_probability_bounds_alternative1}
    \begin{split}
        \underline P^{\rm IMDP}(\qimdp,\aimdp,\qimdp') & \leq \min_{\distribution_v\in\mathcal{P}_v}\min_{x\in \qimdp}  T_{\distribution_v}^{\aimdp}(\qimdp'\mid x),  \\
        \overline P^{\rm IMDP}(\qimdp,\aimdp,\qimdp') & \geq \max_{\distribution_v\in\mathcal{P}_v}\max_{x\in \qimdp}  T_{\distribution_v}^{\aimdp}(\qimdp'\mid x)
    \end{split}
\end{equation}
%
for all $\qmdp\in Q_{\rm{safe}}$, $\qmdp'\in\Qmdp$ $a\in\Aimdp$, and for which $\qmdp_u$ is again absorbing. Therefore, for fixed $\qmdp\in\Qmdp$ and $a\in\Amdp$, set $\Gamma_{\qmdp,a}^{\rm IMDP}$ of $\I$ is defined as
\begin{equation}
\label{eq:set_transition_probabilities_alternative_IMDP}
\begin{split}
    \Gamma_{q,a}^{\rm IMDP} = \{ &\gamma\in\mathcal{D(\Qimdp)}: \underline P^{\rm IMDP}(\qimdp,a,\qimdp')\leq \gamma(\qimdp')\leq \overline P^{\rm IMDP}(\qimdp,a,\qimdp')\:\text{for all}\:\qimdp'\in\Qimdp\}.
\end{split}
\end{equation}
This choice of the abstraction model allows to use efficient synthesis 
algorithms for IMDPs \cite{givan2000bounded}, 
\cite{koutsoukos2006computational} to solve Problem \ref{prob:Syntesis}. By 
the definition of the transition probability bounds of $\I$ in 
\eqref{eqn:transition_probability_bounds_alternative1}, $\Gamma^{\rm IMDP}$ 
satisfies Proposition \ref{prop:consistency}, effectively capturing all 
possible transition probabilities of System~\eqref{eq:system}. However, IMDP 
$\I$ describes the uncertainty in a very loose way, i.e., set 
$\Gamma^{\rm IMDP}$ of $\I$ can be excessively large for many ambiguity sets $\mathcal{P}_v$. Intuitively, this is caused by $\Gamma_{\qimdp,a}^{\rm IMDP}$ being only defined through decoupled interval constraints for every successor state $\qimdp'\in\Qimdp$,
This is likely to result in more conservative solutions to the reach-avoid problem as we show in Section~\ref{sec:case-studies}.
\end{rem}

\section{Robust Strategy Synthesis}
\label{sec:synthesis}

Our goal is to synthesize a switching strategy $\optpolicy_{\x}$ for System 
\eqref{eq:system} that maximizes \eqref{eq:problem}. 
To capture the distributional uncertainty and the effect of quantization,
we consider the proposed robust MDP abstraction $\M$. 
 We synthesize the robustly maximizing strategy $\optpolicy$ for the abstraction $\M$, which we then refine on the original system retaining formal guarantees of correctness.
 In Sections \ref{sec:robust_value_iteration}, \ref{sec:LinearProgrammingValueITeration}, and \ref{sec:DedicatedAlgorithmValueITeration}, we show how an optimal strategy $\optpolicy$ for the abstraction $\M$ can be efficiently computed via linear programming. Then, in Section \ref{sec:correctness}, we prove the correctness of our strategy synthesis approach.

\subsection{Robust Dynamic Programming}
\label{sec:robust_value_iteration}

Recall that, in our robust MDP abstraction $\mathcal{M} = (Q,A,\Gamma)$, the uncertainties of $\mathcal{M}$ are characterized by an adversary $\xi$, which at each time step and given a path of $\mathcal{M}$ and an action, selects a feasible distribution from $\Gamma$ (see Definition \ref{def:adversary}). Therefore, to be robust against all uncertainties, as common in the literature \cite{iyengar2005robust,nilim2005robust}, we aim to synthesize a strategy $\optpolicy$ such that, given a horizon $K\in\naturals_0\cup\{\infty\}$,
\begin{align}
    \label{eq:optimal_strategy_robust_MDPs}
    \optpolicy \in \arg\max_{\policy\in\Sigma}\inf_{\xi\in\Xi}P_{\rm{reach}}(Q_{\rm safe},Q_{\rm{tgt}},K \mid q, \policy,\xi),
\end{align}
for all $\qmdp\in \Qmdp$, where $P_{\rm{reach}}(Q_{\rm safe},Q_{\rm{tgt}},K \mid q, \policy,\xi)$ is defined as in \eqref{eq:reachability_original_system} for System \eqref{eq:system}. We call $\sigma^*$ the \emph{optimal robust} or simply \emph{optimal} strategy.

%

We denote by $\underline p^K$ and $\overline p^K$, respectively, the worst and best-case probabilities of the paths of $\mathcal{M}$ satisfying the reach-avoid specification under optimal strategy $\optpolicy$, i.e., 
\begin{subequations}
    \begin{align}
        \underline p^K(\qmdp) &:= \inf_{\xi\in\Xi}P_{\rm{reach}}( Q_{\rm safe},Q_{\rm{tgt}},K \mid q, \optpolicy,\xi) \label{eq:def_bounds_reachability_robust_MDPs_lower},\\
        \overline p^K(\qmdp) &:=\sup_{\xi\in\Xi}P_{\rm{reach}}(Q_{\rm safe},Q_{\rm{tgt}},K \mid q, \optpolicy,\xi) \label{eq:def_bounds_reachability_robust_MDPs_upper}
    \end{align}
\end{subequations}
for all $\qmdp\in \Qmdp$. 
In the following subsections, we show how to compute the previous quantities via robust dynamic programming. We distinguish the cases of finite ($K<\infty)$ and infinite ($K = \infty$) horizon, since they require different treatment.

\subsubsection{Finite Horizon}

The following proposition shows that the bounds in \eqref{eq:def_bounds_reachability_robust_MDPs_lower} and \eqref{eq:def_bounds_reachability_robust_MDPs_upper} and the optimal strategy in 
 \eqref{eq:optimal_strategy_robust_MDPs} can be obtained via dynamic programming.
\begin{prop}
\label{prop:robust_value_iteration}
Let $\underline p^{K}$ be as defined in \eqref{eq:def_bounds_reachability_robust_MDPs_lower} and $k\in [0 : K-1]$. Then, it holds that
\begin{align}
\label{eq:robust_value_iteration_lower_bound}
    \underline p^{k+1}(\qimdp) = 
        \begin{cases}
             1 &\text{if }\:\qimdp \in Q_{\rm{tgt}}\\ \max\limits_{a\in 
             \Aimdp}\min\limits_{\gamma\in\Gamma_{\qimdp,a}} \sum\limits_{\qimdp'\in 
             Q}\gamma(\qimdp')\underline p^{k}(\qimdp') &\text{otherwise},
             \end{cases}
\end{align}
with initial condition $\underline p^{0}(\qmdp) = 1$ for all $q\in Q_{\rm{tgt}}$ and $0$ otherwise.
Furthermore, for each path $\pathmdp^k$ with $k\in [0 : K-1]$, it holds that
\begin{align}
\label{eq:robust_dynamic_programming_optimal_strategy_robust_MDP}
    \optpolicy(\pathmdp^k) \in \arg\max_{a\in \Aimdp}\left\{ 
    \min_{\gamma\in\Gamma_{\last(\pathmdp^k),a}} \sum_{\qimdp'\in 
    Q}\gamma(\qimdp')\underline p^{K-k-1}(\qimdp') \right\}. 
\end{align}
\end{prop}
\begin{proof}
    The proof follows directly by applying \cite[Theorem 1]{el2005robust} to the setting of reachability.
\end{proof}
A consequence of Proposition \ref{prop:robust_value_iteration} is that, in the finite horizon ($K < \infty$) case, there exists an optimal policy that is Markovian (note that the resulting strategy is time-dependent in general). 
Furthermore, once an optimal strategy $\optpolicy$ is fixed, 
the upper bound $\overline p^{K}$ of $P_{\rm reach}$ in \eqref{eq:def_bounds_reachability_robust_MDPs_upper} can be readily computed via the robust dynamic programming iterations 
\begin{align}
\label{eq:robust_value_iteration_upper_bound}
    \overline p^{k+1}(\qimdp) = 
        \begin{cases}
             1 &\text{if}\:\qimdp\in Q_{\rm{tgt}}\\ 
             \max\limits_{\gamma\in\Gamma_{\qimdp,\optpolicy}} \sum\limits_{\qimdp'\in 
             Q}\gamma(\qimdp')\overline p^{k}(\qimdp') &\text{otherwise},
             \end{cases}
\end{align}
for all $k\in [0 : K-1]$, which is analogous to that in \eqref{eq:robust_value_iteration_lower_bound}
and has initial condition $\overline p^{0}(\qmdp) = 1$ for all $q\in 
Q_{\rm{tgt}}$ and $\overline p^{0}(\qimdp) = 0$ otherwise.

In Section~\ref{sec:LinearProgrammingValueITeration} we show that the robust dynamic iterations in \eqref{eq:robust_value_iteration_lower_bound} and \eqref{eq:robust_value_iteration_upper_bound} boil down to solving a finite number of linear programs, but first, in the next subsection, we consider the infinite horizon case.

\subsubsection{Infinite Horizon}
\label{sec:unbounded_horizon}
In the case of an infinite time horizon ($K = \infty$), we seek to bound the probability of the paths of $\M$ that with arbitrary lengths reach $Q_{\rm{tgt}}$ while remaining in $Q_{\rm{safe}}$. 
To also synthesize an optimal strategy,  we first show that the robust dynamic programming iterations in \eqref{eq:robust_value_iteration_lower_bound} converge to a fixed point. Then, using this result, we show that for the infinite horizon case, there exists an optimal stationary strategy, and we propose an approach to obtain it.
All the proofs of the theorems in this section can be found in~\ref{app:proofs_section_unbounded_horizon}.

We first 
show that the lower bound $\underline p^\infty$
on the infinite-horizon reachability probability obtained by  
\eqref{eq:def_bounds_reachability_robust_MDPs_lower}
for $K = \infty$ is a fixed-point of the robust dynamic programming operator in \eqref{eq:robust_value_iteration_lower_bound}. 
%
\begin{thm}
\label{thm:robust_dynamic_programming_infinite_horizon_lower_bound}
    Let $\{\underline p^k\}_{k\in\naturals_0}$ be the infinite sequence defined recursively by \eqref{eq:robust_value_iteration_lower_bound} with initial condition  $\underline p^0(q) = 1$ for all $q\in Q_{\rm{tgt}}$ and $0$ otherwise. Then $\{\underline p^k\}$ converges to $\underline p^\infty$, which is the least fixed-point of \eqref{eq:robust_value_iteration_lower_bound} on $\mathbb{R}_{\ge 0}^{|Q|}$.
\end{thm}
%
%
Theorem~\ref{thm:robust_dynamic_programming_infinite_horizon_lower_bound} guarantees that, to obtain the optimal reachability probabilities,
we can simply run the robust dynamic programming iteration in \eqref{eq:robust_value_iteration_lower_bound} until convergence. We next address the problem of finding an optimal strategy. In particular, we are interested in determining a stationary optimal strategy that has the convenient feature of using the same decision rule at every time step. The following result establishes that stationary optimal strategies indeed exist.
%
\begin{prop}\label{prop:existence_stationary optimal_strategy_unbounded_horizon}
    The infinite horizon reachability problem \eqref{eq:def_bounds_reachability_robust_MDPs_lower} admits a stationary strategy $\sigma'\in \Sigma_s$.
    \end{prop}

Since the mere existence of an optimal strategy is not sufficient for synthesis purposes, we seek conditions under which a stationary strategy that we will later explicitly construct is also optimal. The following result provides such optimality conditions, which generalize the corresponding requirements in the case of standard MDPs (see \cite[Theorem 7.1.7]{puterman2014markov}).

    \begin{prop}  \label{prop:sufficient_conditions_optimal_strategy_unbounded_horizon}
    Consider the set $Q_{\rm{reach}} := \{q\in Q : \underline p^\infty(q) > 0\}$ of states that have positive probability of reaching $Q_{\rm tgt}$ for some strategy, and let $\sigma^*\in\Sigma_s$ be a stationary strategy that, for each state $q \in Q$, satisfies
    \begin{subequations}
    \begin{align}
    & \sigma^*(q) \in A^*(q) := \arg\max\limits_{a\in \Aimdp}\Big\{ 
    \min\limits_{\gamma\in\Gamma_{q,a}} \sum\limits_{\qimdp'\in 
    Q}\gamma(\qimdp')\underline p^\infty(\qimdp') \Big\} \label{set:A:star} \\
    & \hspace{1.5em}\lim\limits_{k\to\infty} P_\xi^{q,\policy^*}[\pathmdp(k)\in Q_{\rm{reach}}] = 0\quad\textup{for all}\;\xi\in\Xi. 
    \label{proper:strategy}
    \end{align}
    \end{subequations}
    Then $\sigma^*$ is optimal, i.e., it satisfies  \eqref{eq:optimal_strategy_robust_MDPs}. Furthermore, condition \eqref{set:A:star} is also necessary for every $\sigma\in\Sigma_s$ to be optimal.
\end{prop}
The first condition of Proposition~\ref{prop:sufficient_conditions_optimal_strategy_unbounded_horizon} 
imposes the requirement that $\sigma^*$ can only pick actions that attain the maximum in the dynamic programming recursion in~\eqref{eq:robust_value_iteration_lower_bound}.The second requirement of the proposition is that under strategy $\sigma^*$, all paths should eventually exit $Q_{\rm{reach}}$ and remain outside that set forever, for every choice of the adversary and all initial conditions.  A strategy that satisfies both these conditions is called \emph{proper optimal}. We stress that, as we prove in Theorem~\ref{thm:optimal_strategy_unbounded_horizon}, a proper optimal strategy always exists.
\begin{rem}
Restricting ourselves to stationary strategies comes at the price of enforcing $\sigma^*$ to be proper optimal.
In particular, the second condition required by a strategy $\sigma^*$ to be optimal guarantees that, under $\sigma^*$, all states need to eventually exit $Q_{\rm{reach}}$ and remain outside forever. This additional condition is needed because there may be multiple strategies maximizing \eqref{set:A:star} and some of them may introduce loops (cycles) with probability $1$ in the graph of the resulting  Markov chain. As a consequence, states that have nonzero probability of reaching $Q_{\rm tgt}$ for an optimal strategy may never achieve this under a strategy that only satisfies~\eqref{set:A:star}.
\end{rem}
%
%

We next extract an optimal proper stationary strategy by generalizing the approach employed in the proof of \cite[Theorem 10.102]{baier2008principles} for MDPs. 
%
    %
    %
%
To this end, we recursively define the sets 
\begin{align} \label{dfn:nested_sets}
 Q^m := Q^{m-1}\cup\Big\{q\in Q\setminus Q^{m-1} : \exists\: a\in A^*(q) \:\: \text{s.t.}\: \min_{\gamma\in\Gamma_{q,a}} \sum_{\qimdp'\in Q^{m-1}}\gamma(\qimdp')  > 0 \Big\},
 \end{align}
 for each $m\in\{0,1,\dots,m_{\max}\}$, where $Q^0 := Q_{\rm{tgt}}$ and 
 \begin{align*}
 m_{\max}:=\min\{m\in\mathbb N_0:Q^{m}=Q^{m-1}\}.
 \end{align*}
Each $Q^m$ is the set of backward reachable states from $Q_{\rm{tgt}}$ in $m$ steps by taking actions in $A^*$. Here reachability is interpreted in the worst-case sense where the underlying graph of the robust MDP depends on the worst-case choice of the adversary. Since $Q$ is finite, the maximum number of backward steps to reach a state in $Q$ from $Q_{\rm{tgt}}$ is bounded.
Thus, $m_{\max}$ is always well-defined. 
The following result delineates how we can use the sets $Q^m$ to extract a proper and optimal stationary strategy.
\begin{thm}
\label{thm:optimal_strategy_unbounded_horizon}
   %
    %
  Let $A^*(q)$ be as given in \eqref{set:A:star}. Define a strategy $\sigma^*\in\Sigma_s$ with   
    \begin{align*}
   \sigma^*(q) \in \left\{ a\in A^*(q) : 
    \min_{\gamma\in\Gamma_{q,a}} \sum_{\qimdp'\in 
    Q^{m-1}}\gamma(\qimdp') > 0 \right\} 
    \end{align*}
    for all $q\in Q^m\setminus Q^{m-1}$ and $m\in\{1,\dots,m_{\max}\}$, and with arbitrary actions when $q\in Q_{\rm tgt} \cup (Q \setminus Q_{\rm{reach}})$. Then $\sigma^*$ is proper and optimal.
\end{thm}
Theorem~\ref{thm:optimal_strategy_unbounded_horizon} 
shows how to obtain an optimal stationary strategy via a backward reachability analysis.
Note that, although this algorithm has to be employed once $\underline p^{\infty}$ has been obtained via robust dynamic programming, it converges within $|Q_{\rm{reach}}| - |Q_{\rm{tgt}}|$ iterations. Consequently, it is generally much faster than performing robust dynamic programming.

Having computed $\sigma^*$, we obtain the upper bound $\overline p^\infty$ in the reachability probability by iterating on
\eqref{eq:robust_value_iteration_upper_bound} until we achieve convergence, to $\overline p^\infty$, starting from $\overline p^0(q) = 1$ for all $q\in Q_{\rm{tgt}}$ and $0$ otherwise. The proof is analogous to that of Theorem~\ref{thm:robust_dynamic_programming_infinite_horizon_lower_bound}, and is omitted for simplicity.

\subsection{Computation of Robust Dynamic Programming via Linear Programming}
\label{sec:LinearProgrammingValueITeration}

We now show how for each state $\qimdp \in \Qimdp$ and time $k\in [0:K-1],$ dynamic programming in \eqref{eq:robust_value_iteration_lower_bound} reduces to solving $|A|$ linear programs. 
In particular, Theorem \ref{th:RVI_LP} below guarantees that the inner problem in recursion \eqref{eq:robust_value_iteration_lower_bound} can be solved via linear programming. In what follows we explicitly consider $ \underline p^k$, the upper bound $ \overline p^k$ follows similarly.  
Before formally stating our result, we need to introduce some notations. Let $W\subseteq\mathbb{R}^n$ be a set containing the support of $\distribution_v^{\text{true}}$, and for each $q\in Q$ and $a\in A$ denote
\begin{subequations}
\begin{align} 
\overline{\mathcal N}^{q,a} :=  & \{i\in \mathcal N:\overline P(q,a,q_i)>0\}, \label{index:set:hatNqa} \\
\overline{\mathcal N}_W^{q,a} := & \{i\in \mathcal N : (f_a(q)+W)\cap  q_i\ne\emptyset\},  \label{index:set:Nqa}
\end{align}
\end{subequations}
where  $\overline P(q,a,q_i)$ is the upper transition probability bound of the nominal IMDP as defined in \eqref{eqn:transition_probability_bounds1} and \eqref{eq:transition_probability_bounds3}.
$\overline{\mathcal N}^{q,a}$ and $\overline{\mathcal N}^{q,a}_W$ represent respectively the set of indices of the discrete states that are reachable from $q$ under action $a$ for the nominal noise distribution, and for every distribution supported on $W$.
Note that as the nominal distribution is also supported on $W$, it holds that  $\overline{\mathcal N}^{q,a} \subseteq \overline{\mathcal N}_W^{q,a}$, and 
if $W = \mathbb{R}^n$, then $\overline{\mathcal N}_W^{q,a}=\mathcal N$. 
Intuitively, in Theorem \ref{th:RVI_LP}, from each starting region $q$, we only consider transitions to states indexed by $\overline{\mathcal N}_W^{q,a}$ as those are the only states reachable by some distribution with support in $W$.


\begin{thm}[Robust Dynamic Programming as a Linear Program]
\label{th:RVI_LP}
Consider the robust dynamic programming \eqref{eq:robust_value_iteration_lower_bound} for the robust MDP $\mathcal{M} = (\Qmdp,\Amdp,\Gamma)$ and assume that the support of $P_v^{\rm true}$ is contained in the set $W\subseteq\mathbb R^n$. Then, for every $k\in [0:K-1]$,  $\qmdp\in Q,$ and $a\in A$, the inner minimization problem in \eqref{eq:robust_value_iteration_lower_bound} is equivalent to the following linear program: 
\begin{align}
\label{eq:LP_cost}
\min_{\gamma_i,\widehat\gamma_j,\pi_{ij}} \;
\sum_{i\in \overline{\mathcal N}_W^{q,a}}\gamma_i\underline p^{k}(\qmdp_i), \hspace{12em}
\end{align}\vspace{-1em}
\begin{subequations} \label{eq:feasible set_robust_MDP_cost} 
\begin{align}
{\rm s.t.} \;  \underline P(\qmdp,a,\qmdp_j) \le \widehat\gamma_j &\le \overline P(\qmdp,a,\qmdp_j)  \qquad  &j\in\overline{\mathcal N}^{q,a}  \label{constraint:a} \\
 \sum_{j\in \overline{\mathcal N}^{q,a}}\widehat\gamma_j &= 1 & \label{constraint:b} \\
 \pi_{ij}  &\geq 0, &i\in\overline{\mathcal N}_W^{q,a},j\in \overline{\mathcal N}^{q,a} \label{constraint:c} \\
 \sum_{i\in \overline{\mathcal N}_W^{q,a}}\pi_{ij}  &=\widehat \gamma_j, \qquad &j\in \overline{\mathcal N}^{q,a} \label{constraint:d} \\
 \sum_{j\in \overline{\mathcal N}^{q,a}}\pi_{ij} &=\gamma_i,  &i\in \overline{\mathcal N}_W^{q,a} \label{constraint:e} \\
 \sum_{i\in\overline{\mathcal N}_W^{q,a},j\in\overline{\mathcal N}^{q,a}}\pi_{ij}c(q_i,q_j) &\le \varepsilon^s,& \label{constraint:f}
\end{align}
\end{subequations}
where $\underline P$, $\overline P$, and $c$ are defined in \eqref{eqn:transition_probability_bounds1},  \eqref{eq:transition_probability_bounds3}, and \eqref{eq:cost}, respectively, and $s$ is given in Assumption \ref{ass:1}.
\end{thm}
\begin{proof}
We show that, for a fixed state $\qmdp\in \Qmdp$ and action $a\in \Amdp$, the set of transition probabilities $\Gamma_{q,a}$ defined in \eqref{eq:set_transition_probabilities_robust_MDP} is the polytope described by the linear equations~\eqref{eq:feasible set_robust_MDP_cost}. First, assume that $W\equiv\mathbb R^n$ and  consider the expressions in \eqref{eq:feasible set_robust_MDP_cost} for $\overline{\mathcal N}_W^{q,a} = \overline{\mathcal N}^{q,a} = \mathcal{N}$.  Let $\gamma\equiv(\gamma_1,\ldots,\gamma_N)\in \Gamma_{\qmdp,a}$. From the definition of $\Gamma_{\qmdp,a}$ in \eqref{eq:set_transition_probabilities_robust_MDP}, there exist $\widehat\gamma\equiv(\widehat\gamma_1,\ldots,\widehat\gamma_N)\in\widehat\Gamma_{\qmdp,a}$ and an optimal transport plan $\pi\equiv(\pi_{ij})_{i,j=1,\ldots,N}$ that transports mass from $\widehat\gamma$ to $\gamma$ with a cost $\mathcal{T}_c(\gamma,\widehat\gamma)$ smaller than $\varepsilon^s$. Consider now the set $\widehat{\Gamma}_{q,a}$ as defined in \eqref{eq:set_transition_probabilities_nominal_IMDP} with the transition probability bounds $\underline P$ and $\overline P$ given by \eqref{eqn:transition_probability_bounds1} and \eqref{eq:transition_probability_bounds3}. Then $\widehat\gamma$ satisfies the constraints \eqref{constraint:a} and \eqref{constraint:b}. Since the optimal transport cost $\mathcal{T}_c(\gamma,\widehat\gamma)$ is attained by the transport plan $\pi$, it follows from  \eqref{eq:optimal_transport_discrepancy} with $X\equiv Q$ and $c$ as in \eqref{eq:cost} that 
\begin{align*}
\mathcal{T}_c(\gamma,\widehat\gamma)=\sum_{i,j\in \mathcal N}\pi_{ij}c(q_i,q_j). 
\end{align*}
Thus, since $\mathcal{T}_c(\gamma,\widehat\gamma)$ is less than $\varepsilon^s$, we deduce that $\gamma$, $\widehat\gamma$, and $\pi$ satisfy the linear constraints \eqref{constraint:c}-\eqref{constraint:f}. Conversely, one can check along the same lines that for any $\gamma\equiv(\gamma_1,\ldots,\gamma_N)$, $\widehat\gamma\equiv(\widehat\gamma_1,\ldots,\widehat\gamma_N)$, and $\pi\equiv(\pi_{ij})_{i,j=1,\ldots,N}$ satisfying \eqref{eq:feasible set_robust_MDP_cost}, it also holds that $\gamma\in\Gamma_{q,a}$.

In the general case where the support of $p_v^{\rm true}$ belongs to some known set $W\subseteq\mathbb R^n$, we can embed this constraint into the ambiguity sets of transition probabilities and their projections on the abstraction. In particular,
we can describe the set of transition probabilities $\Gamma_{q,a}$ via transport plans that couple the regions reachable through the nominal distribution with the regions reachable by some distribution in the ambiguity set.  From \eqref{index:set:Nqa} and \eqref{index:set:hatNqa}, these regions are indexed respectively by $\overline{\mathcal N}^{q,a}$ and $\overline{\mathcal N}_W^{q,a}$. We can therefore just substitute the set $\mathcal{N}$ by the previous ones, since $\widehat\gamma$, $\gamma$, and $\pi$ are zero for indices outside $\overline{\mathcal N}^{q,a}$, $\overline{\mathcal N}_W^{q,a}$, and $\overline{\mathcal N}_W^{q,a}\times\overline{\mathcal N}^{q,a}$, respectively. The proof is now complete.
\end{proof}

Intuitively, for each state $\qmdp\in \Qmdp$ and action $a\in \Amdp$, 
the Constraints \eqref{constraint:a}-\eqref{constraint:f} capture the union of 
the $c$-transport cost ambiguity balls in $\mathcal{D}(Q)$ that have radius 
$\varepsilon^s$ and centers all possible distributions of the nominal IMDP. 
Specifically, Constraints \eqref{constraint:a} and \eqref{constraint:b} 
represent the distributions $\widehat\gamma$ of the nominal IMDP, i.e., the 
set $\widehat\Gamma_{q,a}$. Constraints 
\eqref{constraint:c}-\eqref{constraint:e} describe a transport plan $\pi$, 
i.e., a nonnegative measure on $Q\times Q$ (cf. \eqref{constraint:c}), which 
has as its marginals the distribution $\widehat\gamma$ of the nominal IMDP (cf. 
\eqref{constraint:d}), and the target distribution $\gamma$ (cf. 
\eqref{constraint:e}), respectively. Finally, Constraint \eqref{constraint:f} 
implies that transport cost to reach the target distribution $\gamma$ is 
bounded by $\varepsilon^s$, namely, that $\gamma$ belongs to the $c$-transport 
cost ambiguity ball of radius $\varepsilon^s$.

\begin{rem}
Theorem \ref{th:RVI_LP} guarantees that, similarly to \emph{IMDP}s without distributional uncertainty \cite{givan2000bounded,lahijanian2015formal}, optimal policies can be computed by solving a set of linear programs. In particular, for every $k\in [0:K-1]$ and $\qmdp\in Q,$ we can solve the linear program in Theorem \ref{th:RVI_LP} for each $a\in A$ and take the action that maximizes the resulting value function. However, in the IMDP case the resulting linear program has substantially fewer variables and constraints compared to the problem in Theorem \ref{th:RVI_LP} (order of $N$ for IMDPs, against order of $N^2$ for Theorem \ref{th:RVI_LP}). Nevertheless, as we detail in Subsection \ref{sec:DedicatedAlgorithmValueITeration}, we can reformulate the LP in Theorem \ref{th:RVI_LP} to obtain a problem of substantially lower complexity, which can be solved more efficiently.


\end{rem}

\subsection{Dual Reformulation of the Robust Dynamic Programming Algorithm}
\label{sec:DedicatedAlgorithmValueITeration}

In this section, we make use of linear programming duality to obtain a computationally efficient  solution to the linear program \eqref{eq:LP_cost}-\eqref{eq:feasible set_robust_MDP_cost}. The approach consists of formulating the dual of program \eqref{eq:LP_cost}-\eqref{eq:feasible set_robust_MDP_cost}, as common in the literature of robust MDPs\cite{ramani2022robust,el2005robust,iyengar2005robust}, and then introducing a tailored algorithm that takes into account its structure to solve it efficiently. 
The following theorem provides the dual  reformulation of the linear program \eqref{eq:LP_cost}-\eqref{eq:feasible set_robust_MDP_cost}.

%
\begin{thm}
\label{thm:dual_problem}
Consider the optimization problem
\begin{align} \label{eq:dual_hscc}
\max_{\mu\ge 0,\lambda} G(\lambda,\mu),
\end{align}
with 
\begin{subequations} \label{eq:dual_hscc:fnc}
\begin{align}
G(\lambda,\mu) \; & := \sum_{j \in \overline{\mathcal N}^{q,a}}\min\{\underline P(q,a,q_j)( h_j(\mu)-\lambda),\overline P(q,a,q_j)(h_j(\mu)-\lambda)\} -\mu\epsilon +\lambda \label{eq:dual_hscc1} \\
h_j(\mu) \; & := \min_{i \in \overline{\mathcal N}_W^{q,a}}\{\underline p^k(q_i) + \mu c(q_i,q_j)\}. 
\label{eq:dual_hscc2}
\end{align}
\end{subequations}
%
%
%
%
%
Then, the function $G$ is concave, and the optimal value of problem \eqref{eq:dual_hscc}-\eqref{eq:dual_hscc:fnc} is the same as that of problem \eqref{eq:LP_cost}-\eqref{eq:feasible set_robust_MDP_cost}.
Furthermore, for every $\mu\ge 0$, the maximization of $G$ with respect to $\lambda$ can be carried out by only evaluating it as $\lambda$ ranges over the finite set $\{h_j(\mu)\}_{j \in \overline{\mathcal N}^{q,a}}$, i.e.,
%
\begin{align*}
    \max_\lambda G(\lambda,\mu) = \max_{j\in \overline{\mathcal N}^{q,a}}\{G(h_j(\mu),\mu)\}.
\end{align*}
\end{thm}
The proof of Theorem~\ref{thm:dual_problem} is given in \ref{app:proof_dual}. 
Theorem \ref{thm:dual_problem} allows one to formulate Problem \eqref{eq:LP_cost}-\eqref{eq:feasible set_robust_MDP_cost} as a concave maximization problem in two scalar variables, $\lambda$ and $\mu$. Furthermore, as Problem~\eqref{eq:dual_hscc} is derived using duality, if we solve it with an iterative algorithm, then the algorithm can be stopped at any time to obtain a valid lower bound; thus, guaranteeing correctness. Note that $G(\lambda,\mu)$ is bivariate, and therefore its minimization does not fall within the class of dual optimization problems in \cite{ramani2022robust}, which also considers Wasserstein ambiguity sets. Nevertheless, our problem essentially reduces to the maximization of the univariate concave function $g(\mu) := \max_\lambda G(\lambda, \mu)$, whose 
values can be found by simply performing a finite search over the set $\{G(h_j(\mu),\mu)\}_{j\in \overline{\mathcal N}^{q,a}}$ by the last part of Theorem \ref{thm:dual_problem}. Therefore, the maximum can be found efficiently by using standard tools from scalar convex optimization.

\subsection{Correctness}
\label{sec:correctness}

In this section, we prove the correctness of our abstraction for System 
\eqref{eq:system}. We begin by refining the strategy $\policy^*$  to System 
\eqref{eq:system}. Let $J:\mathbb{R}^n\rightarrow \Qmdp$ map the continuous 
state $x\in\mathbb{R}^n$ to the corresponding discrete state $\qmdp\in \Qmdp$ 
of $\mathcal{M}$, i.e., for any $x\in\mathbb{R}^n$, 
\begin{align}\label{J:dfn}
J(x) = \qmdp \iff x\in \qmdp.
\end{align}
Given a finite path $\mypath_\x^k=\x_0 \xrightarrow{\pu_0} \x_1 
\xrightarrow{\pu_1} \ldots \xrightarrow{\pu_{k-1}} \x_k$ of System 
\eqref{eq:system}, we define by
$$ J(\mypath_\x^k)= J(\x_0) \xrightarrow{\pu_0} J(\x_1) \xrightarrow{\pu_1}  
\ldots \xrightarrow{\pu_{k-1}} J(\x_k)$$
the corresponding finite path of the MDP abstraction. Consequently, a strategy 
$\policy^*$ for $\mathcal{M}$ is refined to a switched strategy 
$\policy_x^*$ for System \eqref{eq:system} such that
\begin{align}
\label{eq:refine_strat}
    \policy_x^*(\mypath_\x^k) := \policy^*(J(\mypath_\x^k)).
\end{align}
The following theorem, which is a direct consequence of Theorem 
\ref{th:RVI_LP} and Theorem 2 in \cite{jackson2021formal}, ensures that the 
guarantees obtained for the robust MDP abstraction also hold for System~\eqref{eq:system}.
\begin{thm}[Correctness]
\label{th:correctness}
%
Let $\mathcal{M}$ be a robust MDP abstraction of System~\eqref{eq:system}, $\optpolicy\in \Sigma$ be an optimal strategy for $\mathcal{M}$, and $\optpolicy_\x$ be the corresponding switching strategy. Then, for every
%
$q\in Q$, $x\in \qmdp$, and $p_v\in\mathcal{P}_v$, it holds that
\begin{align*}
    \overline p^K(\qmdp) \geq P_{\rm{reach}}( X, X_{\rm tgt},K \mid x, 
    \optpolicy_{\x},\distribution_v) \geq\underline p^K(\qmdp),
\end{align*}
where $\underline p^K$ and $\overline p^K$ are defined in \eqref{eq:def_bounds_reachability_robust_MDPs_lower} and \eqref{eq:def_bounds_reachability_robust_MDPs_upper}.
\end{thm}
%
Theorem \ref{th:correctness} guarantees that in order to solve Problem 
\ref{prob:Syntesis} we can synthesize an optimal strategy $\optpolicy$ for a 
robust MDP abstraction of System \eqref{eq:system} and then simply check if 
$\underline p^K(J(x))$ is greater than the given threshold $p_{th}$.
\section{Experimental Results}
\label{sec:case-studies}

In this section, we evaluate our framework on several benchmarks. In Sections \ref{sec:unicycle2d}, \ref{sec:Jacksonnl} and \ref{sec:LTLf}, we consider systems with data-driven ambiguity sets built from $M$ i.i.d. samples $\noise^1,\ldots,\noise^M$ of the noise distribution $\distribution_v^{\rm{true}}$. 
Using these samples, we construct an ambiguity ball, whose center is the empirical distribution
\begin{align}
    \widehat p_v = \frac{1}{M}\sum_{i=1}^{M}\delta_{\noise^i},
\end{align}
of the data in~\cite{mohajerin2018data}, where $\delta_{\noise^i}$ denotes the Dirac distribution that assigns unit mass to  $\noise^i$.
As a metric we consider the 2-Wasserstein distance, which is often used in the literature \cite{AH-IY:21}. 
The radius $\varepsilon$ of the ambiguity set is a tuning parameter to hedge against the uncertainty coming from the imperfect knowledge of $\distribution_v^{\rm{true}}$.
%
Then, in Section \ref{sec:switched}, we depart from the data-driven setting and consider the case where the nominal distribution is Gaussian and $\varepsilon$ quantifies the degree of robustness against distributional shifts.


To synthesize the respective strategies,  we run the dynamic programming algorithm  for  required time horizon $K\in \mathbb{N}\cup\{\infty\}$ and use the following ``average error" to quantify the performance 
%
%
\begin{align}
    e_{\rm{avg}} := \frac{1}{N}\sum_{q\in Q}(\overline p^{K}(q) - \underline p^{K}(q)),
\end{align}
where $\overline p^{K}(q)$ and $ \underline p^{K}(q))$ are the upper and lower bounds of probability of satisfaction, respectively, as defined in \eqref{eq:robust_value_iteration_lower_bound} and \eqref{eq:robust_value_iteration_upper_bound}.
A summary of the numerical results for all the case studies are shown in Table~\ref{tab:results}. The main highlight in Table \ref{tab:results} is the superiority of our dual algorithm of Section~\ref{sec:DedicatedAlgorithmValueITeration} for strategy synthesis when compared to using the standard linear programming routine \emph{Linprog} (two orders of magnitude speed-up).

For all the case studies, the theoretical bounds on the reachability probabilities are validated empirically via $1000$ Monte Carlo simulations starting from $1000$ random initial conditions. 
To demonstrate that the theoretical bounds obtained with our approach are sound even in extreme cases, in Section~\ref{sec:switched}  (Experiment $\#12$), we consider remarkably large noise 
and show that the Monte Carlo probability values are still within the theoretical bounds.

All results are obtained by running a Matlab implementation of our framework on a single thread of an Intel Core i7 3.6GHz CPU with 32GB of RAM

\subsection{Linear System}
\label{sec:unicycle2d}

We consider a discrete-time version of the unicycle model in \cite{tabuada2008approximate}, which is obtained using a first order 
Euler discretization with a time step $\Delta t = 1$. We fix the velocity of the vehicle to the constant value $1$ and consider its orientation angle $u$ as the control input leading to the following switched system:
\begin{align}
\label{eq:unicycle_model}
    x_{k+1} &= x_k +  \Delta t\left(\begin{bmatrix}
    \cos{(u_k)}\\
    \sin{(u_k)}
    \end{bmatrix} + \noise_{k}\right),
\end{align}
where the state of the system is the vehicle position $x_k\in\mathbb{R}^2$ and its control input $u_k$ takes values in $U = \{0,\frac{2\pi}{8},\dots,7\frac{2\pi}{8}\}$.


%
\begin{table}
\centering
\caption{Summary of the results obtained for all case studies. The label ``robust MDP" 
denotes our proposed approach, while the label ``IMDP" refers to the 
alternative approach pointed out in Remark \ref{rem:model:choice}.}
\label{tab:results}
\scalebox{0.9}{
\begin{tabular}{ l c c c c l l l l }
\toprule
\multirow{2}{*}{\#} & Abstraction & \multirow{2}{*}{Spec.} & \multirow{2}{*}{$|Q|$} & \multirow{2}{*}{$\varepsilon$} & \multirow{2}{*}{$e_{\rm{avg}}$}  & \multirow{2}{*}{DP Alg.} & Abst. & Synth.\\
 & Class & & & &  &  & Time & Time\\
\midrule
$1$ & Robust MDP & Reach-avoid & $1601$ & $5\times10^{-3}$ & $0.05$ & Linprog & $6$ sec & $10.1$ h\\
$2$ & Robust MDP & ($K=40$) & $1601$ & $5\times10^{-3}$ & $0.05$ & Dual & $6$ sec &  $76$ sec\\
\midrule
$3$ & Robust MDP & Reach-avoid & $1601$ & $10^{-2}$ & $0.18$ & Linprog & $6$ sec & $10.4$ h\\
$4$ & Robust MDP & ($K=40$) & $1601$ & $10^{-2}$ & $0.18$ & Dual  & $6$ sec &  $94$ sec\\
\midrule
$5$ & Robust MDP & Reach-avoid & $1601$ & $1.5\times 10^{-2}$ & $0.33$ & Linprog & $6$ sec & $11.4$ h\\
$6$ & Robust MDP & ($K=40$) & $1601$ & $1.5\times 10^{-2}$ & $0.33$ & Dual  & $6$ sec &  $93$ sec\\
\midrule
$7$ & IMDP & Reach-avoid & $1601$ & $5\times 10^{-3}$ & $0.83$ & Linprog & $6$ sec & $1.2$ h\\
$8$ & IMDP & ($K=40$) & $1601$ & $5\times 10^{-3}$ & $0.83$ & Alg. in \cite{givan2000bounded} & $6$ sec &  $3.9$ sec\\
\midrule
$9$ & Robust MDP & Reach-avoid & $1601$ & $5\times 10^{-2}$ & $0.25$ & Linprog & $12$ sec & $2.43$ h\\
$10$ & Robust MDP & ($K=15$) & $1601$ & $5\times 10^{-2}$ & $0.25$ & Dual  & $12$ sec &  $22$ sec\\
\midrule
$11$ & Robust MDP & Reach-avoid & $3601$ & $1.3\times 10^{-2}$ & $0.13$ & Dual  & $35$ sec & $19.5$ min\\
$12$ & Robust MDP & ($K=\infty$) & $3601$ & $1.3\times 10^{-3}$ & $0.47$ & Dual  & $64$ sec &  $42$ min\\
\midrule
$13$ & Robust MDP & $\text{LTL}_f$ & $3601$ & $5\times 10^{-3}$ & $0.01$ & Dual  & $6$ sec & $2.78$ min\\
\bottomrule
\end{tabular}
}
\end{table}

The considered safe, unsafe (obstacle), and target sets are shown in Figure~\ref{fig:strategy_unicycle_R_6}. We partition the $[0,1]^2$ rectangle into a grid of $N = 1600$ regions. 
The ambiguity set is centered at an empirical distribution of $M = 10$ samples, which are drawn from a Gaussian mixture with two components, centered at $[-0.01 ,0]$ and $[0.01 ,0]$, respectively, and with the same covariance matrix ${\rm diag}(2.5\times 10^{-5},2.5\times 10^{-5})$. As a result, the centers of both components are separated by a distance close to the size of the state discretization. 

We obtained results for multiple values of the radius $\varepsilon$, shown in Table \ref{tab:results} (Experiments $\#1-\#6$), and compare the performance of our dual approach to solve the inner problems in \eqref{eq:robust_value_iteration_lower_bound} to the one obtained by using Linprog.
%
%
We observe
that the average error $e_{\rm{avg}}$ increases as the ambiguity set grows. This means that, as expected, bigger ambiguity leads to more conservative bounds.
%
\begin{figure}[h!] 
  \begin{subfigure}[t]{0.33\textwidth}
    \centering
    \includegraphics[width=\textwidth]{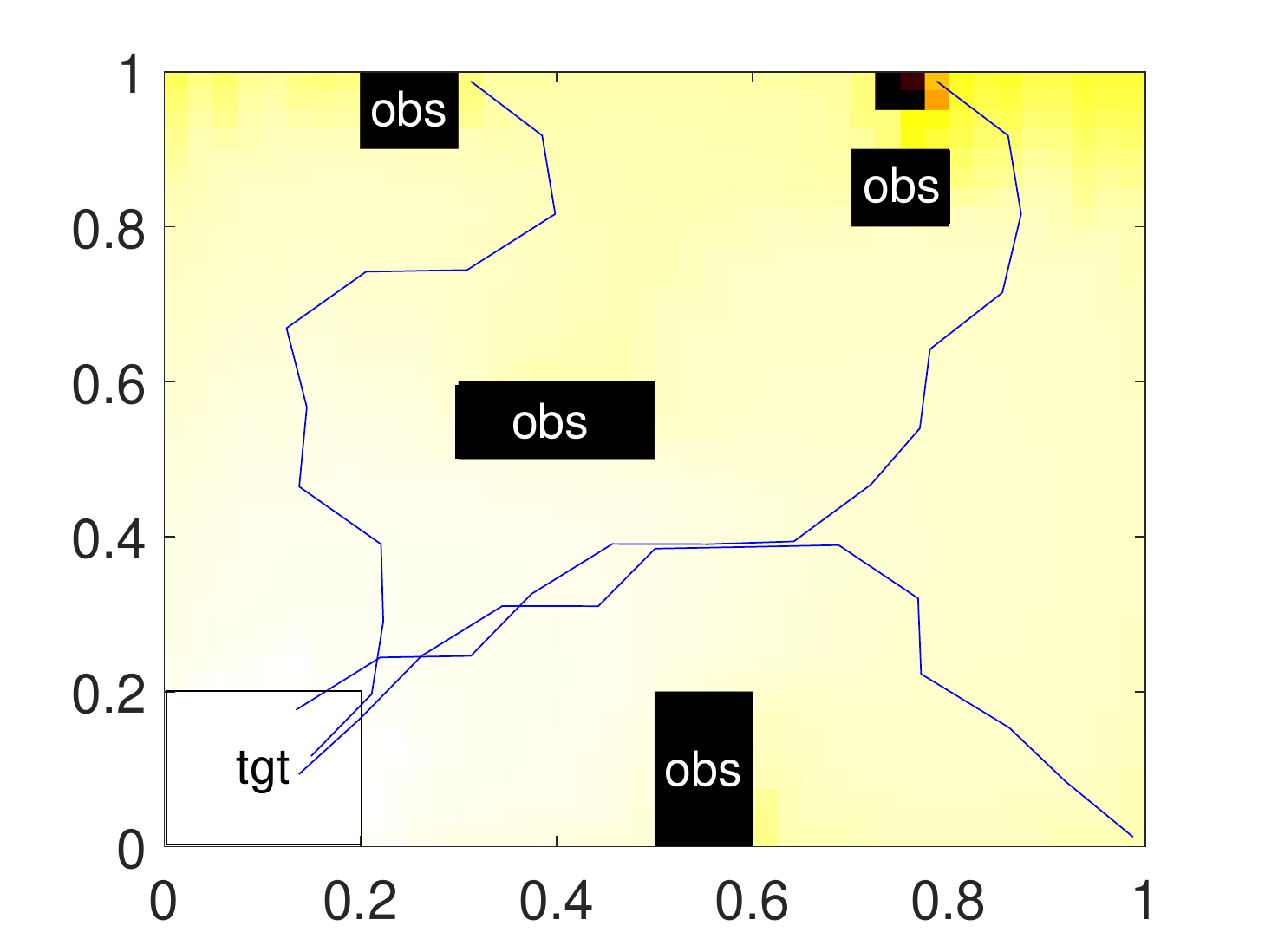}
    \caption{$\underline{p}^K$ for $\#2$}
\label{fig:P_reach_lower_unicycle_R_6}
  \end{subfigure}
  \begin{subfigure}[t]{0.33\textwidth}
    \centering
    \includegraphics[width=\textwidth]{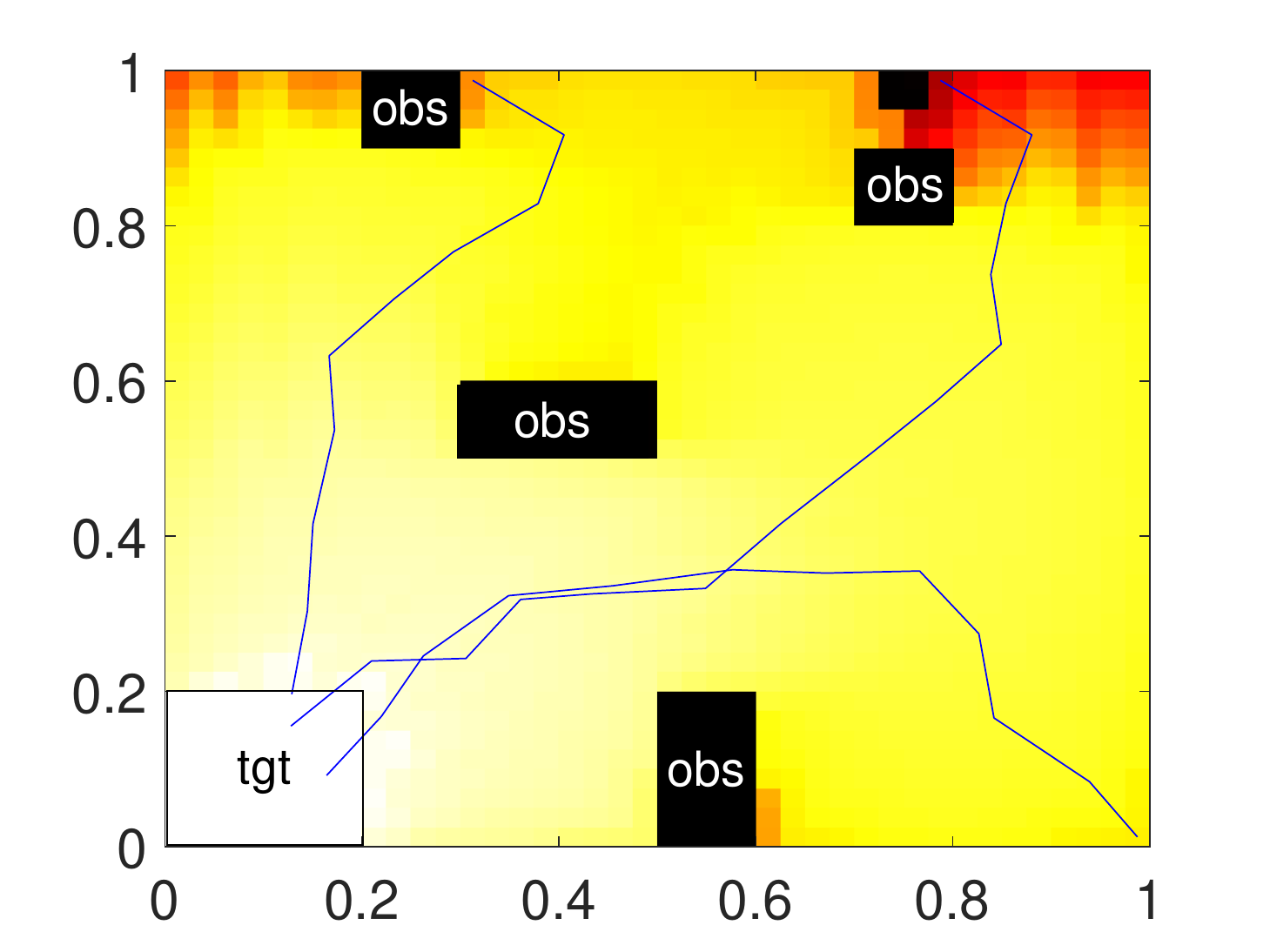}
    \caption{$\underline{p}^K$ for $\#3$}
\label{fig:P_reach_lower_unicycle_R_5}
  \end{subfigure} 
  \begin{subfigure}[t]{0.33\textwidth}
    \centering
    \includegraphics[width=\textwidth]{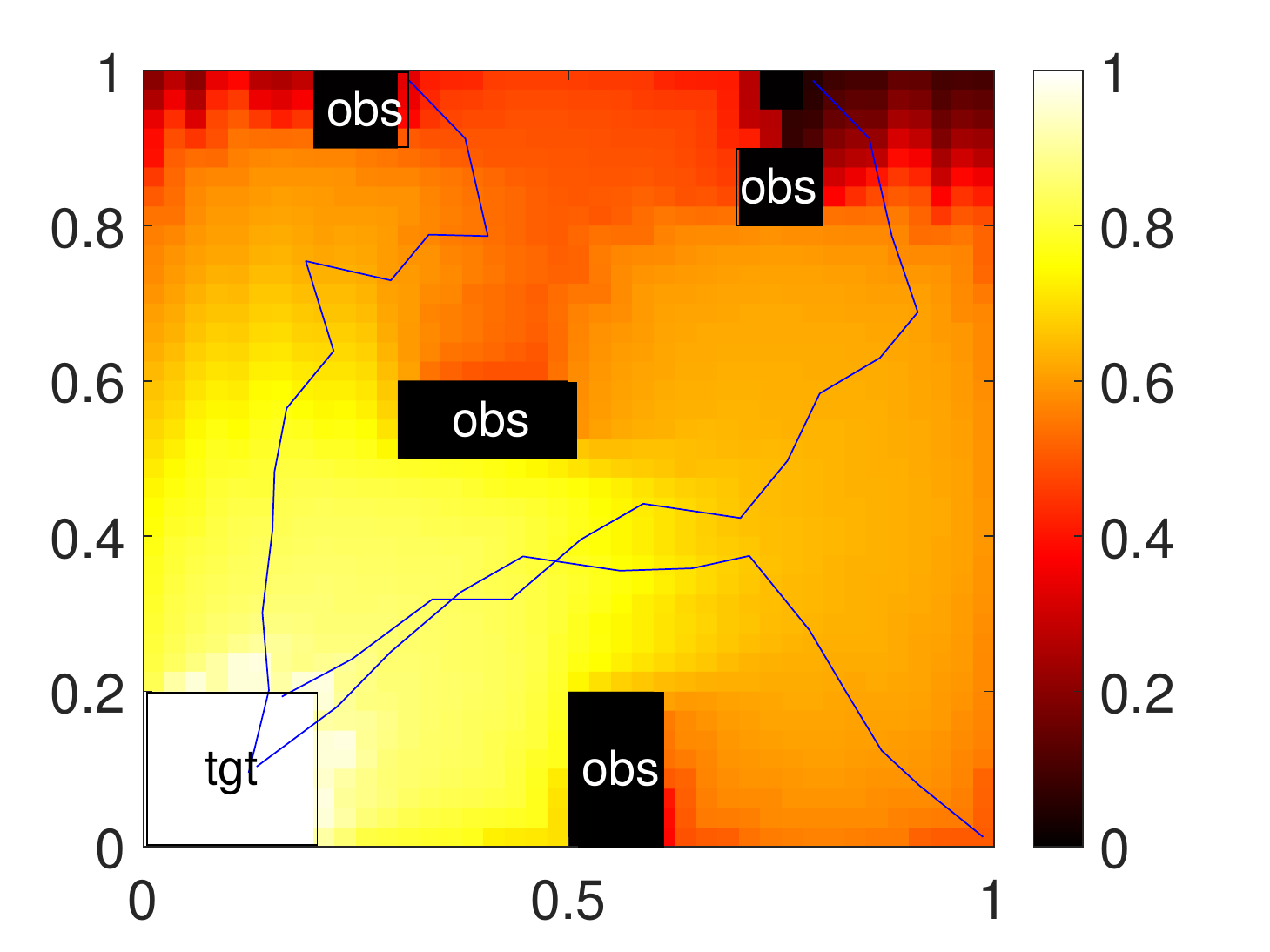}
    \caption{$\underline{p}^K$ for $\#5$}
\label{fig:P_reach_lower_unicycle_R_7}
  \end{subfigure}\\
  \centering
  \begin{subfigure}[t]{0.33\textwidth}
    \centering
    \includegraphics[width=\textwidth]{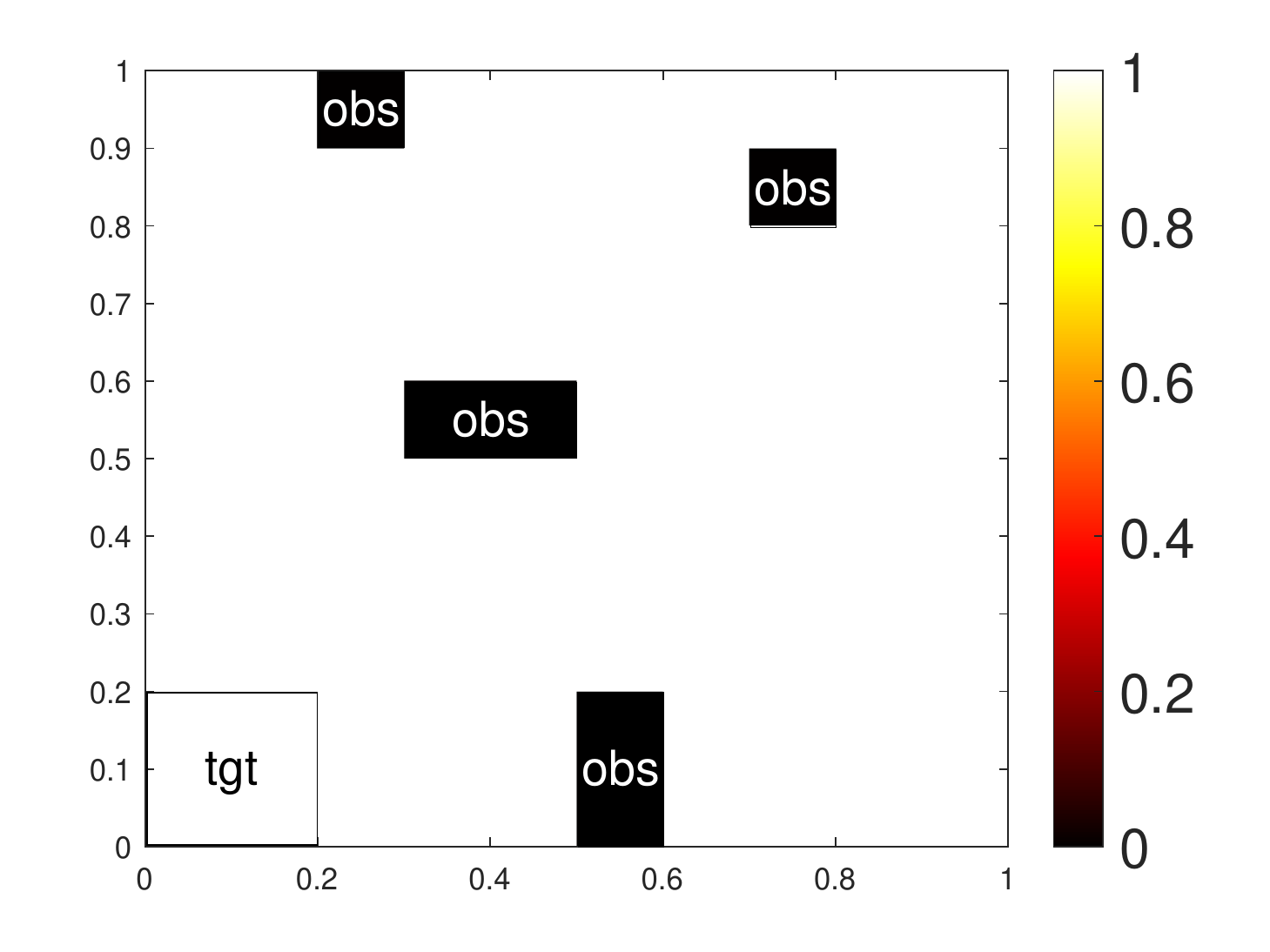}
    \caption{$\overline{p}^K$ for $\#2$}
\label{fig:P_reach_upper_unicycle2D}
  \end{subfigure}
  \hspace{0.05\textwidth}
  \begin{subfigure}[t]{0.33\textwidth}
    \centering
    \includegraphics[width=\textwidth]{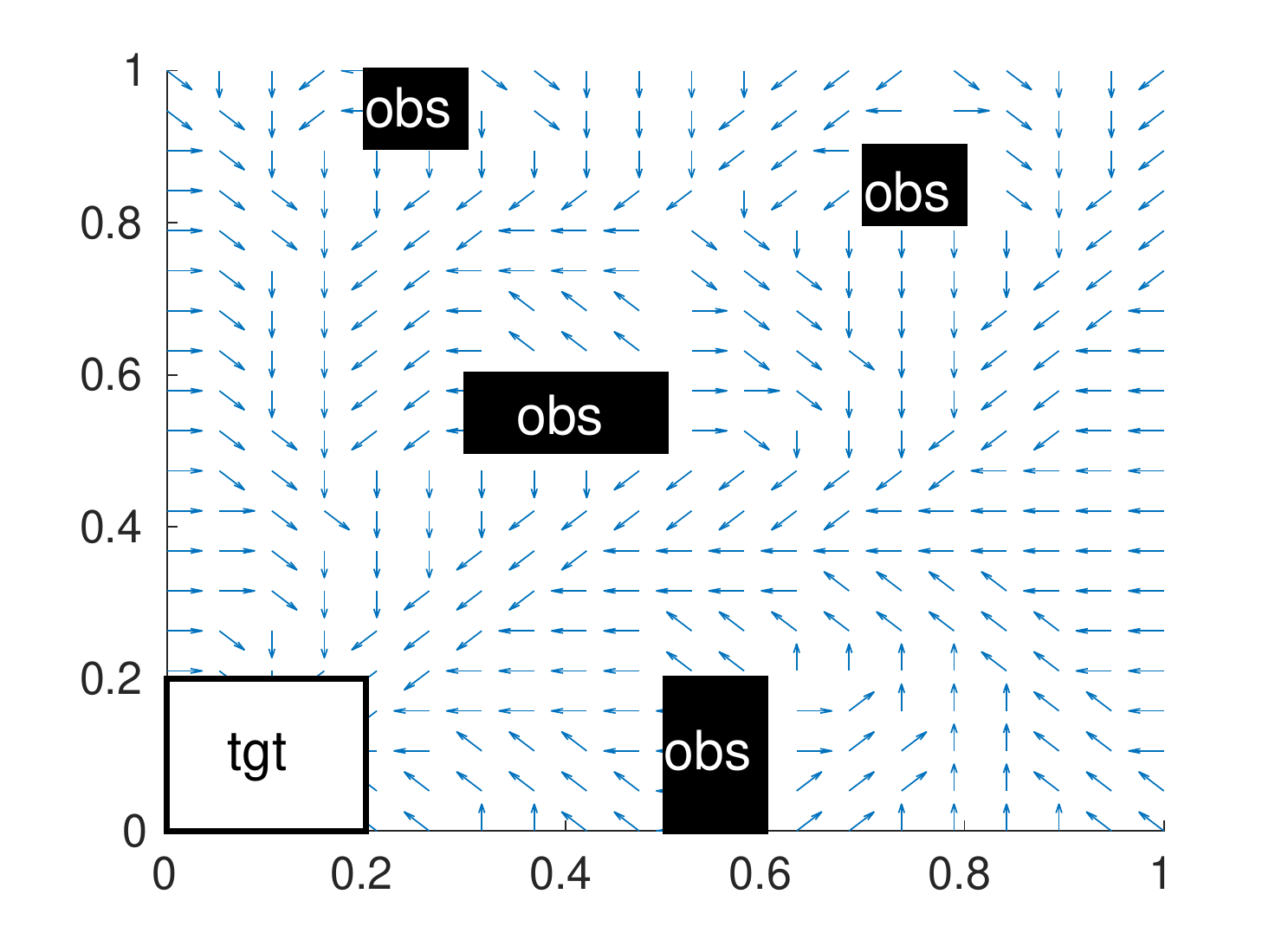}
    \caption{$\sigma^*$ for $\#2$}
\label{fig:strategy_unicycle_R_6}
\end{subfigure}%
  \caption{Results of Experiments $\#1-\#6$ in Table \ref{tab:results}. In~\ref{fig:P_reach_lower_unicycle_R_6}, \ref{fig:P_reach_lower_unicycle_R_5}, and \ref{fig:P_reach_lower_unicycle_R_7}, lower bound in the probability of reachability corresponding to Experiments $\#1$ (and $\#2$), $\#3$ (and $\#4$) and $\#5$ (and $\#6$), respectively. The plotted trajectories correspond to Monte Carlo simulations of System \eqref{eq:unicycle_model} taking samples from a distribution $\widetilde p_v\in\mathcal{P}_v$. In~\ref{fig:P_reach_upper_unicycle2D}, upper bound in the probability of reachability corresponding to Experiments $\#1-\#6$. In \ref{fig:strategy_unicycle_R_6}, optimal strategy of Experiments $\#1$ (and $\#2$).}
  \label{fig:P_reach_lower_unicycle_R}
\end{figure}
%
The bounds on the reachability probability for Experiments $\#1-\#6$ are shown in Figure~\ref{fig:P_reach_lower_unicycle_R}, together with the vector field of the system in close loop with the optimal strategy. Since the obtained upper bound in the reachability probability is the same across Experiments $\#1-\#6$ and the optimal strategy is practically the same, we only show these results in the case of Experiment $\#1$ (equivalently $\#2$).
%

To compare our approach with the one that relies on IMDP abstractions described in Remark~\ref{rem:model:choice}, we present the results obtained with the latter as Experiments $\#7$ and $\#8$ in Table \ref{tab:results}. 
Observe that
the average error between the bounds in reachability corresponding to Experiments $\#7$ and $\#8$ is way higher than the one corresponding to Experiments $\#1$ and $\#2$ despite the fact that $\varepsilon$ is the same. This result highlights how our proposed approach, despite requiring a larger amount  of time to perform strategy synthesis, is able to provide non-trivial satisfaction guarantees, unlike the approach from Remark \ref{rem:model:choice}. Furthermore, our approach is able to synthesize a strategy that satisfies the reachability task, in contrast to the approach based on the IMDP abstraction. We compare the obtained lower bounds on the reachability probability for experiments $\#1$ (equivalently $\#2$) and $\#7$ (equivalently $\#8$) in Figure \ref{fig:comparison}.
\begin{figure}[h!]
\centering
  \begin{subfigure}[t]{0.4\textwidth}
    \centering
    \includegraphics[width=\textwidth]{unicycle_R_6_reach_lower-eps-converted-to.pdf}
    \caption{Robust MDP abstraction approach}
\label{fig:P_reach_lower_unicycle_R_6_comparison}
  \end{subfigure}
  \hspace{0.1\textwidth}
  \begin{subfigure}[t]{0.4\textwidth}
    \centering
    \includegraphics[width=\textwidth]{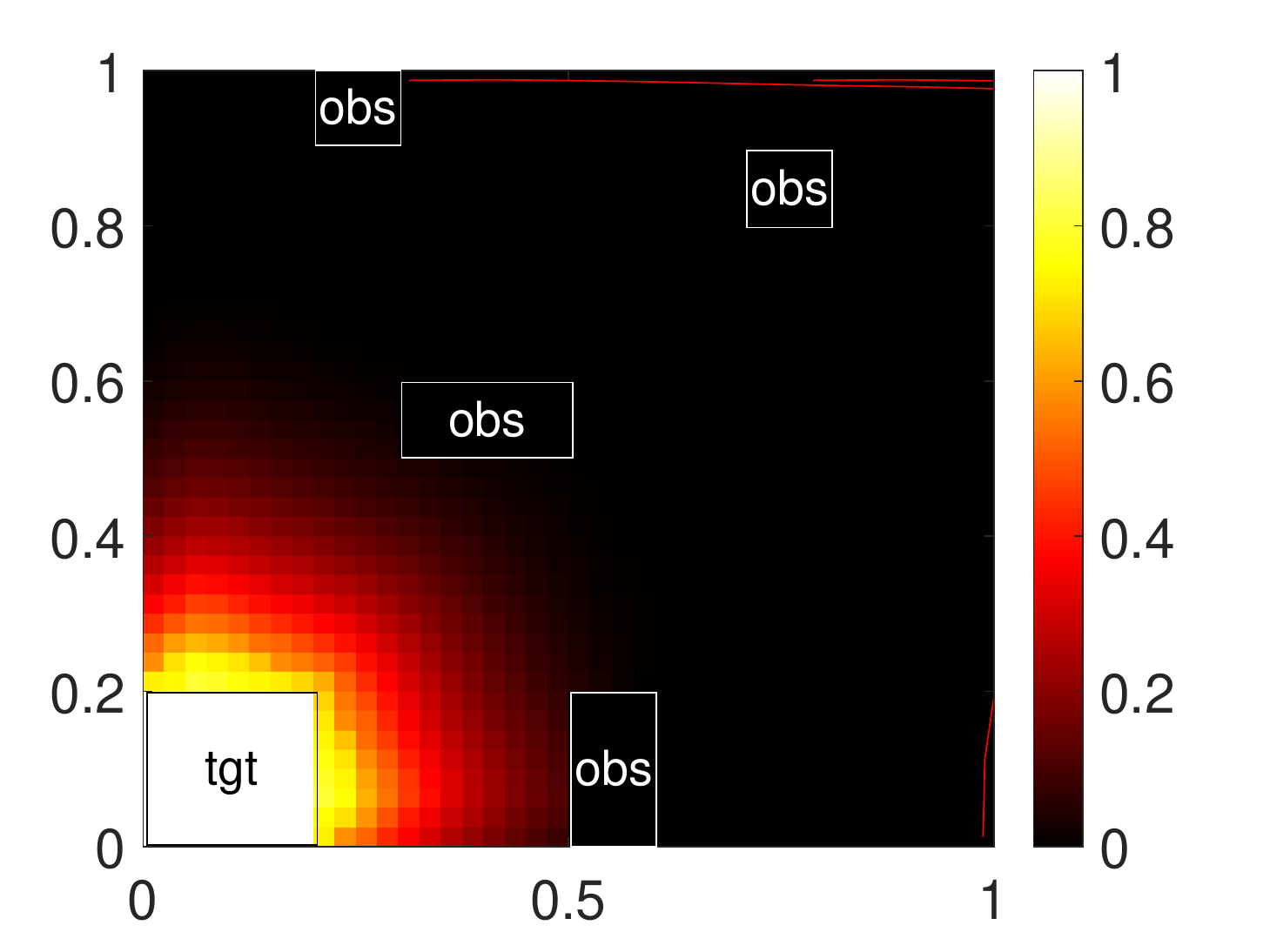}
    \caption{IMDP abstraction approach}
\label{fig:P_reach_lower_unicycle_26_comparison}
\end{subfigure}%
  \caption{Results of experiments $\#1$ and $\#7$ in Table \ref{tab:results}. Lower bound in the probability of reachability. The trajectories in both figures correspond to Monte Carlo simulations taking samples from a distribution $\widetilde p_v\in\mathcal{P}_v$. The ones that satisfy the specification are presented in blue, while the ones that do not are presented in red.}
  \label{fig:comparison}
\end{figure}

Note that Experiments $\#7$ and $\#8$ only differ in the algorithm used for strategy synthesis: whereas in Experiment $\#7$ Linprog was used, a dedicated algorithm was employed in $\#8$. The purpose of showing both results is to compare the computational improvement that stems from employing a dedicated synthesis algorithm for robust MDPs with this same benefit in the case of IMDPs.

\subsection{Nonlinear System with 4 Modes}
\label{sec:Jacksonnl}

For the second case study we consider the nonlinear system from  \cite{jackson2021strategy,adams2022formal} with dynamics:
\begin{equation}
\label{eq:jacksonnl}
    x_{k+1} = x_k + \widetilde f_{u_k}(x_k) + \noise_k.
\end{equation}
Denoting by $x^{(i)}$ the $i$-th component of the state, the map $ \widetilde f_{u_k}$ is given by
\begin{align}
\label{eq:jacksonnl_modes}
    \widetilde f_u(x) = \begin{cases}
    [0.5 + 0.2\sin(x^{(2)}), 0.4\cos(x^{(1)})]^T \qquad &\text{if}\:u = 1\\
    [-0.5 + 0.2\sin(x^{(2)}), 0.4\cos(x^{(1)})]^T \qquad &\text{if}\:u = 2\\
    [0.4\cos(x^{(2)}), 0.5 + 0.2\sin(x^{(1)})]^T \qquad &\text{if}\:u = 3\\
    [0.4\cos(x^{(2)}), -0.5 + 0.2\sin(x^{(1)})]^T \qquad &\text{if}\:u = 4.
    \end{cases}
\end{align}
The system has state $x_k\in\mathbb{R}^2$ and its control input $u_k$ switches between the discrete values 1, 2 , 3, and  4. The safe set is shown in Figure~\ref{fig:strategy_Jackson_R_4}.
We discretize the rectangle $[-2, 2]^2$ into a uniform grid, which results in abstraction with $N = 1601$ states. The ambiguity set is centered at an empirical distribution of $M = 20$ samples, which are again drawn from a Gaussian mixture with two components, centered at $[-0.05 ,0]$ and $[0.05 ,0]$, respectively, with covariance matrix ${\rm diag}(4\times 10^{-4},4\times 10^{-4})$. The centers of both components are separated by a distance close to the size of the state discretization. 
The obtained results are shown in Figure \ref{fig:Results_JacksonNL} and in Table~\ref{tab:results} in the row corresponding to Experiments $\#9$ and $\#10$. As it can be observed, our approach is able to synthesize a robust control strategy that yields high satisfaction probabilities in a big portion of the state space, despite the size of the obstacles and the system being nonlinear.
\begin{figure}[h!] 
  \begin{subfigure}[t]{0.33\textwidth}
    \centering
    \includegraphics[width=\textwidth]{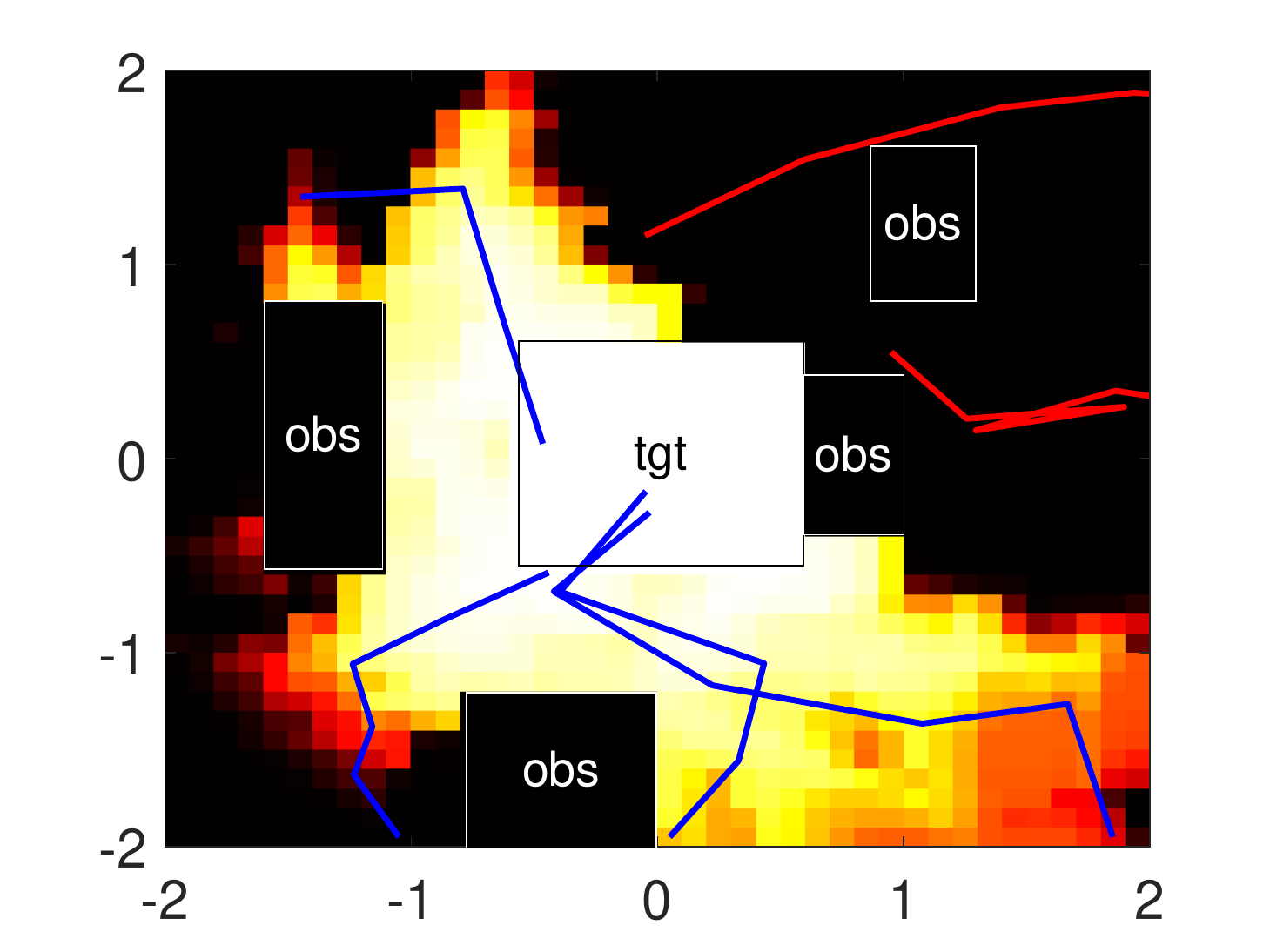}
    \caption{Lower bound on the  \\ \phantom{\quad\:\:} reachability 
    probability}
\label{fig:P_reach_lower_Jackson_R_4}
  \end{subfigure}
  \begin{subfigure}[t]{0.33\textwidth}
    \centering
    \includegraphics[width=\textwidth]{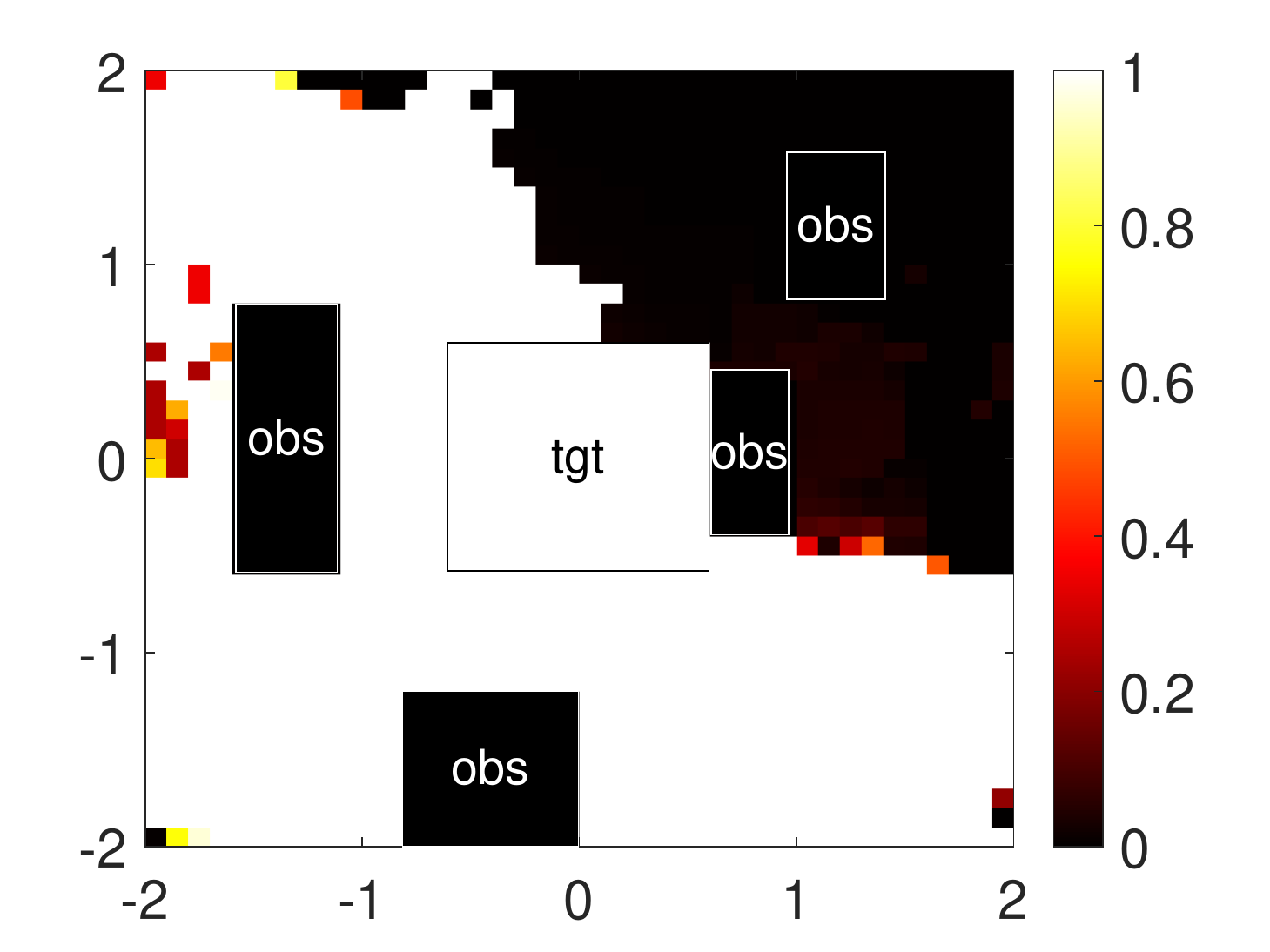}
    \caption{Upper bound on the  \\ \phantom{\quad\:\:} reachability  
    probability}
\label{fig:P_reach_upper_Jackson_R_4}
  \end{subfigure} 
  \begin{subfigure}[t]{0.33\textwidth}
    \centering
    \includegraphics[width=\textwidth]{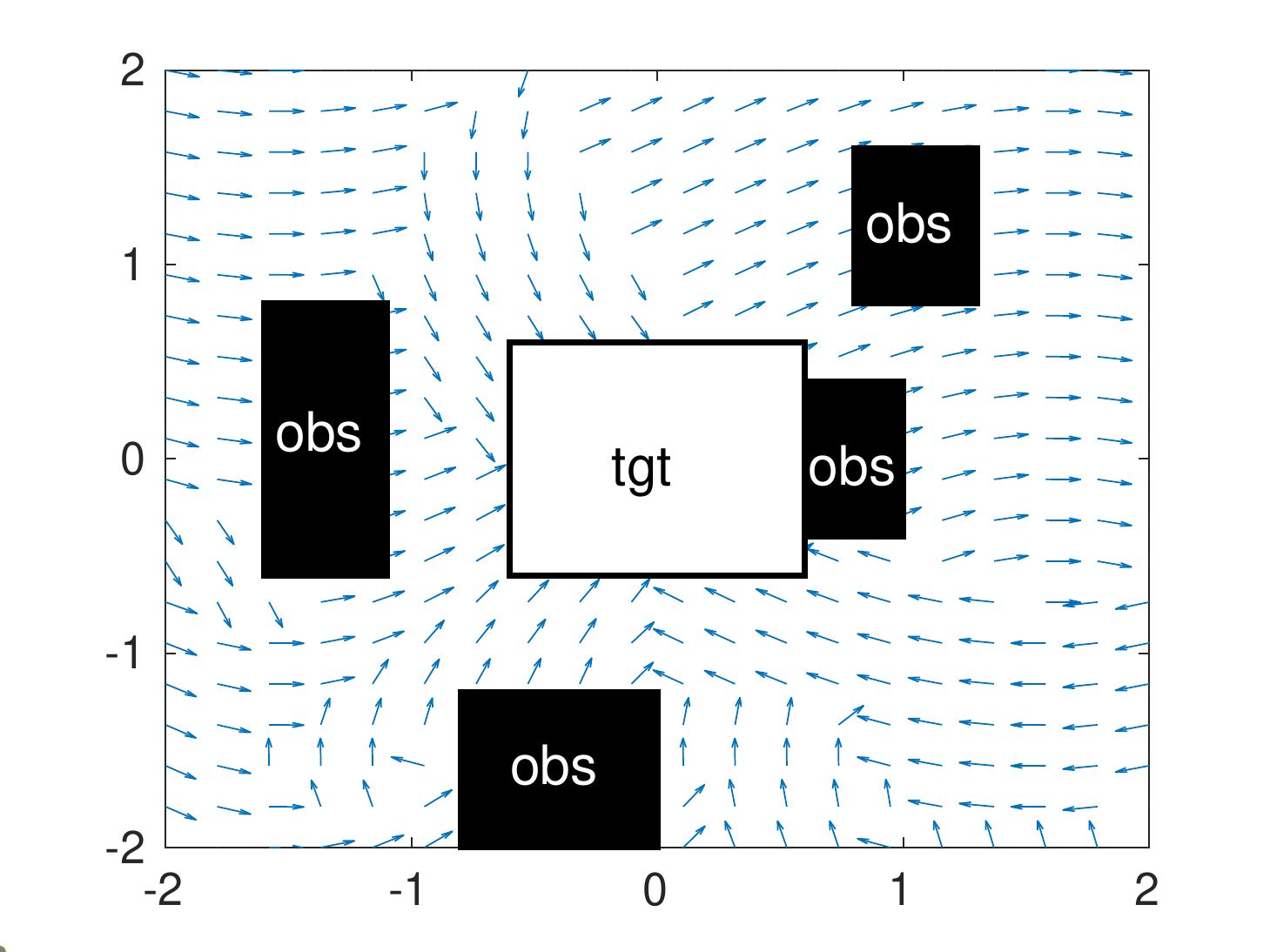}
    \caption{Optimal strategy}
\label{fig:strategy_Jackson_R_4}
  \end{subfigure}%
  \caption{Results of experiment $\#9$ (and $\#10$). The trajectories in Figure \ref{fig:P_reach_lower_Jackson_R_4} correspond to Monte Carlo simulations taking samples from a distribution $\widetilde p_v\in\mathcal{P}_v$. The ones that satisfy the specification are presented in blue, while the ones that do not are presented in red.}
  \label{fig:Results_JacksonNL}
\end{figure}

\subsection{Switched Linear Systems with 5 Modes}
\label{sec:switched}

The third case study is a system that switches between $5$ linear modes. Its dynamics are of the form
\begin{align}
    \label{eq:switched_linear2D}
    x_{k+1} = A_{u}x_k + \noise_k,
\end{align}
where $U = \{1,\cdots,5\}$ and
\begin{align*}
    A_1 &= \begin{pmatrix}
        0.79 & 0.035\\
        0 & 0.825
    \end{pmatrix},\:
    A_2 = \begin{pmatrix}
        0.79 & 0.175\\
        0 & 0.825
    \end{pmatrix},\:
    A_3 = \begin{pmatrix}
        0.79 & 0\\
        0.175 & 0.825
    \end{pmatrix},\\
    A_4 &= \begin{pmatrix}
        1 & 0.2\\
        -0.2 & 1
    \end{pmatrix},\:
    A_5 = \begin{pmatrix}
        1 & -0.2\\
        0.2 & 1
    \end{pmatrix}.
\end{align*}
The state space is defined in the same way as in the previous case studies, with the safe set as shown in Figure~\ref{fig:strategy_switched_linear2D}. The rectangle $[-2,2]^2$ is discretized via a uniform grid, resulting in $N = 3601$ states. Unlike in the previous examples, the motivation for this case study is no longer a data-driven setting, but one of robustness against distributional shifts: the nominal distribution $\widehat\distribution_v$ is known to be a Gaussian of zero mean and covariance ${\rm diag}(9\times 10^{-4},9\times 10^{-4})$, but we want to be robust against possible changes in this distribution. To this end, we pick an ambiguity set defined by $\varepsilon = 0.0127$, with $s=2$ and $d$ being the $2$-norm. This set contains, but is not limited to, Gaussian distributions with a covariance $1.7$ times bigger than the nominal one, or off-diagonal terms bigger than $1/2$ times the diagonal entries. Note that, while we truncate the $\distribution_v$ to $4$ standard deviations, the ambiguity set contains also unbounded distributions.

In this case, we obtained the abstraction by making use of the approach in \cite{adams2022formal}. The results correspond to Experiment $\#11$ in Table~\ref{tab:results}.
The lower and upper bounds on the reachability probabilities are shown in Figures~\ref{fig:P_reach_lower1_switched_linear2D} and \ref{fig:P_reach_upper1_switched_linear2D}. Additionally, the vector field of the system in closed loop with the optimal strategy is shown in Figure \ref{fig:strategy_switched_linear2D}.
\begin{figure}[ht] 
  \begin{subfigure}[t]{0.33\textwidth}
    \centering
    \includegraphics[width=\textwidth]{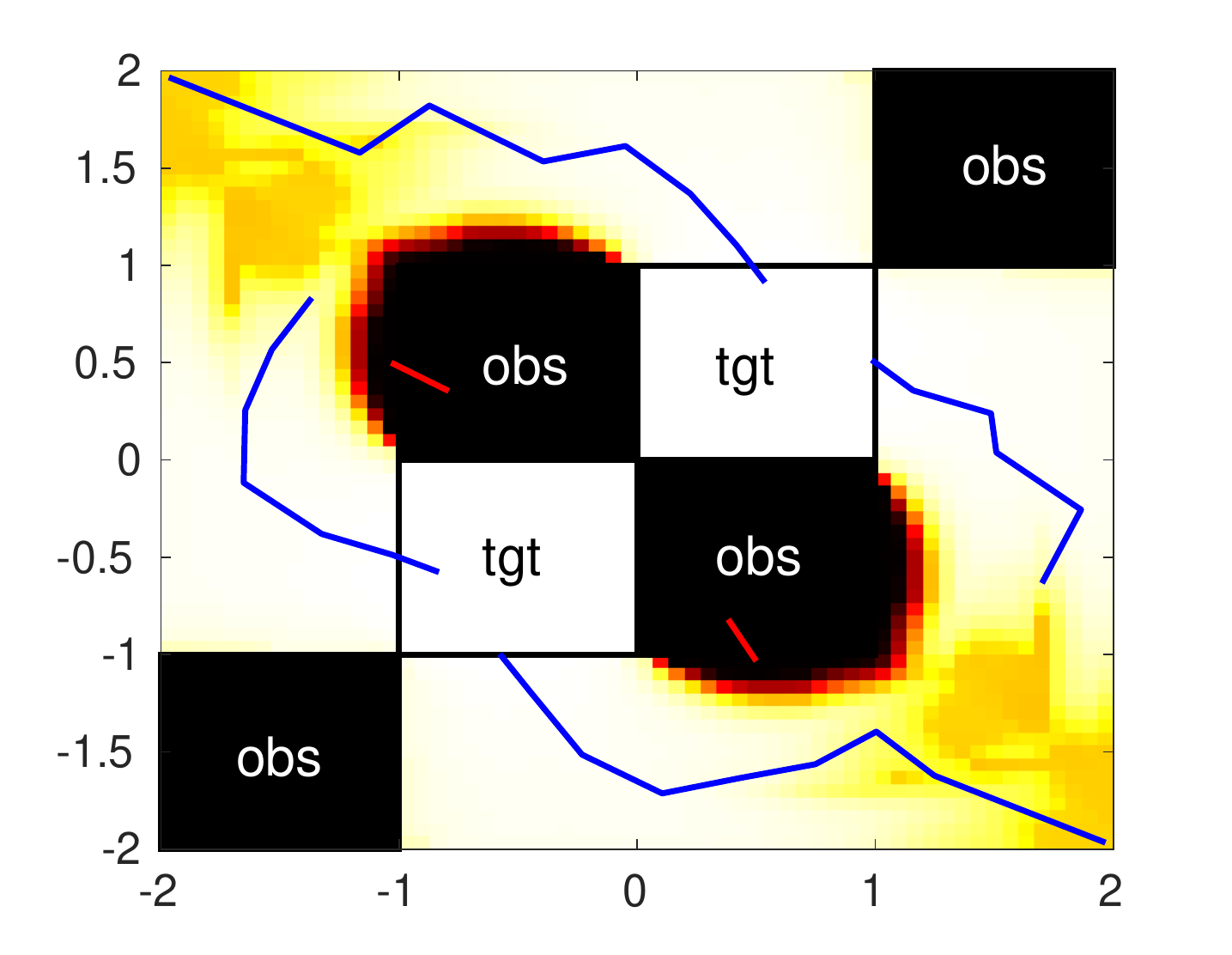}
    \caption{Lower bound on the  \\ \phantom{\quad\:\:} reachability 
    probability}
\label{fig:P_reach_lower1_switched_linear2D}
  \end{subfigure}
  \begin{subfigure}[t]{0.35\textwidth}
    \centering
    \includegraphics[width=\textwidth]{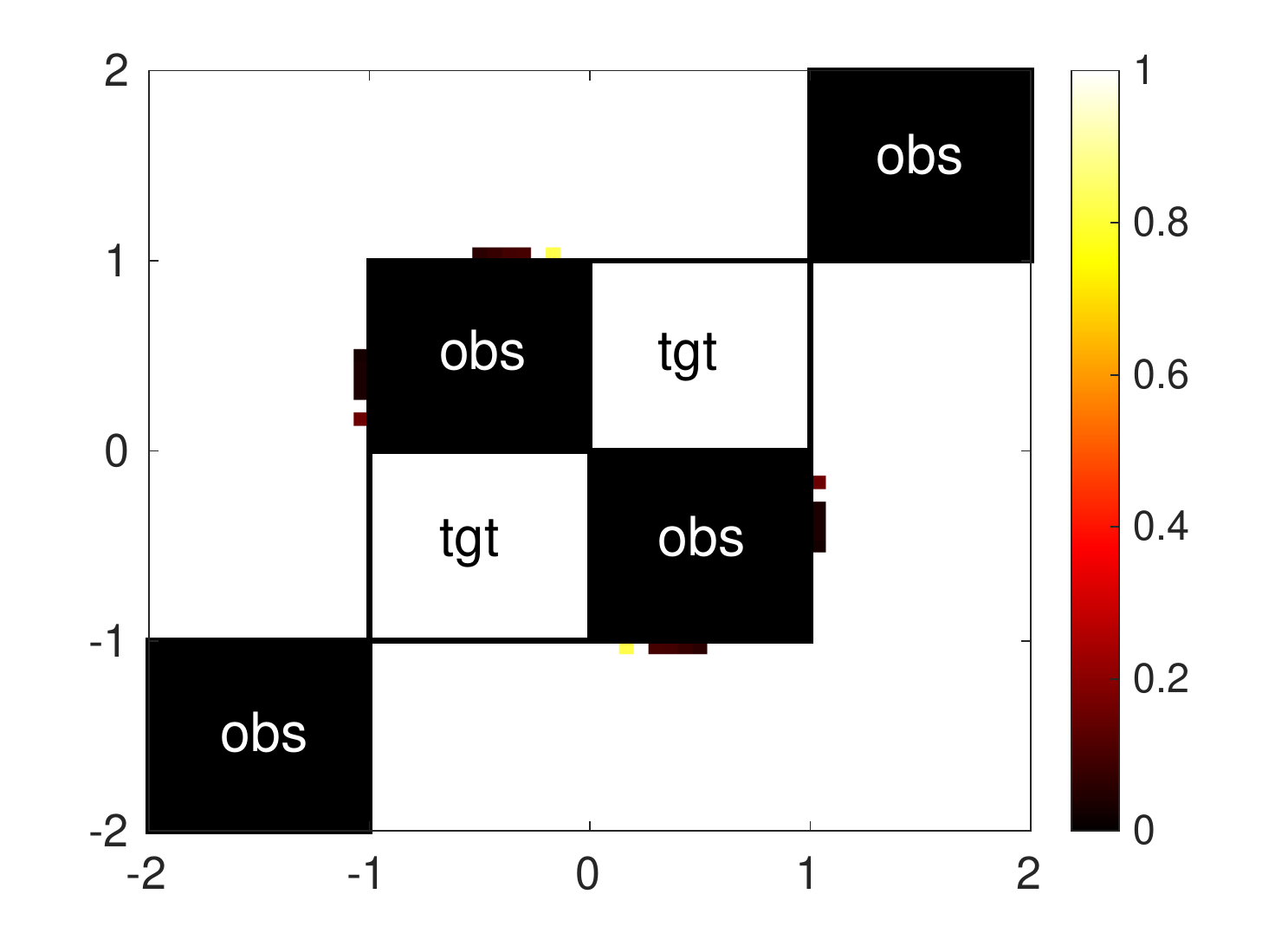}
    \caption{Upper bound on the  \\ \phantom{\quad\:\:} reachability  
    probability}
\label{fig:P_reach_upper1_switched_linear2D}
  \end{subfigure}%
  \begin{subfigure}[t]{0.35\textwidth}
    \centering
    \includegraphics[width=\textwidth]{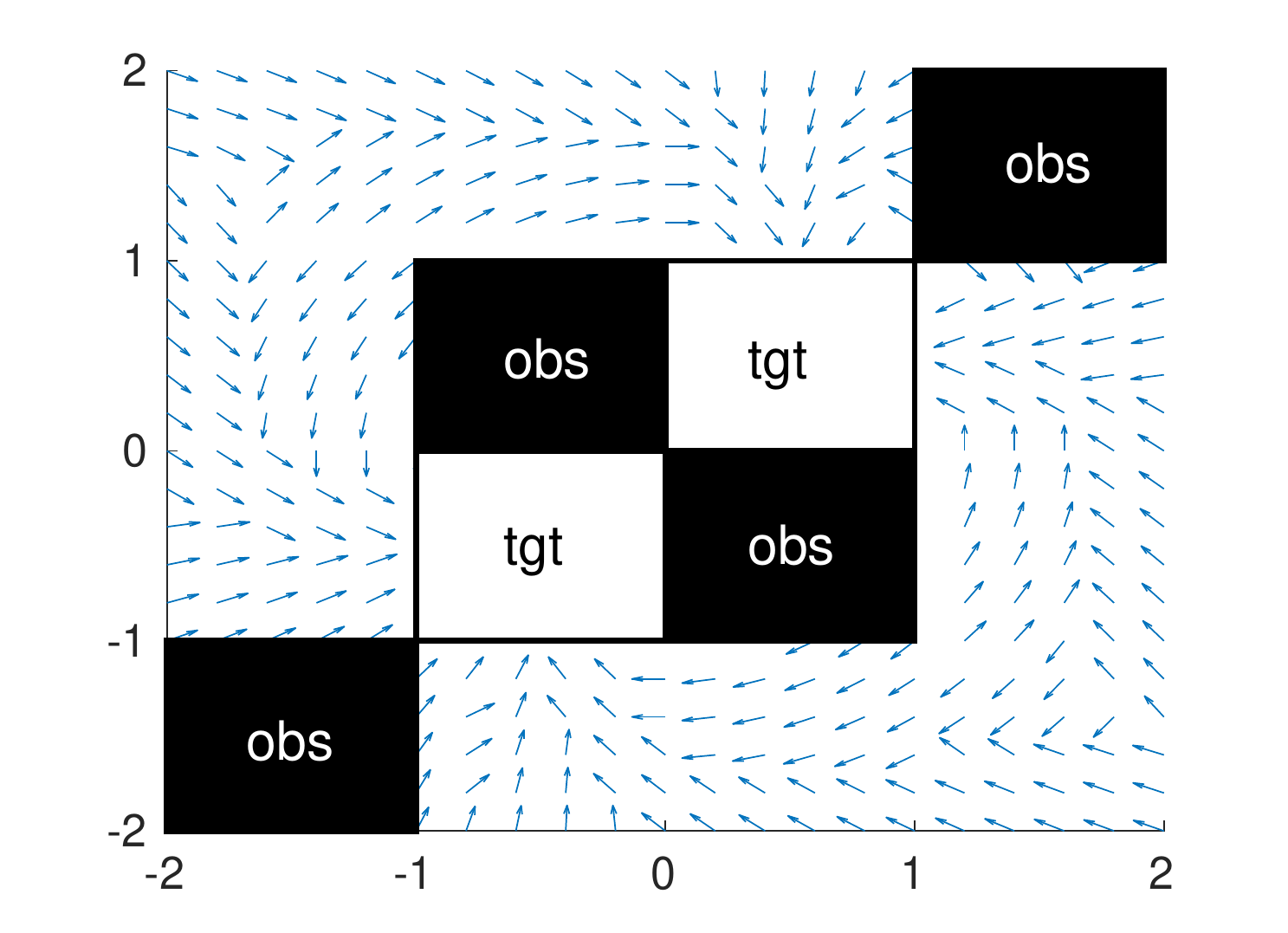}
    \caption{Optimal strategy}
\label{fig:strategy_switched_linear2D}
  \end{subfigure}%
  \caption{Results of Experiment $\#11$. The trajectories in Figure \ref{fig:P_reach_lower1_switched_linear2D} correspond to Monte Carlo simulations taking samples from a Gaussian distribution $\widetilde p_v\in\mathcal{P}_v$ at the boundary of the ambiguity set. The ones that satisfy the specification are presented in blue, while the ones that do not are presented in red.}
  \label{fig:Results_switched_linear2D}
\end{figure}
The empirical results obtained via Monte Carlo simulations are always found within the bounds in the probability of reachability. Furthermore, the synthesized strategies achieve an average (across $1000$ random initial conditions) empirical reachability probability of $\approx 1 $.

To show that the bounds in the probability of reachability are sound even when the empirical probability of reachability is smaller than $1$, we increase the variance of the noise. The new nominal noise distribution has covariance ${\rm diag}(0.04,0.04)$ truncated at $3$ standard deviations and $\varepsilon =1.3\times 10^{-3}$. The results correspond to Experiment $\#12$ in Table~\ref{tab:results}.
In this case, the average empirical probability of reachability was $0.802$, a non-trivial result, and the empirical reachability probabilities for each initial condition were always found within the theoretical bounds. Note that the synthesis time corresponding to this case study is the longest among the cases in which our dual approach was used. The cause is that the variance of the nominal noise distribution is large. Therefore its support, corresponding to $3$ standard deviations, is way bigger than in all the other experiments, which increases the computational complexity of the synthesis process. The average error $e_{\rm avg}$ is also noticeable bigger, due to the higher-variance noise. We plot the bounds in the probability of reachability together with the vector field of the closed-loop system in Figure \ref{fig:Results_switched_linear2D_2}.
\begin{figure}[ht] 
  \begin{subfigure}[t]{0.33\textwidth}
    \centering
    \includegraphics[width=\textwidth]{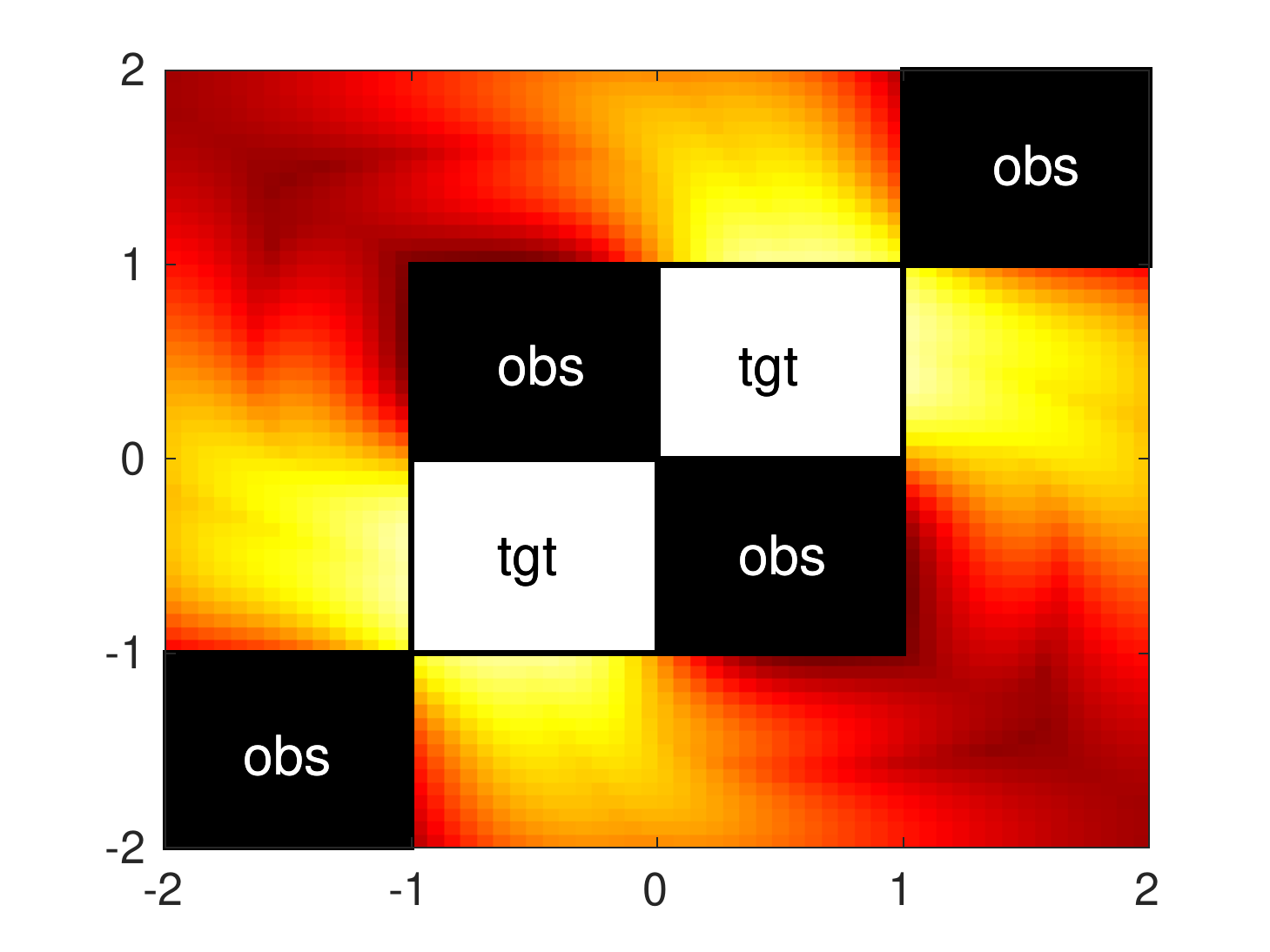}
    \caption{Lower bound on the  \\ \phantom{\quad\:\:} reachability 
    probability}
\label{fig:P_reach_lower2_switched_linear2D}
  \end{subfigure}
  \begin{subfigure}[t]{0.35\textwidth}
    \centering
    \includegraphics[width=\textwidth]{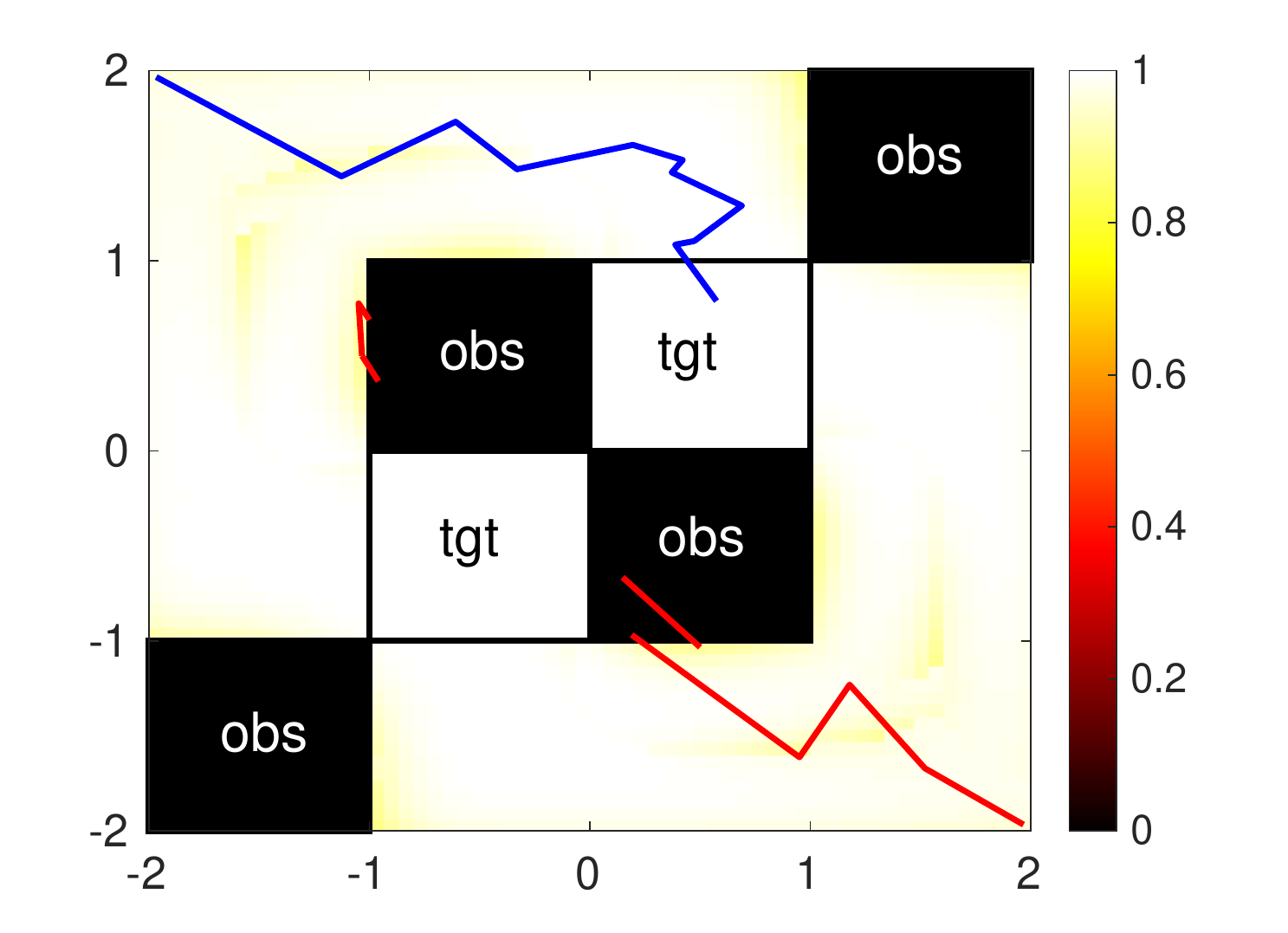}
    \caption{Upper bound on the  \\ \phantom{\quad\:\:} reachability  
    probability}
\label{fig:P_reach_upper2_switched_linear2D}
  \end{subfigure}%
  \begin{subfigure}[t]{0.35\textwidth}
    \centering
    \includegraphics[width=\textwidth]{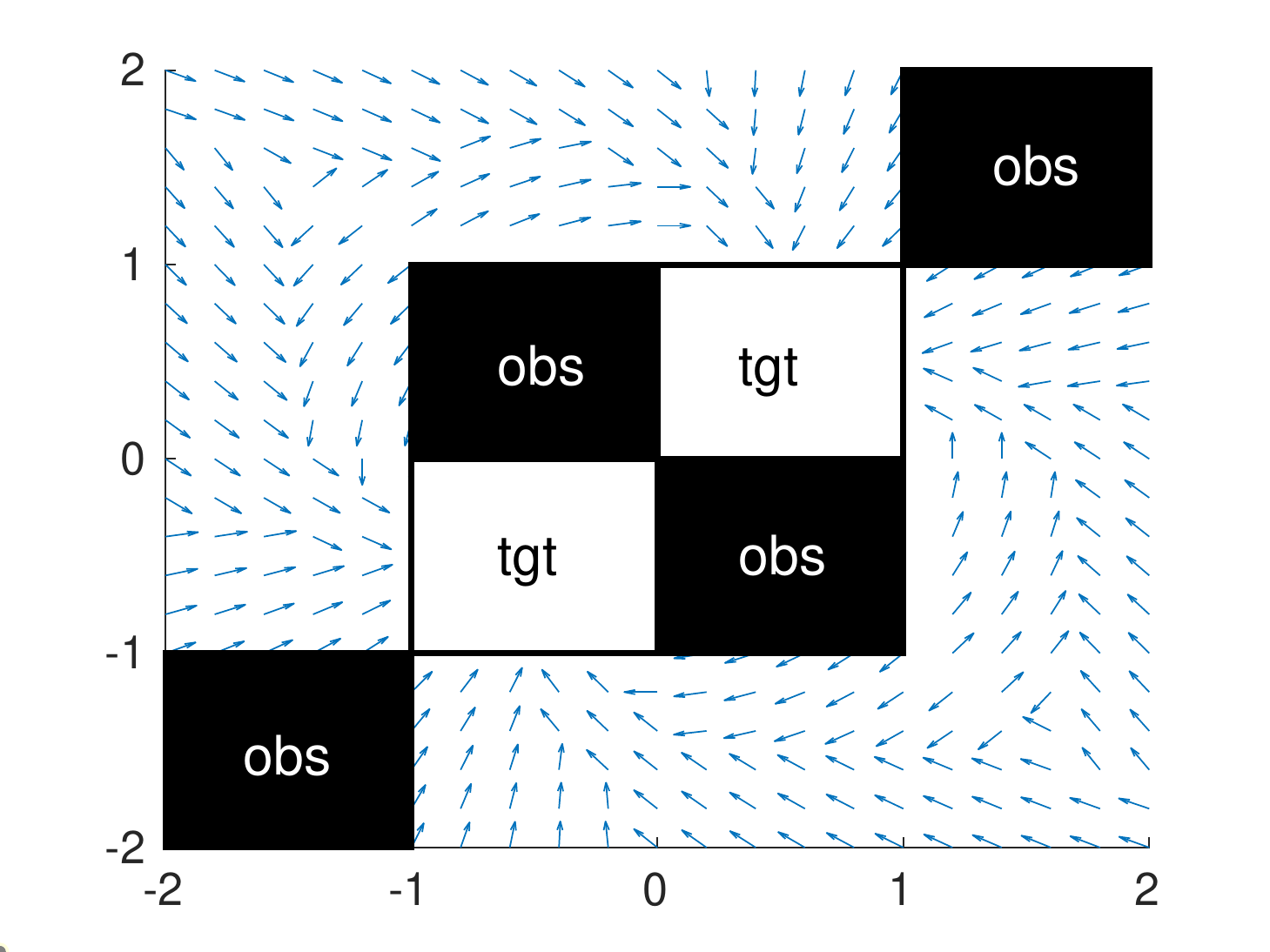}
    \caption{Optimal strategy}
\label{fig:strategy_switched_linear2D_2}
  \end{subfigure}%
  \caption{Results of Experiment $\#12$. The trajectories in Figure \ref{fig:P_reach_upper2_switched_linear2D} correspond to Monte Carlo simulations taking samples from a Gaussian distribution $\widetilde p_v\in\mathcal{P}_v$ at the boundary of the ambiguity set. The ones that satisfy the specification are presented in blue, while the ones that do not are presented in red.}
  \label{fig:Results_switched_linear2D_2}
\end{figure}

In both Experiments $\#11$ and $\#12$, the optimal stationary strategy was found via the approach proposed in Theorem~\ref{thm:optimal_strategy_unbounded_horizon}, which took only a few seconds.

\subsection{Complex Specifications}
\label{sec:LTLf}

In this section we showcase how our approach to solve unbounded reachability problems allows us to synthesize strategies that yield high probability of satisfying more complex specifications. Specifically, we consider the same setting as in case study Experiment $\#2$, together with the following $\text{LTL}_f$ specification, borrowed from \cite{vazquez2018learning}: the system must reach a charge station while remaining safe and, if the system goes through a region with water, then it should first dry in a carpet before charging.
The results are shown as Experiment $\#13$ in Table~\ref{tab:results}. When compared with Experiment $\#2$, the synthesis time is significantly higher. This happens due to the fact that the specification is more complex than just reachability. As a consequence, strategy synthesis requires solving a reachability problem on a bigger robust MDP, obtained by combining the abstraction with the \emph{automaton} that represents the specification (for more details see \cite{wolff2012robust}). The optimal stationary strategy was extracted via the approach proposed in Theorem~\ref{thm:optimal_strategy_unbounded_horizon} which, again, took only a few seconds.

Additionally, the bounds in the satisfaction probabilities are shown in Figures \ref{fig:unicycle2D_LTLf_bounds} and the vector field of the system in closed loop with the optimal strategy is plotted in \eqref{fig:unicycle2D_LTLf_strategy1}. Note that the strategy is composed by two stationary strategies, as can be observed in Figure~\ref{fig:unicycle2D_LTLf_strategy1}. The system follows the strategy ``go to charge'' when it is dry, and ``go to carpet'' when it is wet.
\begin{figure}[h!]
\centering
  \begin{subfigure}[t]{0.4\textwidth}
    \centering
    \includegraphics[width=\textwidth]{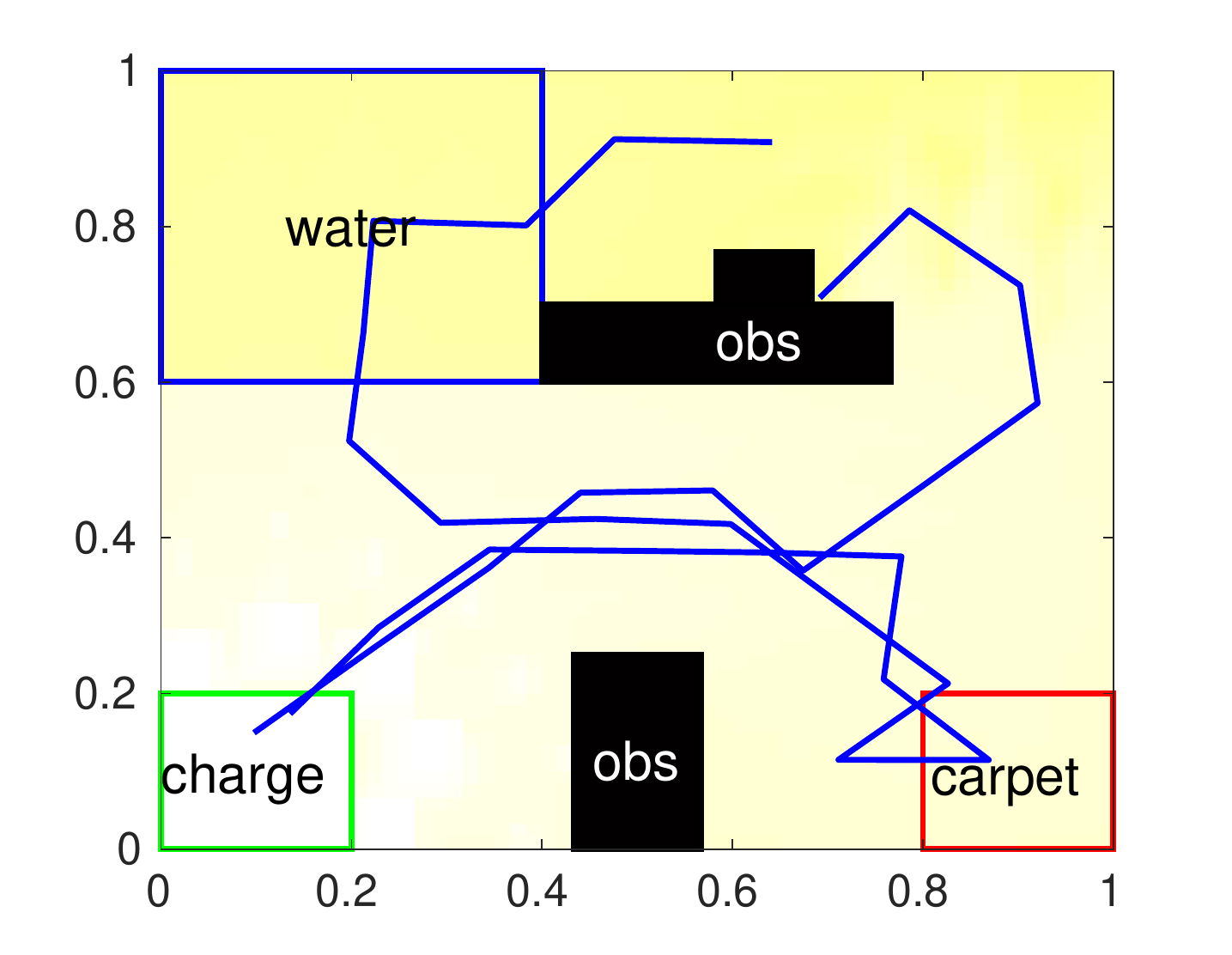}
    \caption{Lower bound on the reachability probability}
\label{fig:unicycle2D_LTLf_P_reach_lower1}
  \end{subfigure}
  \hspace{0.1\textwidth}
  \begin{subfigure}[t]{0.42\textwidth}
    \centering
    \includegraphics[width=\textwidth]{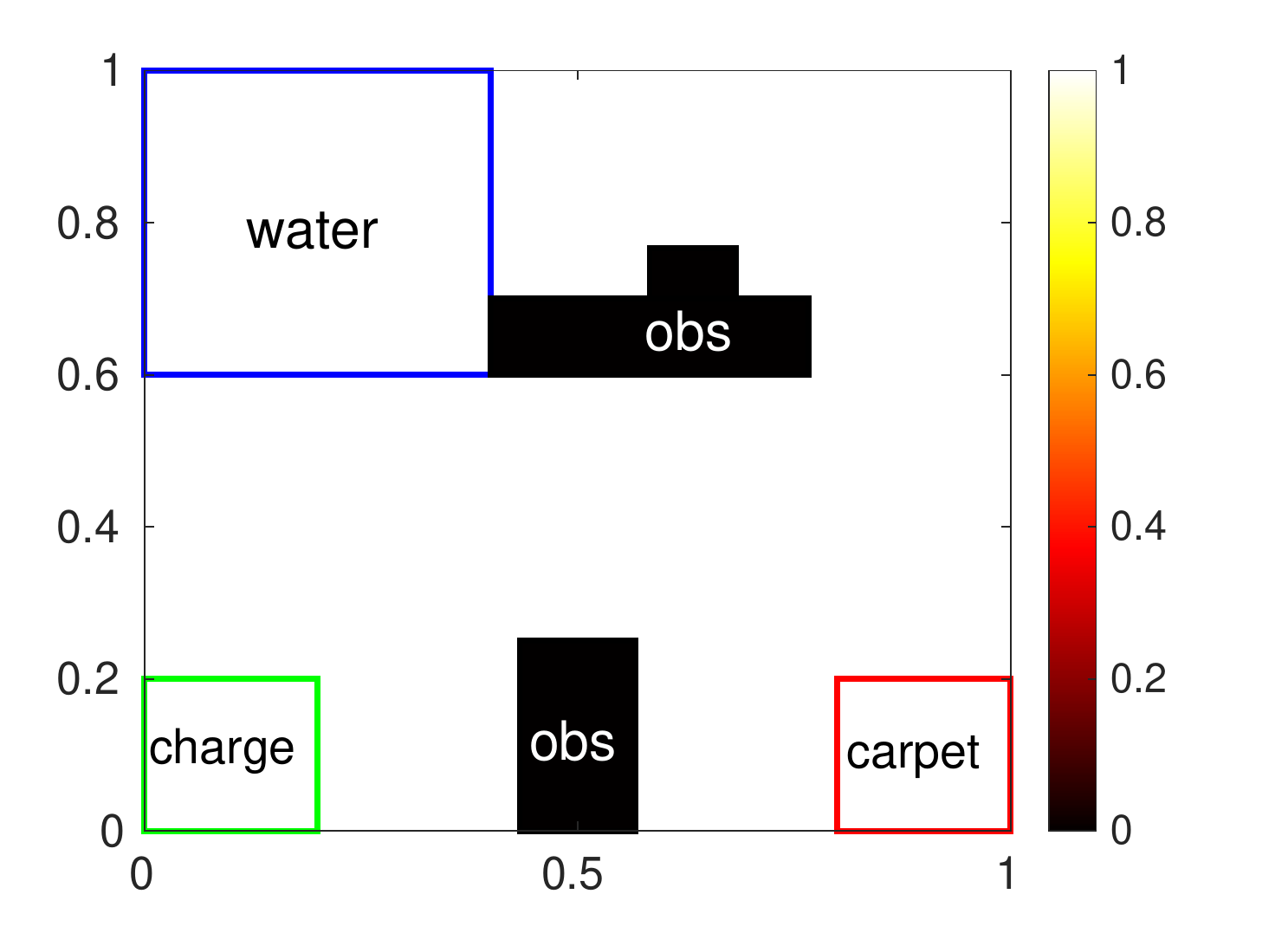}
    \caption{Upper bound on the reachability probability}
\label{fig:unicycle2D_LTLf_P_reach_upper1}
\end{subfigure}%
  \caption{Bounds in the probability of reachability for Experiment $\#13$. The trajectories in blue correspond to Monte Carlo simulations taking samples from a distribution $\widetilde p_v\in\mathcal{P}_v$.}
  \label{fig:unicycle2D_LTLf_bounds}
\end{figure}
\begin{figure}[h!]
\centering
  \begin{subfigure}[t]{0.4\textwidth}
    \centering
    \includegraphics[width=\textwidth]{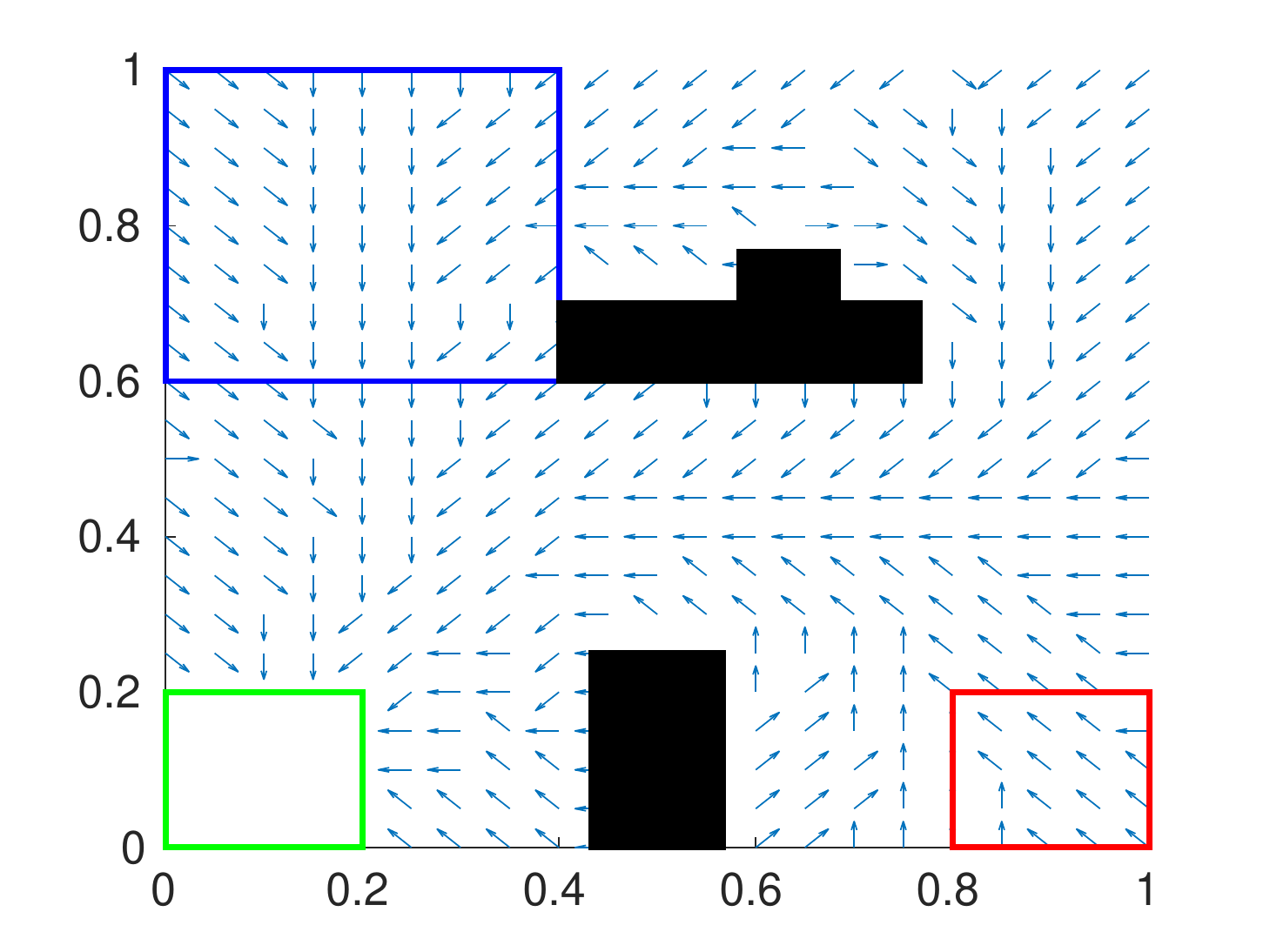}
    \caption{Mode ``go to charge''}
\label{fig:unicycle2D_LTLf_strategy_mode1_1}
  \end{subfigure}
  \hspace{0.1\textwidth}
  \begin{subfigure}[t]{0.4\textwidth}
    \centering
    \includegraphics[width=\textwidth]{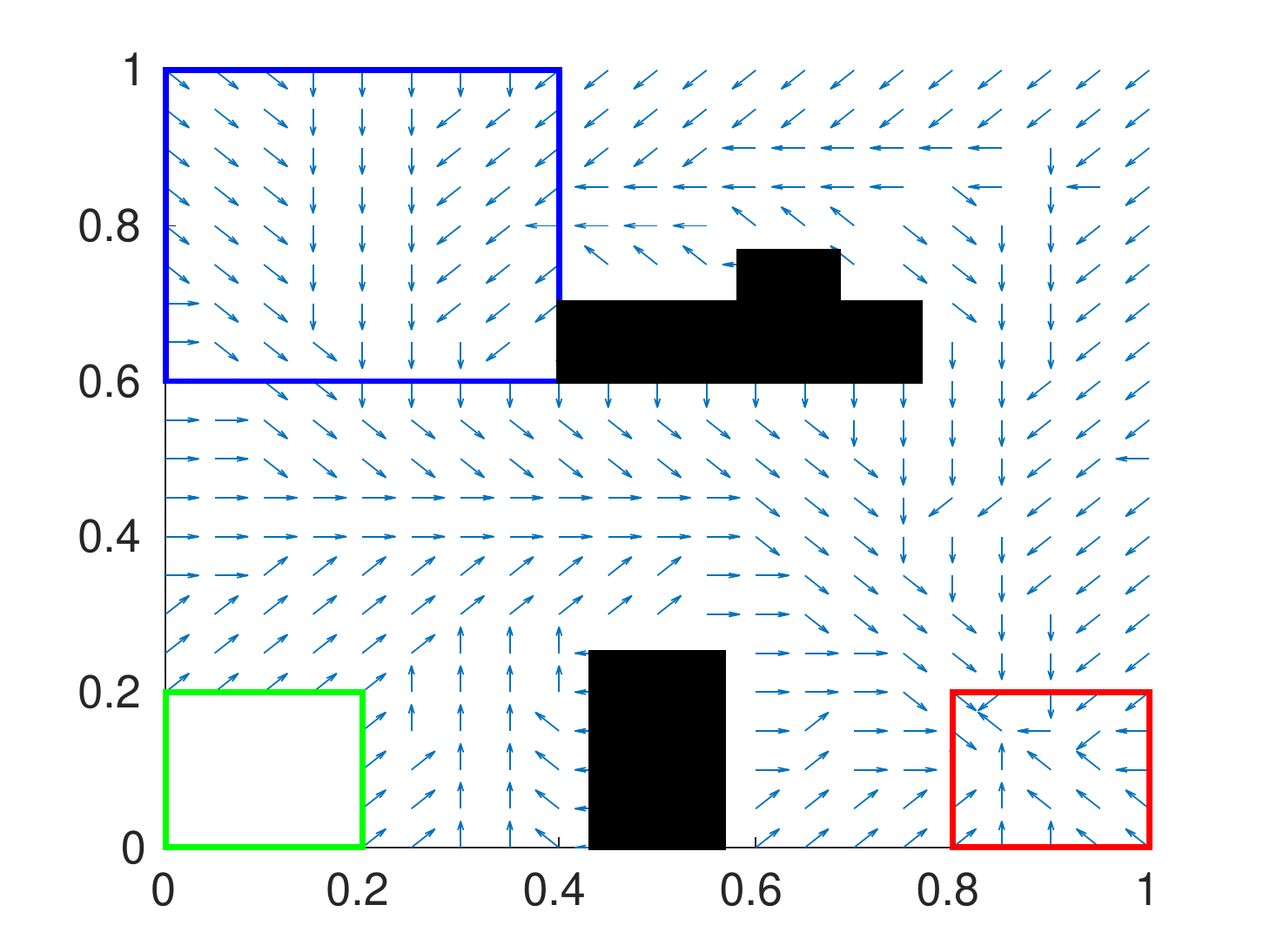}
    \caption{Mode ``go to carpet''}
\label{fig:unicycle2D_LTLf_strategy_mode3_1}
\end{subfigure}%
\caption{Optimal strategy for the case of Experiment $\#13$.}
  \label{fig:unicycle2D_LTLf_strategy1}
\end{figure}
\section{Conclusion}


In this paper we presented a framework for the formal control of switched 
stochastic systems with additive, random disturbances whose probability 
distribution belongs to a Wasserstein ambiguity set. To this end, we derived a 
robust MDP abstraction of the original system and proposed an algorithm, 
termed  robust dynamic programming, to synthesize robust strategies that maximize 
the probability of satisfying a (finite or unbounded time horizon) reach-avoid specification. The obtained results 
demonstrate the effectiveness of our approach in systems with both linear and 
nonlinear dynamics, and even under complex $\text{LTL}_f$ specifications. Our results also show the superiority of our abstraction approach with respect to leveraging directly IMDP abstractions, and the computational advantage of our synthesis algorithm with respect to using off-the-shelf linear programming solvers.

\appendix
\section{Proof of Proposition \ref{prop:consistency}}
\label{app:proof_consistency}

To prove the proposition we will use two technical lemmas that link couplings and optimal transport discrepancies in the continuous and abstract space.
\begin{lemma}[Induced Coupling on the Discrete Space]
\label{lemma:coupling_equivalence}
Consider the coupling $\pi\in\mathcal{P}(\reals^n\times\reals^n)$ with marginals $p$ and $p'$, a finite (measurable) partition $\Qimdp$ of $\reals^n$, and the induced distributions $\gamma,\gamma'\in \mathcal D(Q)$ with $\gamma(q):=p(q)$ and $\gamma'(q):=p(q)$ for all $q\in Q$. Then $\nu\in\mathcal{P}(\Qmdp\times\Qmdp)$, defined by 
\begin{align} \label{nu:dfn}
    \nu(q,q') := \int_{q\times q'}d\pi(x,y)
\end{align}
is a coupling between $\gamma$ and $\gamma$. 
\end{lemma}
\begin{proof} 
The proof follows directly from the fact that 
\begin{align*}
\sum_{q'\in Q}\nu(q,q')= \sum_{q'\in Q}\int_{q\times q'}d\pi(x,y)
=\int_{q\times\reals^n}d\pi(x,y)=\gamma(q),
\end{align*}
and analogously for the other marginal. 
\end{proof}
Next, given two distributions on the continuous space $\reals^n$ we establish bounds on the optimal transport discrepancy $\mathcal T_c$ of their induced distributions on $Q$, based on their $s$-Wasserstein distance in the continuous space.
\begin{lemma}[Induced optimal transport discrepancy]
\label{lemma:relation_wasserstein_and_discrepancy}  
Let $p,p'\in\mathcal{D}_s(\reals^n)$ and consider the induced distributions $\gamma,\gamma'\in\mathcal{D}(Q)$ with $\gamma(\qmdp):=p(q)$ and $\gamma'(\qmdp):=p(q)$ for all $q\in \Qmdp$. Then for any $s\ge 1$ and $\varepsilon \ge 0$ it holds that 
\begin{align*}
    \mathcal{W}_s(p,p')\leq \varepsilon \Rightarrow \mathcal{T}_c(\gamma,\gamma')\leq \varepsilon^s,
\end{align*}
where $c$ is given in \eqref{eq:cost}.
\end{lemma}

\begin{proof}
Consider the map $J$ in \eqref{J:dfn} and note that due to \eqref{eq:cost},
\begin{align} \label{cost:inequality} 
\|x - y\|^s\geq c(J(x),J(y))
\end{align}
for all $x,y\in\mathbb{R}^n$. Let $\pi$ be an optimal coupling for the $s$-Wasserstein distance $\mathcal{W}_s(p,p')$ and $\nu$ be the induced coupling on $Q$ given by \eqref{nu:dfn}. Then we get from \eqref{eq:optimal_transport_discrepancy}, \eqref{J:dfn}, and  \eqref{cost:inequality}  that   
\begin{align*}
\mathcal{W}_s(p,p')^s & = \int_{\mathbb{R}^n\times\mathbb{R}^n}\|x - y\|^sd\pi(x,y) \ge \int_{\mathbb{R}^n\times\mathbb{R}^n}c(J(x),J(y))d\pi(x,y) \\
& = \sum_{\qmdp,\qmdp'\in\Qmdp}c(q,q')\int_{\qmdp\times\qmdp'}d\pi(x,y) = \sum_{\qmdp,\qmdp'\in\Qmdp}c(q,q')\nu(q,q') \geq \mathcal{T}_c(\gamma,\gamma'),
\end{align*}
which implies the result. The last inequality follows from \eqref{eq:optimal_transport_discrepancy} and   Lemma~\ref{lemma:coupling_equivalence}, which asserts that $\nu$ is a coupling between $\gamma$ and $\gamma'$. The proof is complete.
\end{proof}
\noindent The intuition behind Lemma \ref{lemma:relation_wasserstein_and_discrepancy} is the following: if the $s$-Wasserstein distance between two distributions in $\mathbb{R}^n$ is at most $\varepsilon$, then the optimal transport discrepancy (based on $c$) between their induced distributions on $\Qmdp$ is not more than $\varepsilon^s$.

\begin{proof}[Proof of Proposition~\ref{prop:consistency}] 
Let $q$, $a$, $x$, and $p_v$ as given in the statement and define  
\begin{align*}
\widehat\gamma_{x,a}(q'):=T_{\widehat p_v}^a(q'\mid x)
\end{align*}
for all $q'\in Q$. Then it follows from 
\eqref{eqn:transition_probability_bounds1} and \eqref{eq:set_transition_probabilities_nominal_IMDP} that 
\begin{align}\label{gamma:in:hat:Gamma}
\widehat\gamma_{x,a}\in\widehat\Gamma_{q,a}.
\end{align}
Next, we get from \eqref{transition:kernel} and Assumption~\ref{ass:1} that $T_{p_v}^a(\cdot\mid x)$ and $T_{\widehat p_v}^a(\cdot\mid x)$ are distributions in $\mathcal{D}_s(\reals^n)$ and that 
\begin{align*}
\mathcal W_s\Big(T_{p_v}^a(\cdot\mid x),T_{\widehat p_v}^a(\cdot\mid x)\Big)\le\varepsilon.   
\end{align*}
Since the induced distributions of $T_{p_v}^a(\cdot\mid x)$ and $T_{\widehat p_v}^a(\cdot\mid x)$ on $Q$ are $\gamma_{x,a}$ and $\widehat\gamma_{x,a}$, respectively, it follows from Lemma~\ref{lemma:relation_wasserstein_and_discrepancy} that  $\mathcal{T}_c(\gamma_{x,a},\widehat\gamma_{x,a})\leq \varepsilon^s\equiv\epsilon$, namely, $\gamma_{x,a}\in \mathcal{T}_c^\epsilon(\widehat\gamma_{x,a})$. Thus, we deduce from  \eqref{eq:set_transition_probabilities_robust_MDP} and \eqref{gamma:in:hat:Gamma} that $\gamma_{x,a}\in\Gamma_{q,a}$ and conclude the proof.
\end{proof}
\section{Proofs of Section~\ref{sec:unbounded_horizon}}
\label{app:proofs_section_unbounded_horizon}

We first introduce certain mappings and preliminary results that will used for the proofs of this section. Consider the discounted robust MDP with rewards $\widetilde{\mathcal{M}} := (Q, A, \widetilde{\Gamma}, r, \beta)$ as an extension of $\mathcal{M}$, where
\begin{align*}
    r(q) := \begin{cases}
        1\quad &\text{if}\: q\in Q_{\rm{tgt}}\\
        0 \quad &\text{otherwise}
    \end{cases}
\end{align*}
is the reward function, $\beta\in(0,1]$ is the discount factor and $\widetilde\Gamma$ is given by
\begin{align*}
    \widetilde\Gamma_{q,a} := \begin{cases}
        \{\delta_{q_u}\} \quad &\text{if} \: q\in Q_{\rm{tgt}}\\
        \Gamma_{q,a} \quad &\text{otherwise}
    \end{cases}
\end{align*}
for all $q\in Q$, $a\in A$. With a small abuse of notation we let $\Sigma$ and $\Xi$ denote, respectively, the sets of strategies and adversaries of $\widetilde{\mathcal M}$. We define the \emph{value function} of a strategy $\sigma\in\Sigma$ and adversary $\xi\in\Xi$ as the \emph{expected total reward} \cite{puterman2014markov}
\begin{align*}
    V_{\sigma,\xi,\beta}(q) := \mathbb{E}_{P_\xi^{q,\policy}}\Big[\sum_{k=0}^\infty \beta^k r(\omega(k))\Big]
\end{align*}
of $\widetilde{\mathcal{M}}$. Note that
$V_{\sigma,\xi,1}(q) = P_{\rm{reach}}( Q_{\rm safe},Q_{\rm{tgt}},\infty \mid q, \policy,\xi)$,
which is the desired reachability probability of $\mathcal{M}$ under $\sigma$ and $\xi$. By \eqref{eq:optimal_strategy_robust_MDPs} and \eqref{eq:def_bounds_reachability_robust_MDPs_lower}, this implies 
\begin{align} \label{value:function:vs:worst:case:reachability}
\sup_{\sigma\in\Sigma}\inf_{\xi\in\Xi}V_{\sigma,\xi,1}(q) = \underline p^\infty(q)
\end{align} 
for all $q\in Q$. We next establish that $V_{\sigma,\xi,\beta}$ is continuous at $\beta = 1$, which generalizes the corresponding well-known result for MDPs~\cite[Lemma 7.18]{puterman2014markov}.
\begin{lemma}
\label{lemma:continuity_beta}
    Given a strategy $\sigma\in\Sigma$ and an adversary $\xi\in\Xi$ of $\widetilde{\mathcal{M}}$, the value function $\beta\mapsto V_{\sigma,\xi,\beta}$ is continuous on $(0,1]$.
\end{lemma}
\begin{proof}
For each $q\in Q$, $V_{\sigma,\xi,\beta}(q)$ is the infinite sum of the functions
\begin{align*}
    g_k(\beta) := \beta^k\mathbb{E}_{P_\xi^{q,\policy}}[r(\omega(k))], \quad k\in\naturals_0,
\end{align*}
which are uniformly bounded on $\beta\in(0,1]$ by $\lambda_k := \mathbb{E}_{P_\xi^{q,\policy}}[r(\omega(k))]$, $k\in\naturals_0$, respectively. Since $\sum_{k=0}^\infty \lambda_k = V_{\sigma,\xi,1}(q)$, which is finite, as it is the probability of reaching $Q_{\rm tgt}$ from $q$, $\sum_{k=0}^\infty g_k$ converges uniformly to $V_{\sigma,\xi, \beta}(q)$ by the Weierstrass M-test. Thus, since each $g_k$ is continuous with respect to $\beta$ on $(0,1]$, the same holds also for $V_{\sigma,\xi,\beta}$.
\end{proof}

We identify the space of value functions $V:Q\to\mathbb R_{\ge 0}$ with $\mathcal{V} := \mathbb{R}_{\ge 0}^{|Q|}$ and write $V \ge (\le,>,<) V'$ for $V,V'\in\mathcal{V}$ to indicate that the respective inequality holds component-wise. The limits of sequences $\{V^k\}$ in $\mathcal{V}$ are also interpreted component-wise.
Let $\Xi_s$ be the set of stationary adversaries of $\widetilde{\mathcal{M}}$. Namely, for every $q\in Q$ and $a\in A$, an adversary $\xi\in\Xi_s$ chooses a distribution $\gamma^{\xi}_{q,a} \in \widetilde\Gamma_{q,a}$. Given $\sigma\in\Sigma_s$, $\xi\in\Xi_s$, and $\beta\in(0,1]$, we define the Bellman operator $T_{\sigma,\xi,\beta}:\mathcal{V} \longrightarrow \mathcal{V}$ of the auxiliary MDP $\widetilde{\mathcal{M}}$ by
%
\begin{align*}
    T_{\sigma,\xi,\beta}[V](q) := r(q) +\beta\sum_{q'\in Q}\gamma_{q,\sigma(q)}^{\xi}(q')V(q') 
\end{align*}
for any $V\in\mathcal{V}$, $q\in Q$. We also define, given $\sigma\in\Sigma_s$, $T_{\sigma, \beta}:\mathcal{V} \longrightarrow \mathcal{V}$ with
\begin{align*}
    T_{\sigma,\beta}[V](q) := r(q) + \min\limits_{\gamma\in\widetilde{\Gamma}_{q,\sigma(q)}}\beta\sum_{q'\in Q}\gamma(q')V(q')
\end{align*}
for any $V\in\mathcal{V}$, $q\in Q$. Finally, we define $T_{\beta}:\mathcal{V} \longrightarrow \mathcal{V}$ with
\begin{align*}
    T_\beta[V](q) := r(q) + \max\limits_{a\in A}\min\limits_{\gamma\in\widetilde{\Gamma}_{q,a}}\beta\sum_{q'\in Q}\gamma(q')V(q')
\end{align*}
for any $V\in\mathcal{V}$, $q\in Q$.

%
%
When $\beta = 1$, we simply write $T := T_1$, $T_{\sigma} : = T_{\sigma,1}$ and $T_{\sigma,\xi} := T_{\sigma,\xi,1}$, respectively. Note that the sequence obtained by iterating on $T$ with initial condition $0$ is equivalent to performing robust dynamic programming as in Theorem~\eqref{thm:robust_dynamic_programming_infinite_horizon_lower_bound}. We can now proceed with the proofs of the results from  Section~\ref{sec:unbounded_horizon}.

\begin{proof}[Proof of Theorem~\ref{thm:robust_dynamic_programming_infinite_horizon_lower_bound}]
Consider the sequence of value functions $\{V^k\}_{k\in \mathbb{N}_0}$, defined recursively by $V^{k+1} = T[V^k]$ and $V^0 = 0$. 
First, we prove that this sequence converges to a fixed point of $T$ in $\mathcal{V}$. It is rather straightforward to show that the operator $T$ is monotone, namely, that $T[V]\ge T[V']$ holds for all $V,V'\in\mathcal{V}$ with $V\ge V'$, which implies that $V^{k+1} \ge V^k$ for all $k\in\mathbb{N}_0$. One can also readily check that $\{V^k\}_{k\in \mathbb{N}_0}$ is bounded. Therefore, since $Q$ is finite, the sequence converges uniformly to $V^\infty \in\mathcal V$. 
In addition, by following the exact same arguments as in \cite[Proof of Theorem 3.2a]{iyengar2005robust}, which establishes that discounted Bellman operators of robust MDPs are contractions, we deduce that $T$ has Lipschitz modulus one and is therefore continuous. Thus, it follows that
\begin{align*}
    V^\infty := \lim_{k\to\infty} V^{k+1} =  \lim_{k\to\infty}T[V^k] = T[\lim_{k\to\infty}V^k] = T[V^\infty],
\end{align*}
namely, $V^\infty$ is a fixed point of $T$. In order to show minimality of $V^\infty$, let $V^*\in \mathcal{V}$ be an arbitrary fixed point of $T$. By monotonicity of $T$ we get that $T[0] \le T[V^*]$, and thus $V^1 \le V^*$. By induction, we obtain that $V^k \le V^*$ for all $k\in\mathbb{N}_0$, so $V^\infty \le V^*$. Since $V^*$ was chosen arbitrarily, $V^\infty$ must be the least fixed point of $T$ in $\mathcal{V}$.

Next, we follow the same strategy as in the proof of \cite[Theorem 3]{blackwell1967positive} to show that $V^\infty=\underline p^\infty$. 
To this end, note that $T_\beta[V] \le T[V]$ for each $V\in\mathcal{V}$ and $\beta\in(0,1)$. Denoting by $T_\beta^k$ the $k$-fold composition of $T_\beta$ with itself, we obtain by induction  that $T_\beta^k[0] \le T^k[0]$ for all $k\in\mathbb{N}_0$. Thus
\begin{align}
    \label{eq:aux1_app_B}
    \lim_{k\to\infty} T_\beta^k[0] \le \lim_{k\to\infty} T^k[0] = V^\infty.
\end{align}
From the theory of robust MDPs \cite[Theoren 3.2a]{iyengar2005robust}, the operator $T_{\beta}$ is a contraction mapping on $\mathcal{V}$ for all $\beta\in(0,1)$. This implies that $T^{\beta}$ has a  unique fixed point $V_\beta^{\infty}$, which by  \cite[Theoren 3.2b]{iyengar2005robust} satisfies \begin{align} \label{fixed:point}
V_\beta^{\infty}(q)= \sup_{\sigma\in\Sigma}\inf_{\xi\in\Xi}V_{\sigma,\xi,\beta}(q)
\end{align}
for all $q\in Q$, and $\{T_\beta^k[0]\}_{k\in\mathbb{N}_0}$ necessarily converges to it. Thus, we deduce from \eqref{eq:aux1_app_B} that $\sup_{\sigma\in\Sigma}\inf_{\xi\in\Xi}V_{\sigma,\xi,\beta}(q) \le V^\infty(q)$ for all $q\in Q$ and $\beta\in(0,1)$. 
From this fact and Lemma~\ref{lemma:continuity_beta} we get  
\begin{align*}
    \lim_{\beta\to 1}\sup_{\sigma\in\Sigma}\inf_{\xi\in\Xi}V_{\sigma,\xi,\beta}(q) = \sup_{\sigma\in\Sigma}\inf_{\xi\in\Xi} \lim_{\beta\to 1} V_{\sigma,\xi,\beta}(q) = \sup_{\sigma\in\Sigma}\inf_{\xi\in\Xi}V_{\sigma,\xi}(q)\le V^\infty(q)
\end{align*}
for all $q\in Q$. Hence, from the last inequality and \eqref{value:function:vs:worst:case:reachability} we obtain  $\underline p^\infty\le V^\infty$. 

To also prove that $\underline p^\infty \ge V^\infty$, note that $V_{\sigma,\xi,\beta} \le V_{\sigma,\xi}$ for all $\sigma\in\Sigma$, $\xi\in \Xi$, and $\beta\in(0,1)$. Thus, we get again from \eqref{value:function:vs:worst:case:reachability}  that  $\sup_{\sigma\in\Sigma}\inf_{\xi\in\Xi}V_{\sigma,\xi,\beta}(q) \le \underline p^\infty(q)$ for all $q\in Q$ and $\beta\in(0,1)$. Combining this with \eqref{fixed:point} and the fact that each $T_\beta$ is also monotone, which implies that the sequence $\{T_\beta^k[0]\}_{k\in\mathbb{N}_0}$ is upper bounded by its fixed point $V_\beta^\infty$, we get that 
%
$T_\beta^k[0](q) 
    \le \underline p^\infty(q)
$ 
for all $q\in Q$, $k\in\mathbb{N}_0$, and $\beta\in(0,1]$. Thus, by continuity of each  $T_\beta^k$ with respect to  $\beta$, which follows directly from continuity of $T_\beta$, we get 
\begin{align*}
V^k=T^{k}[0] = \lim_{\beta\to 1}T_\beta^k[0] \le \underline p^\infty    
\end{align*}
for all $k\in\mathbb{N}_0$, and taking the limit as $k\to\infty$ we obtain $V^\infty \le \underline p^\infty$. Therefore, $V^\infty = \underline p^\infty$, which concludes the proof. 
\end{proof}


\begin{proof}[Proof of Proposition~\ref{prop:existence_stationary optimal_strategy_unbounded_horizon}]

Consider first the discounted case with discount factor $\beta\in(0,1)$ and denote, for all $q\in Q$, $\sigma\in\Sigma$, $V_{\sigma,\beta}(q) := \inf_{\xi\in\Xi} V_{\sigma,\xi,\beta}(q)$ and $V_{\beta}(q) := \sup_{\sigma'\in\Sigma} V_{\sigma',\beta}(q)$.
By \cite[Theorem 4]{nilim2005robust}, for every discount factor $\beta \in (0,1)$, there exists a stationary strategy $\sigma'\in\Sigma_s$ that is optimal, i.e., $V_{\sigma',\beta} = V_{\beta}$. Let $\{\beta_m\}_{m\in\naturals_0}$ be a non-decreasing sequence that converges to $1$. As in the proof of \cite[Theorem 7.1.9]{puterman2014markov}, since $\Sigma_s$ is finite, there exist a strategy $\sigma'\in\Sigma_s$ and a subsequence $\{\beta_{m_i}\}_{i\in\naturals_0}$ such that
\begin{align*}
    V_{\sigma',\beta_{m_i}} = V_{\beta_{m_i}}, \:\: \forall i \in\naturals_0.
\end{align*}
Since both $\beta\mapsto V_{\sigma',\beta}$ and $\beta\mapsto V_{\beta}$ are continuous with respect to $\beta\in (0,1]$  by Lemma~\ref{lemma:continuity_beta}, it follows that for every $\sigma\in\Sigma$ and $q\in Q$,
\begin{align*}
    V_{\sigma}(q) &= \lim_{\beta\to 1} V_{\sigma,\beta}(q) = \lim_{i\to \infty} V_{\sigma,\beta_{m_i}}(q)\\
    &\le \lim_{i\to \infty} \max_{\sigma\in\Sigma} V_{\sigma,\beta_{m_i}}(q) = \lim_{i\to \infty} V_{\sigma',\beta_{m_i}}(q) = V_{\sigma'}(q).
\end{align*}
Thus, 
$\sigma'\in\Sigma_s$ is our desired optimal strategy 
for the undiscounted case, i.e., $V_{\sigma'} = \underline p^\infty$.
\end{proof}

\begin{proof}[Proof of Proposition~\ref{prop:sufficient_conditions_optimal_strategy_unbounded_horizon}]

First, we show that any $\sigma'\in\Sigma_s$ that is optimal satisfies the condition \eqref{set:A:star}. Denote $V_{\sigma'}(q) := \inf_{\xi\in\Xi} V_{\sigma',\xi}(q)$ for all $q\in Q$. From optimality of $\sigma'$ we get that $V_{\sigma'} = \underline p^\infty$.
Additionally, since $\underline p^\infty$ and $V_{\sigma'}$ are respective fixed points of $T$ and $T_{\sigma'}$, it follows that
\begin{align*}
    \underline p^\infty(q) = T_{\sigma'}[\underline p^\infty](q) \le \max_{\sigma\in\Sigma_s} T_{\sigma}[\underline p^\infty](q) = T[\underline p^\infty](q) = \underline p^\infty(q),
\end{align*}
%
%
for all $q\in Q$, where the $\max$ should be interpreted in a component-wise fashion. This implies that
\begin{align*}
    \sigma'(q) \in \arg\max_{a\in A}\min_{\gamma\in\widetilde{\Gamma}_{q,a}}\sum_{q'\in Q}\gamma(q')\underline p^\infty(q')
\end{align*}
for all $q\in Q$. Therefore, we conclude that \eqref{set:A:star} is a necessary condition for optimality of stationary strategies.

Next, we prove that conditions \eqref{set:A:star} and \eqref{proper:strategy} are sufficient for optimality of stationary strategies. Let $\widetilde{\mathcal{M}}_{\sigma^*,\xi}$ be the time-varying Markov chain obtained by fixing the strategy and adversary of $\widetilde{\mathcal{M}}$ to $\sigma^*$ and some $\xi\in\Xi$, respectively. Denote also by  $\widetilde M_\xi^k \in \reals^{|Q|\times|Q|}$ the transition matrix of $\widetilde{\mathcal{M}}_{\sigma^*,\xi}$ at time step $k$ so that the $i$-th row of $\widetilde M_\xi^k$ is a probability distribution from $\widetilde\Gamma_{q_i,\sigma^*(q_i)}$. 
Note that by the definition of the Bellman operators $T$ and $T_{\sigma,\xi}$ at the beginning of this section, and the fact that $\underline p^\infty$ is a fixed point of $T$ we have 
\begin{align*} 
 T[\underline p^\infty]= \max_{\sigma\in\Sigma_s}\min_{\xi\in \Xi_s} T_{\sigma,\xi}[\underline p^\infty]
\end{align*}
where the $\max$ and $\min$ are interpreted in a component-wise fashion. Thus, since $\sigma^*$ satisfies condition \eqref{set:A:star},
\begin{align*}
    \underline p^\infty = T[\underline p^\infty] = \max_{\sigma\in\Sigma_s}\min_{\xi\in\Xi_s} T_{\sigma,\xi}[\underline p^\infty] = \min_{\xi\in\Xi_s} T_{\sigma^*,\xi}[\underline p^\infty] \le r + \widetilde M_\xi^1 \underline p^\infty. 
\end{align*}
By applying inductively the same argument, we get 
\begin{align*}
 \underline p^\infty \le r +  \bigg( \sum_{k = 1}^{K-1} \widetilde M_\xi^{k} \bigg) r + \bigg( \prod_{k = 1}^K \widetilde M_\xi^k\bigg)\underline p^\infty  \quad \forall K\in\naturals, \xi\in\Xi.
\end{align*}
Since the $i$-th entry in the last term of this bound 
equals 
\begin{align*}
    \sum_{q'\in Q}P_\xi^{q_i,\policy^*}[\pathmdp(K) \in q']\underline p^\infty(q') = \sum_{q'\in Q_{\rm reach}}P_\xi^{q_i,\policy^*}[\pathmdp(K) \in q']\underline p^\infty(q'),
\end{align*}
%
the last term converges to $0$ as $K\to\infty$ by the definition of $Q_{\rm reach}$ and condition \eqref{proper:strategy}. Furthermore, the sum of the first two terms converges to $V_{\sigma^*, \xi}$. As a result, $\underline p^\infty(q) \le \inf_{\xi\in\Xi} V_{\sigma^*, \xi}(q)$ for all $q\in Q$.
On the other hand, we have from \eqref{value:function:vs:worst:case:reachability} that 
$\underline p^\infty(q) \ge \inf_{\xi\in\Xi} V_{\sigma^*, \xi}(q)$ and we conclude that $\underline p^\infty(q) = \inf_{\xi\in\Xi} V_{\sigma^*, \xi}(q)$, 
namely, that $\sigma^*$ is optimal. 
\end{proof}

\begin{proof}[Proof of Theorem~\ref{thm:optimal_strategy_unbounded_horizon}]

We begin by proving that $\sigma^*$ is well-defined. This is the case if $\sigma^*$ is defined for all states, which is equivalent to having $Q^{m_{\max}} = Q_{\rm reach}$. Let $\sigma'\in\Sigma_s$ be a stationary optimal strategy, namely, that  $V_{\sigma'} = \underline p^\infty$ holds. Therefore, the sequence $\{V^k\}_{k\in\naturals_0}$ obtained via robust dynamic programming as in \eqref{eq:robust_value_iteration_lower_bound} with fixed strategy $\sigma'$ converges monotonically to $\underline p^\infty$. Now we start an induction argument by assuming that for $k\in\naturals_0$, $V^{k}(q) > 0$ implies $q\in Q^{k}$, for $q\in Q$, which holds for $k = 0$. If the previous condition holds at iteration $k$, then for all $q\in Q\setminus Q^k$ we obtain that
\begin{align*}
    V^{k+1}(q) &= \min_{\gamma\in\widetilde{\Gamma}_{q,\sigma'(q)}}\sum_{q'\in Q}\gamma(q')V^k(q') = \min_{\gamma\in\widetilde{\Gamma}_{q,\sigma'(q)}}\sum_{q'\in Q^{k}}\gamma(q')V^k(q')\\
    &\le \max_{a\in A^*(q)}\min_{\gamma\in\widetilde{\Gamma}_{q,a}}\sum_{q'\in Q^k}\gamma(q')V^k(q') \le \max_{a\in A^*(q)}\min_{\gamma\in\widetilde{\Gamma}_{q,a}}\sum_{q'\in Q^k}\gamma(q'),
\end{align*}
where the first inequality follows from the last statement of Proposition~\ref{prop:sufficient_conditions_optimal_strategy_unbounded_horizon}.
This result, together with monotonicity of $\{V^k\}_{k\in\naturals_0}$, guarantees that at iteration $k+1$, $V^{k+1}(q) > 0$ implies $q\in Q^{k+1}$ for all $q\in Q$. Thus the induction argument holds and we get that $\lim_{k\to\infty}Q^k = Q^{m_{\max}} = Q_{\rm reach}$.

Next, to prove optimality of $\sigma^*$, it suffices to show that $\sigma^*$ satisfies the two conditions of Proposition~\ref{prop:sufficient_conditions_optimal_strategy_unbounded_horizon}. Condition~\eqref{set:A:star} is satisfied directly by the construction of $\sigma^*$.
In the following, we show that $\sigma^*$ also satisfies condition~\eqref{proper:strategy}. The proof follows closely that of \cite[Theorem 10.25]{baier2008principles}. Let $\widetilde{\mathcal{M}}_{\sigma^*,\xi}$ denote the time-varying Markov Chain corresponding to our strategy $\sigma^*$ and an arbitrary adversary $\xi\in \Xi$. To show that the paths starting at any initial condition eventually exit $Q_{\rm{reach}}$ and remain outside forever with probability one, we equivalently show that its complement, namely, that $Q_{\rm{reach}}$ is visited infinitely often, has probability zero. Note that by definition of $\sigma^*$ via backward reachability, all states $q \in Q_{\rm{reach}}$ have nonzero probability of eventually reaching $Q_{\rm{tgt}}$ and thereafter $q_u$ under $\sigma^*$ and $\xi$. This implies that, for each $q \in Q_{\rm{reach}}$, there exists a path fragment that reaches $q_u$ and which has positive probability since it is a finite fragment. Let $p > 0$ be a uniform lower bound on this probability across all states $q \in Q_{\rm{reach}}$. 
Then, the probability of the event ``$Q_{\rm{reach}}$ is visited at least $k$ times, but the corresponding path fragment that reaches $q_u$ is never taken'' 
is upper bounded by $(1-p)^k$. 
Taking the limit as $k\to\infty$, we obtain that the probability that $Q_{\rm{reach}}$ is visited infinitely often is upper bounded by $\lim_{k\to\infty}(1-p)^k = 0$, meaning that all paths eventually reach $Q\setminus Q_{\rm{reach}}$ and remain there forever with probability one.
Since the choice of the adversary was arbitrary, we conclude that condition~\ref{proper:strategy} is also fulfilled, which completes the proof.
\end{proof}
\section{Proof of Theorem~\ref{thm:dual_problem}}
\label{app:proof_dual}

To facilitate the readability of the proof, we use the shorthand notation $p_i\equiv\underline p^k(q_i)$, $\underline P_j \equiv\underline P(q,a,q_j)$, $\overline P_j \equiv\overline P(q,a,q_j)$, $\overline{\mathcal N}\equiv\overline{\mathcal N}^{q,a}$ and $\overline{\mathcal N}_W\equiv\overline{\mathcal N}_W^{q,a}$.

\begin{proof}[Proof of Theorem~\ref{thm:dual_problem}]
Concavity of $G$ follows by standard arguments from convex analysis  \cite{boyd2004convex}. In particular, the functions $h_j$ are concave as the pointwise minimum of affine functions and $\min\{\underline P_j(\lambda + h_j(\mu)),\overline P_j(\lambda + h_j(\mu))\}$ is in turn concave as the minimum of concave functions. Thus, $G$ is also concave as the sum of concave functions.

To obtain~\eqref{eq:dual_hscc}, we first eliminate $\gamma$ and $\widehat \gamma$ from \eqref{eq:LP_cost}-\eqref{eq:feasible set_robust_MDP_cost} and get the equivalent problem  
\begin{equation*}
\begin{aligned}
\min_{\pi_{ij}} \quad \sum_{i \in \overline{\mathcal N}_W,j \in\overline{\mathcal N}}\pi_{ij} & p_i\\
\textrm{s.t.} \quad \sum_{i \in \overline{\mathcal N}_W}\pi_{ij} & \ge \underline P_j  & \forall j\in\overline{\mathcal N} \\
 \sum_{i \in \overline{\mathcal N}_W}- \pi_{ij} & \ge - \overline P_j & \forall j\in\overline{\mathcal N} \\
  \sum_{i \in \overline{\mathcal N}_W,j\in\overline{\mathcal N}}-c(q_i,q_j)\pi_{ij} & \ge - \epsilon \\
  \sum_{i \in \overline{\mathcal N}_W,j\in\overline{\mathcal N}}\pi_{ij} & = 1 \\
 \pi_{ij} & \ge 0 & \forall i\in \overline{\mathcal N}_W,j\in\overline{\mathcal N}
\end{aligned}
\end{equation*}
By denoting $\alpha:=(\alpha_1,\ldots,\alpha_N)$, $\beta:=(\beta_1,\ldots,\beta_N)$, $\mu$, and $\lambda$ the corresponding dual variables of all but the last set of constraints, we obtain as in \cite[Page 142]{bertsimas1997introduction} the dual linear problem 
\begin{align*}
\max_{\alpha,\beta,\mu\ge 0,\lambda} & \sum_{j \in\overline{\mathcal N}}  \big(\alpha_j\underline P_j-\beta_j\overline P_j\big)-\mu\epsilon+\lambda \\
\textrm{s.t.} \quad & \alpha_j-\beta_j-\mu c(q_i,q_j)+\lambda \le p_i \qquad  \forall i\in \overline{\mathcal N}_W,j\in\overline{\mathcal N}.
\end{align*}
Since the feasible set of the primal problem is nonempty and compact, strong duality holds (cf. \cite[Theorem 4.4]{bertsimas1997introduction}). The constraints of the dual problem are equivalently written 
\begin{align*}
     - \alpha_j + \beta_j + \min_{i\in \overline{\mathcal N}_W} \{p_i + \mu c(q_i,q_j)\} - \lambda \ge 0 \qquad  \forall j\in\overline{\mathcal N} 
\end{align*}
and since $\beta_j\ge 0$, they can be cast in the form 
\begin{align*}
     \beta_j \ge \max\{0, \alpha_j - \min_{i\in \overline{\mathcal N}_W} \{p_i + \mu c(q_i,q_j)\}+ \lambda\} \qquad  \forall j\in\overline{\mathcal N}.
\end{align*}
Taking into account that the objective function of the dual problem is decreasing in each $\beta_j$, its optimal value is the same as that of the problem   
\begin{align}
    & \max_{\alpha,\mu\ge 0,\lambda}  \sum_{j \in\overline{\mathcal N}} \big(\alpha_j\underline P_j - \max\{0, \alpha_j- \min_{i\in\overline{\mathcal N}_W} \{p_i + \mu c(q_i,q_j)\}+ \lambda\}\overline P_j \big)-\mu\epsilon+\lambda \nonumber \\
    & \qquad = \max_{\alpha,\mu\ge 0,\lambda}\sum_{j \in\overline{\mathcal N}}\big(\alpha_j\underline P_j + \min\{0, -\alpha_j+h_j(\mu) - \lambda\}\overline P_j \big)-\mu\epsilon+\lambda \nonumber \\
    & \qquad = \max_{\alpha,\mu\ge 0,\lambda}\sum_{j \in\overline{\mathcal N}}\min\{\alpha_j\underline P_j, \alpha_j(\underline P_j-\overline P_j) + \overline P_j(h_j(\mu) - \lambda)\}-\mu\epsilon+\lambda, \label{eq:eliminate_alpha} 
\end{align}
where we used \eqref{eq:dual_hscc1} in the first equality.
Next, since each $\overline P_j>0$, the corresponding function  
\begin{align*}
\alpha_j\mapsto \varphi_j(\alpha_j):=\min\{\alpha_j\underline P_j, \alpha_j(\underline P_j-\overline P_j) + \overline P_j(h_j(\mu)- \lambda)\}
\end{align*}
consists of two affine branches that intersect at $h_j(\mu)- \lambda$. In particular, over the left segment $\alpha_j\mapsto \alpha_j\underline P_j$ where $\alpha_j\le h_j(\mu)-\lambda$, the slope is non-negative, and it is non-positive over the right segment $\alpha_j\mapsto \alpha_j(\underline P_j-\overline P_j) + \overline P_j(h_j(\mu)- \lambda)$ where $\alpha_j\ge h_j(\mu)-\lambda$. Thus, the maximum of the function $\varphi_j$ is attained at $h_j(\mu)-\lambda$ when $h_j(\mu)-\lambda\ge 0$ and equals $(h_j(\mu)-\lambda)\overline P_j$, and at $0$ when $h_j(\mu)-\lambda<0$ where it is equal to $(h_j(\mu)-\lambda)\underline P_j$. Namely, we have
\begin{align*}
\max_{\alpha_j\ge 0}\varphi_j(\alpha_j)&=\begin{cases}
(h_j(\mu)-\lambda)\overline P_j &{\rm if}\;h_j(\mu)-\lambda\ge 0 \\
(h_j(\mu)-\lambda)\underline P_j & {\rm if}\;h_j(\mu)-\lambda<0
\end{cases}
\\
&=\min\{(h_j(\mu)-\lambda)\overline P_j,(h_j(\mu)-\lambda)\underline P_j\},
\end{align*}
which by virtue of \eqref{eq:eliminate_alpha} implies the dual reformulation~\eqref{eq:dual_hscc}.

To prove the last assertion of the theorem, fix an arbitrary value of $\mu\ge 0$ and consider the function $\lambda\mapsto G_\mu(\lambda)\equiv G(\lambda,\mu)$ with $G$ as given in \eqref{eq:dual_hscc}. Notice that each $\min$ term in the expression of $G_\mu$ is a piecewise affine function of $\lambda$ with two segments and a breakpoint at $h_j(\mu)$ (assuming without loss of generality that $\underline P_j \neq \overline P_j$). Since $G_\mu$ is the sum of piecewise affine functions, it is also piecewise affine and its breakpoints are included in the breakpoints of its components, i.e., in the set $\{h_j(\mu)\}_{j\in\overline{\mathcal N}}$. Taking further into account that $G_\mu$ is upper bounded by the optimal value of problem \eqref{eq:dual_hscc}-\eqref{eq:dual_hscc:fnc}, which is finite by strong duality, its maximum is necessarily attained at some of its breakpoints. This concludes the proof.
\end{proof}

\bibliographystyle{IEEEtran}\bibliography{ref2}

\end{document}